\documentclass{article}
\usepackage[utf8]{inputenc}
\usepackage[left=2cm,top=2.5cm,right=2cm,nohead]{geometry}
\usepackage{color}
\usepackage{graphicx}
\usepackage{amsmath,amssymb,amsthm}
\usepackage{newtxtext,newtxmath}
\usepackage{hyperref}
\hypersetup{colorlinks=false,
            citebordercolor=1 1 1,
            linkbordercolor=1 1 1}
\usepackage{cite}
\usepackage{Qcircuit}
\usepackage{authblk}

\newtheorem{theorem}{Theorem}[section]
\newtheorem{lemma}[theorem]{Lemma}

\newcommand{\sfrac}[2]{\ensuremath{{\textstyle{\frac{#1}{#2}}}}}
\newcommand{\half}{\sfrac{1}{2}}
\newcommand{\bra}[1]{\ensuremath{\langle{#1}\vert}}
\newcommand{\ket}[1]{\ensuremath{\vert{#1}\rangle}}

\newcommand{\be}{\begin{equation}}
\newcommand{\ee}{\end{equation}}
\newcommand{\ba}{\begin{array}}
\newcommand{\ea}{\end{array}}

\newcommand{\calC}{{\cal C }}

\newcommand{\calP}{{\cal P }}

\newcommand{\Gate}[1]{\textsc{#1}}

\newcommand{\cnotgate}{\Gate{cnot}}

\newcommand{\swapgate}{\Gate{swap}}

\newcommand{\ibmwashington}{{\textit{ibm\_washington}}}
\newcommand{\ibmithaca}{{\textit{ibm\_ithaca}}}

\title{Single-shot error mitigation by coherent Pauli checks}

\author{Ewout van den Berg, Sergey Bravyi, Jay M. Gambetta, Petar Jurcevic,\\[2pt] Dmitri Maslov, and Kristan Temme}

\affil{IBM Quantum, IBM T.J. Watson Research Center, Yorktown Heights, NY 10598, USA}

\date{\today}

\begin{document}

\maketitle

\begin{abstract}
Generating samples from the output distribution of a quantum circuit is
a ubiquitous task used as a building block of many quantum algorithms. 
Here we show how to accomplish this task
on a noisy quantum processor lacking full-blown error correction
for a special class of quantum circuits dominated by Clifford gates.
Our approach is based on Coherent Pauli Checks (CPCs)
that  detect errors in a Clifford 
circuit  by verifying commutation rules between 
random Pauli-type check operators and the considered circuit.  
Our main contributions are as follows.
First, we derive 
a simple formula
for the probability that a Clifford circuit
protected by CPCs contains a logical error.
In the limit of a large number of checks, the logical error probability is shown to approach the value ${\approx} 7\epsilon n/5$, where $n$ is the number of qubits and 
$\epsilon$ is the depolarizing error rate. 
Our formula agrees nearly perfectly with the numerical simulation results.
Second, we show that CPCs are well-suited for quantum processors with a limited
qubit connectivity. For example, the difference between all-to-all and linear qubit
connectivity is only a 3$\times$ increase in the number of {\sc{cnot}} gates required to implement CPCs.
Third, we describe simplified one-sided CPCs which are well-suited for mitigating
measurement errors in the single-shot settings.
Finally, we report an experimental demonstration of CPCs with up to 10 logical qubits and 
more than 100 logical {\sc{cnot}} gates. 
Our experimental results show that CPCs provide a marked improvement in the
logical error probability 
for the considered task of sampling the output distribution of quantum circuits. 
\end{abstract}

\section{Introduction}

Quantum error mitigation (QEM)  is a versatile set of tools for 
improving reliability of quantum circuits executed on noisy hardware~\cite{temme2017error,li2017efficient,kandala2019error}.
QEM supplements more traditional approaches to quantum fault-tolerance 
based on error-correcting codes. It is well-suited for quantum processors available today 
that do not yet meet stringent gate fidelity requirements
of full-blown quantum error correction. Most of the known  QEM methods combat noise by measuring a redundant set of data
generated by a suitable ensemble of noisy quantum circuits. 
Classical post-processing is then applied to the measured data
to filter out the contribution of noise and predict the outcome that would be 
observed in the absence of noise.
A comprehensive review of modern QEM protocols can be found in~\cite{endo2021hybrid}.

In contrast to quantum error correction, QEM introduces only a minor (if any) overhead in terms of
ancillary qubits and circuit depth while obviating the need to compile a circuit using a fault-tolerant gate set
such as the Clifford+T library, and tolerates error rates above the threshold of the known quantum codes.
However, QEM has two major limitations. First, the error mitigation overhead, as measured by the
number of circuit repetitions, scales exponentially with the circuit size, limiting the scope
of QEM to relatively shallow circuits.  While this overhead appears unavoidable, 
the exponential scaling becomes very mild in the regime of small error rates enabling
QEM demonstrations for medium-size circuits with
20 or more qubits and 1000 or more gates~\cite{kim2021scalable,vandenBerg2022pec}. 
Perhaps more importantly, the scope of almost all known QEM protocols 
is severely limited in terms of how
the output of a quantum circuit can be accessed. Namely, these protocols
apply only to quantum algorithms with an {\em expected value readout}.
Such algorithms can use the output state of a quantum circuit only to measure
the expected value of some observable such as a Pauli operator,
a Hamiltonian composed of several Pauli terms, or a projector onto some basis state.
Notable  examples of quantum algorithms
with the expected value readout are variational quantum simulators~\cite{peruzzo2014variational,kandala2017hardware,motta2020determining} and supervised learning with quantum kernels~\cite{schuld2019quantum,havlivcek2019supervised}.
However, unlocking the full computational power of quantum algorithms may require
a {\em single-shot readout} -- the ability to generate samples from the probability distribution
describing the output of a quantum circuit.
Thus, if $\psi$ denotes the output state of a quantum circuit, we would like to sample a bit string $x$
from the probability distribution $|\langle x|\psi\rangle|^2$.
For example, simulation algorithms based
on the quantum phase estimation~\cite{abrams1999quantum}, 
Shor's factoring algorithm~\cite{shor1999polynomial},
Grover's search~\cite{grover1997quantum}, 
quantum approximate optimization algorithm~\cite{farhi2014quantum} (QAOA),
quantum volume~\cite{jurcevic2021demonstration}, random circuit sampling~\cite{aaronson2016complexity}, state learning algorithms based on 
classical shadows~\cite{huang2020predicting},
and quantum-enhanced Markov Chain Monte Carlo algorithms~\cite{layden2022quantum}
all require single-shot readout. 
Moreover, certain families of quantum circuits with the expected value readout 
can be efficiently simulated on a classical computer (in time polynomial or quasi-polynomial
in the number of qubits), whereas their counterparts with the single-shot readout are believed to be hard for classical simulators.
This is the case, for example, for geometrically local constant-depth circuits on a
finite-dimensional grid of qubits~\cite{bravyi2021classical,coble2022quasi,dontha2022approximating},
instantaneous quantum polynomial circuits~\cite{bremner2017achieving}, 
and QAOA circuits with a few entangling steps~\cite{farhi2016quantum,wang2018quantum,bravyi2020obstacles,bravyi2022hybrid}.
These results strongly suggest that single-shot readout can 
endow quantum circuits with extra computational power.
Thus the ability to do error mitigation for quantum circuits
with the single-shot readout is a highly desirable yet elusive goal.

In the present paper, we examine QEM protocols pioneered by Roffe et al.~\cite{roffe2018protecting}
and developed further by Debroy and Brown~\cite{debroy2020extended}.
The key building block of these protocols is a {\em coherent Pauli check}.
It enables single-shot error mitigation for arbitrary circuits composed of Clifford gates
as well as layers of Clifford gates embedded into a larger, possibly non-Clifford circuit.
A coherent Pauli check (CPC) detects errors 
by verifying commutation rules between Pauli and Clifford gates, as described in more detail in Section~\ref{sec:CPC}. A single CPC applied to a payload circuit with $n$ qubits
requires the overhead of only one ancillary qubit and at most $O(n)$ gates while eliminating roughly 
$50\%$ of the errors that may occur in the payload circuit.
Despite their promise, 
QEM protocols based on CPCs have received surprisingly little attention.
Recent works by Debroy and Brown~\cite{debroy2020extended}, and 
Gonzales et al.~\cite{gonzales2022quantum} examined the effectiveness of CPCs using numerical simulations.
Here, we propose a simple theoretical model
that can be used to predict the performance of QEM protocols with multiple CPCs
for a very large number of qubits. Our model takes into account errors that occur in the payload circuit as well as errors introduced by CPCs themselves. 
We observe a nearly perfect agreement
between the predictions of our model and numerical simulation results.
Next, we show how to enhance the performance of CPCs by augmenting them with flag qubits and how to efficiently implement QEM protocols with multiple
CPCs for the linear nearest neighbor (LNN) qubit connectivity. 
Finally, we report an experimental demonstration of error-mitigated quantum
circuits with CPCs and single-shot readout.

Let us briefly comment on the earlier work relevant for our study. The key ideas behind CPCs
are analogous to entanglement assisted quantum error correction 
proposed by Brun, Devetak, and Hsieh~\cite{brun2006correcting}. These authors explored
catalytic quantum codes described by Pauli check operators that do not obey the
standard pairwise commutativity condition.  It was observed that catalytic codes can
nevertheless be useful for quantum communication in the presence of entanglement shared
between the sender and the receiver. Moreover,
it was found that even a small amount of preexisting entanglement 
can enable reliable transmission of a large number of qubits. The authors of~\cite{brun2006correcting}
also commented that ``catalytic quantum codes open the possibility of application to error correction
in quantum computing where we can think of decoherence as a channel into the future."
This possibility was explored further by Chancellor, Roffe et al.~in ~\cite{chancellor2016coherent,roffe2018protecting,roffe2019coherent}
who introduced the notion of CPCs and used them as a  tool for constructing conventional
quantum error-correcting codes. Ref.~\cite{roffe2018protecting} also reported
the first experimental demonstration of CPCs. However, the main focus of~\cite{chancellor2016coherent,roffe2018protecting,roffe2019coherent} was on realizing a
fault-tolerant quantum memory (the identity payload circuit). 
A seminal work by Debroy and Brown~\cite{debroy2020extended} pioneered applications
of CPCs in the context of quantum computing and circuit verification.
Ref.\cite{debroy2020extended} developed strategies for optimizing Pauli checks
and numerically observed that fidelities of small Clifford and near-Clifford circuits can be significantly improved 
in the presence of CPCs. Most of the constructions used in the present paper were introduced in~\cite{debroy2020extended}.  More recently, Gonzales et al.~\cite{gonzales2022quantum}
analyzed the performance of CPCs in the presence of coherent (non-Pauli) errors
and proposed an efficient algorithm for finding Pauli checks compatible with 
a given payload circuit composed of Clifford gates and single-qubit (non-Clifford) $Z$-rotations. We note that CPCs can also be viewed as a partially fault-tolerant implementation of error correction with flag qubits introduced by Chao and Reichardt~\cite{chao2018quantum}.

\section{Coherent Pauli checks}
\label{sec:CPC}

In this section, we summarize the construction of CPCs proposed
in~\cite{roffe2018protecting, debroy2020extended,gonzales2022quantum}.
Let $\calP_n$ and $\calC_n$ be the groups of $n$-qubit Pauli
and Clifford operators respectively. By definition, 
any element of $\calP_n$ has a form 
$\omega Q_1\otimes Q_2 \otimes \cdots \otimes Q_n$,
where $Q_j \in \{I,X,Y,Z\}$ are single-qubit Pauli operators
and $\omega \in \{\pm 1,\pm i\}$ is a phase factor.
The Clifford group $\calC_n$ contains all $n$-qubit unitary
operators $U$ such that $U\calP_n U^\dag = \calP_n$.  

\subsection{Two-sided checks}

For any Clifford circuit $U\in \calC_n$ 
and any Pauli $L\in \calP_n$ with corresponding $R = ULU^{\dag} \in \calP_n$ it  holds that
\be
\label{LUR}
\Qcircuit @C=1em @R=.7em {
      & \gate{L} & \gate{U}  & \gate{R} & \qw
      & = & &  \gate{U} & \qw
}\;,
\ee
where the desired Pauli operator $R$ can be efficiently computed using the standard stabilizer formalism.
The circuit identity Eq.~\ref{LUR} holds even if $U$ is a part of a larger
quantum circuit that possibly contains non-Clifford gates. 
Suppose now that $U$ contains some faulty gates. For simplicity, we will consider
the depolarizing noise model such that a faulty gate is modeled by an
ideal gate followed by a Pauli error. Since $U$ contains only Clifford gates, any Pauli
error can be propagated to the beginning of $U$, which results in a noisy circuit
$\tilde{U}=UE$ for some Pauli error $E\in \calP_n$.
We conclude that a noisy version of the identity Eq.~\ref{LUR} is 
\be
\label{LURnoisy}
\Qcircuit @C=1em @R=.7em {
      & \gate{L} & \gate{\tilde{U}}  & \gate{R} & \qw & &
      & = (-1)^s & &
      & \gate{\tilde{U}} & \qw
}
\ee
where $s=0$ if the error $E$ commutes with $L$ and $s=1$ if $E$ anti-commutes with $L$. 
We will refer to $s$ as an {\em error syndrome}. By definition, $s=0$ with certainty
in the absence of errors. 
A CPC works by measuring the syndrome $s$ using one ancillary qubit and post-selecting on the outcome $s=0$. 
The simplest version of a syndrome measurement circuit is illustrated in Figure~\ref{fig:CPC1}
for $n=3$ qubits. We refer to the operators $L$, $U$, and $R$ as the left Pauli check, the payload circuit, and the right Pauli check respectively.

We can detect an error $E$ whenever it anti-commutes with the left Pauli check $L$. Therefore, if we know anything about the expected distribution of the errors $E$, we could choose $L$ such that the probability of catching an error is maximized (the corresponding right Pauli check $R$ is uniquely determined using Eq.~\ref{LUR}). When the error distribution is not known in advance, a good strategy is to pick the $L$ at random from the uniform distribution on the Pauli group $\calP_n$.
This choice of $L$ ensures that any non-identity error $E$ occurring in the payload circuit
anti-commutes with the check $L$ with probability $1/2$. If the errors are uniformly distributed, we can detect half of all the errors in the payload circuit using only a single Pauli check. 

\begin{figure}[t]
\centering
\includegraphics[height=3cm]{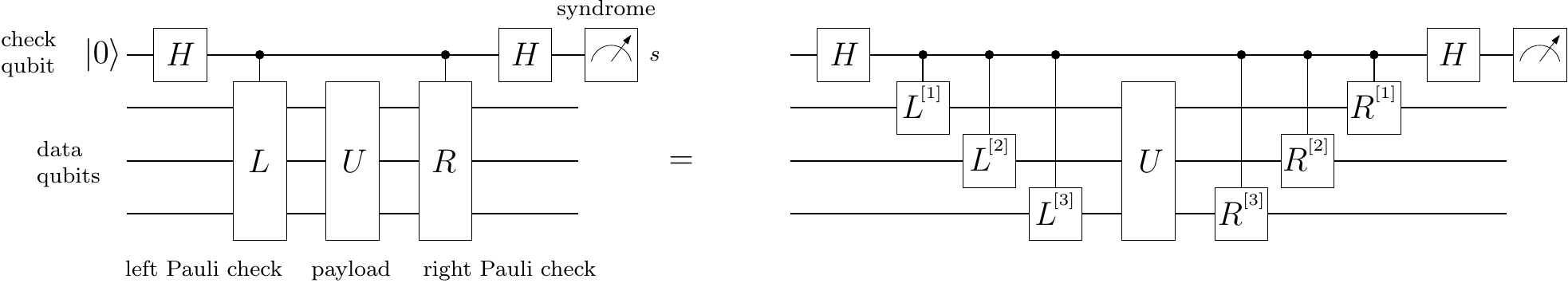}\\[2pt]
\caption{Example of a  coherent Pauli check with $n=3$ data qubits. Superscript indices in the right-side circuit denote the different components of the three-qubit Pauli terms L and R.}\label{fig:CPC1}
\end{figure}

To increase the fraction of detected errors, the CPC construction can be extended to multiple Pauli checks. The simplest version of this is illustrated in Figure~\ref{Fig:MultipleChecks} for $m=3$ checks. Alternatively, we could view the circuit formed by the previous $m-1$ checks as yet another Clifford circuit, and apply a check on it to obtain nested checks. This scheme also provides some means of detecting an error in the previous checks themselves.
For our simple scheme, we can pick the left checks $L_i$ with $i=1,2,\ldots,m$ uniformly at random
from the Pauli group $\calP_n$ and determine the corresponding right checks as $R_i = UL_iU^{\dag}$. Since there is no advantage in applying the same check twice, we can sample from the Pauli group without replacement to ensure that all the left checks are unique. The identity operator commutes with all possible errors, and we therefore omit this element from the Pauli group when sampling.
In the generalized scheme, we again post-select on the zero syndrome for each of the $m$ checks.
As the number of checks $m$ grows, CPC is capable of detecting more and more errors in the payload circuit, at the cost of an exponentially decreasing
post-selection probability.

\begin{figure}[!h]
\centering
\includegraphics[height=3.5cm]{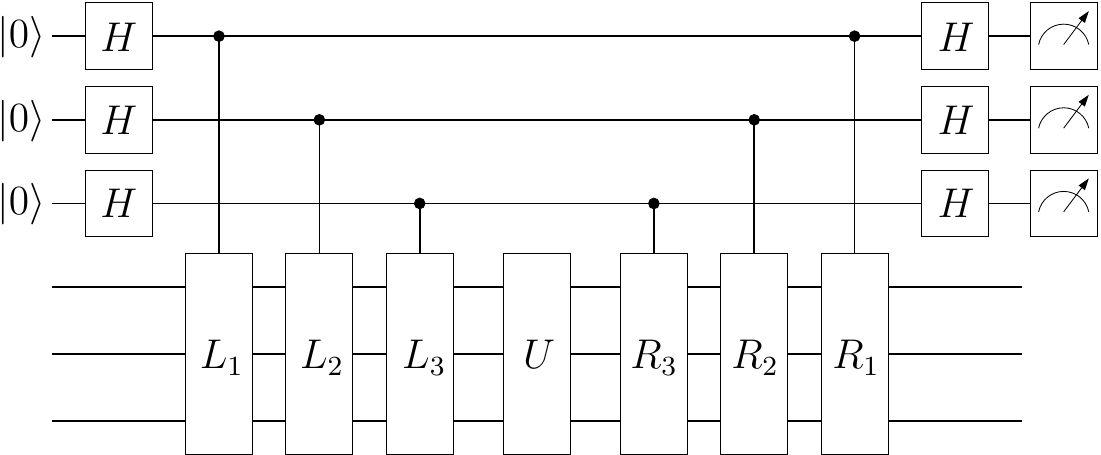}
\caption{Example of a coherent Pauli check on three data qubits. The check consists of three nested left-right checks, each defined by a pair of Pauli terms $L_i$ and $R_i$ and implemented using a single ancillary check qubit. For post-selection, the measurements on all three check qubits should be zero.}\label{Fig:MultipleChecks}
\end{figure}

\subsection{One-sided checks}\label{Sec:OneSided}

We now propose a special optimization to the CPC scheme for the case where the data qubits are measured directly after applying the (checked) payload circuit. The error mitigation protocol based on one-sided CPCs may find applications in quantum state tomography based on classical shadows~\cite{huang2020predicting} where a random Clifford operator
applied to the state of interest is directly followed by the measurement of each qubit in the standard basis. 
Instead of considering some Pauli error $E'$ occurring at the beginning of the payload circuit, we can push it through the circuit and obtain the equivalent error $E = UE'U^{\dag}$. The error syndrome $s$ now depends on whether $E$ commutes or anti-commute with the right Pauli check $R$. Given that we are now at the end of the circuit, we can disregard all Pauli-Z components in $E$, since these do not affect measurements in the standard basis. This leaves us with an effective error $E \in \{I,X\}^{\otimes n}$. Without loss of generality, it suffices to then choose the right Pauli check $R$ from the $n$-qubit Pauli-Z group $\{I,Z\}^{\otimes n}$, and set the associated left check to $L = U^{\dag}RU$. By inserting a pair of Hadamard gates between the left and right Pauli check gates on the ancillary check qubit, and inserting pairs of Hadamard gates between successive checks within the right check (see right-hand side of Figure~\ref{Fig:OneSided}), the right check consists of a series of subcircuits of the form

\begin{center}
\begin{tabular}{ccc}
\raisebox{-12pt}{\includegraphics[height=34pt]{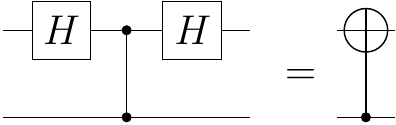}}&
\hspace*{12pt}
&
\begin{tabular}{ccc}
\ket{00} $\rightarrow$ \ket{00}\\
\ket{01} $\rightarrow$ \ket{11}\\
\ket{10} $\rightarrow$ \ket{10}\\
\ket{11} $\rightarrow$ \ket{01}
\end{tabular}
\end{tabular}
\end{center}
Next to the subcircuit, we show the effect it has on the computational basis states. For an input state $\ket{q_1, q_2}$ we can concisely represent the result of this transformation as $\ket{q_1 {\oplus} q_2, q_2}$, where $\oplus$ denotes the Boolean exclusive-{\sc{or}} ({\sc{xor}}) operation. In the absence of readout errors, we can perform these operations on the classical bits representing the measurements and therefore implement the right Pauli check entirely classically. When measurement errors can be modeled as a product of symmetric bit-flip channels, we can equivalently apply any readout error before or after the measurement. When implementing the right check as a quantum circuit any readout error on the payload qubits will not be detected and may pass post-selection, depending on the measured syndrome value, which itself may be affected by readout errors. When implementing the right check classically, the readout error on the data qubits can be viewed as a part of the payload circuit error. Detection depends on whether the combined error commutes with the right check or not, and whether errors are present on the check qubits, either during the application of the circuit or during readout.

\begin{figure}[!h]
\centering
\begin{tabular}{ccc}
\includegraphics[height=100pt]{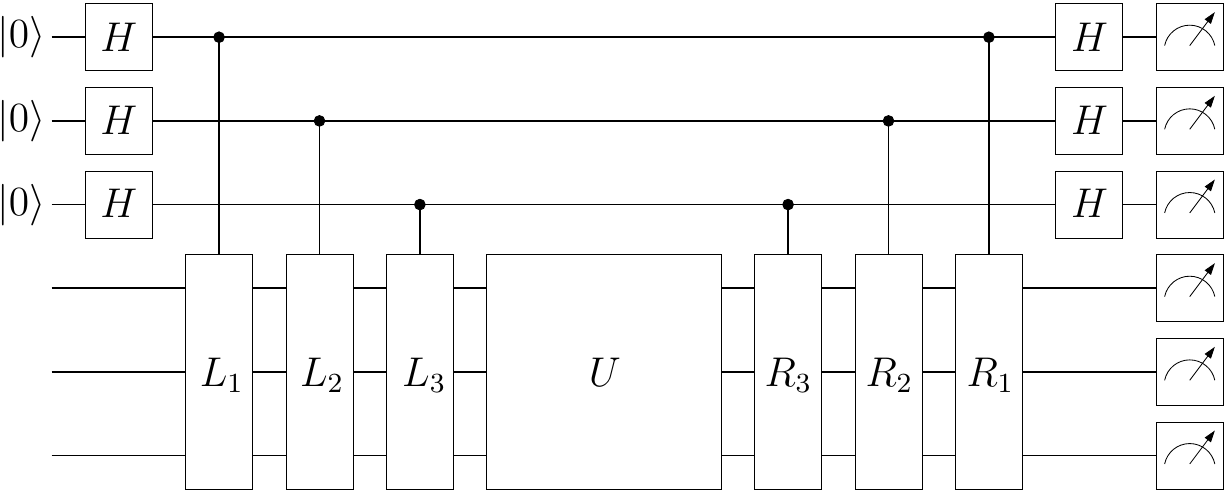}&
& \hspace*{12pt}
\includegraphics[height=100pt]{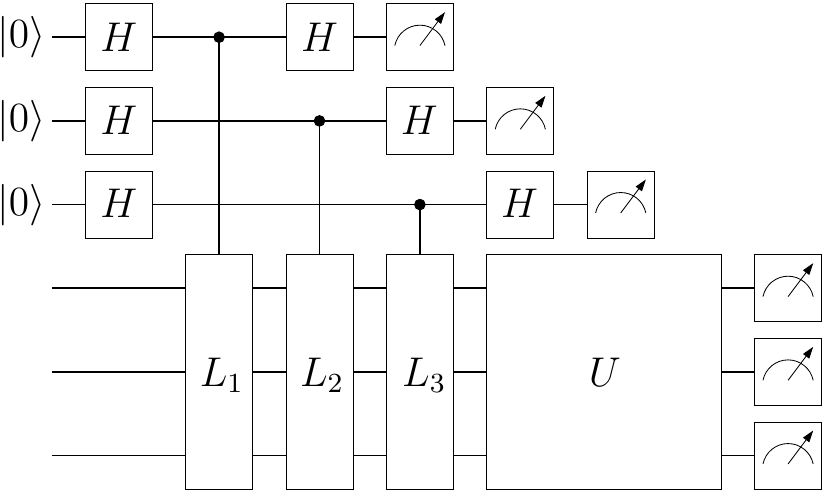}\\[4pt]
({\bf{a}}) & & ({\bf{b}})
\end{tabular}
\caption{Illustration of (a) a Clifford payload circuit $U$ with Pauli-Z right checks $R_i$ and corresponding left checks $L_i = U^{\dag}R_iU$, just prior to measurement; (b) simplified left-only Pauli check circuit in which the right checks are evaluated classically. Note that check qubits can be measured before the payload circuit has completed or even started.}\label{Fig:OneSided}
\end{figure}

\subsection{Flag qubits}

As gates implementing the left and right checks are themselves subject to noise, errors can occur on them. For two-sided check circuits, there can additionally be a considerable amount of idle time between the left and right checks as payload circuits and possibly other checks are applied. Such long idle times leave check qubits susceptible to noise such as thermal relaxation or coherent errors. Although such noise can be partially alleviated using techniques such as dynamical decoupling, some overall noise will remain. We can push any noise term occurring on the check qubits toward the end of the circuit.
In the absence of readout errors, any Pauli X or Y error on the check qubit will result in a syndrome value of one, directing the rejection of the data, since the post-selection criterion is not satisfied. When the error on the check qubit is Z, it will remain undetected by measurement in the computational basis. Although having a Pauli-Z error itself prior to measurement is harmless, it can be caused by a Pauli-Y error prior to the final Hadamard gate in the check circuit. This is important since pushing a Pauli X or Y noise term on the control of a conditional-P gate with $P\in \{X,Y,Z\}$ introduces a Pauli P error on the target qubit. That means that some error on the check qubit, which eventually reveals itself as a Pauli Z error could have introduced errors on the data qubits along the way. As such, we may want to have a mechanism for catching this type of error as well. For this, we introduce a second ancillary qubit that flags errors on the check qubits as follows:

\begin{center}
\includegraphics[height=65pt]{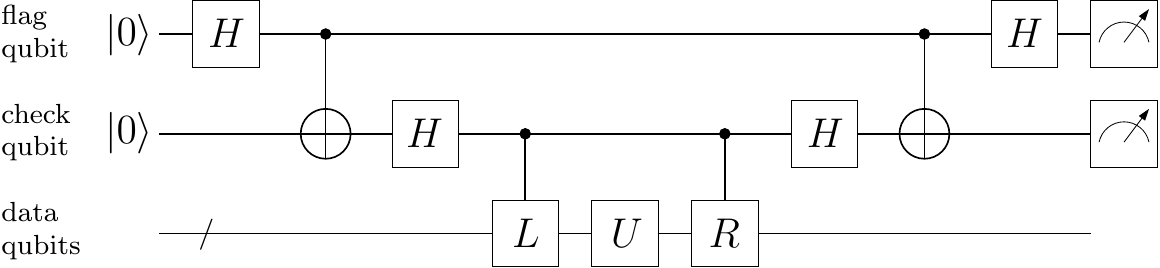}
\end{center}

This can be seen as a coherent Pauli-X check that applies only to the check qubit. If needed, this scheme can be repeated by adding another flag qubit to guard against errors on the first flag qubit, and so on. Given that the idle time of nested checks only increases we can expect the outer levels to be more susceptible to noise. Given that, ideally, idle time is a particular form of an identity operation it is also possible to apply the CPC approach to guard against error during this time by applying one or more checks with idle time as the payload circuit.

\subsection{Readout-error mitigation}\label{Sec:ReadoutIntro}

In Section~\ref{Sec:OneSided} we noted that one-sided checks allow us to incorporate readout errors on the data qubits as a part of the payload error. A special case of this is to have an empty payload on a single target qubit at the end of the circuit and catch any Pauli-X or Y errors (which manifest themselves as bit-flips in the computational basis measurement) using a one-sided Pauli-Z check. In essence, this amounts to setting up a repetition code prior to readout. Adding two levels of checks results in a three-bit measurement for the target qubit. We can then post-select the measurement if all bits are equal and the code word is valid, or loosen this criterion and apply majority voting on the measured bits to resolve the measurement of the target qubit. Repeated Pauli-Z checks are implemented by applying the $\cnotgate$ gates on each ancillary repeat qubit controlled by the target qubit, or previously connected repeat qubits. Similar ideas for readout-error mitigation were proposed in~\cite{Guenther_2022,PhysRevA.105.012419}.

\begin{figure}[t]
\centering
\begin{tabular}{ccc}
\includegraphics[height=145pt]{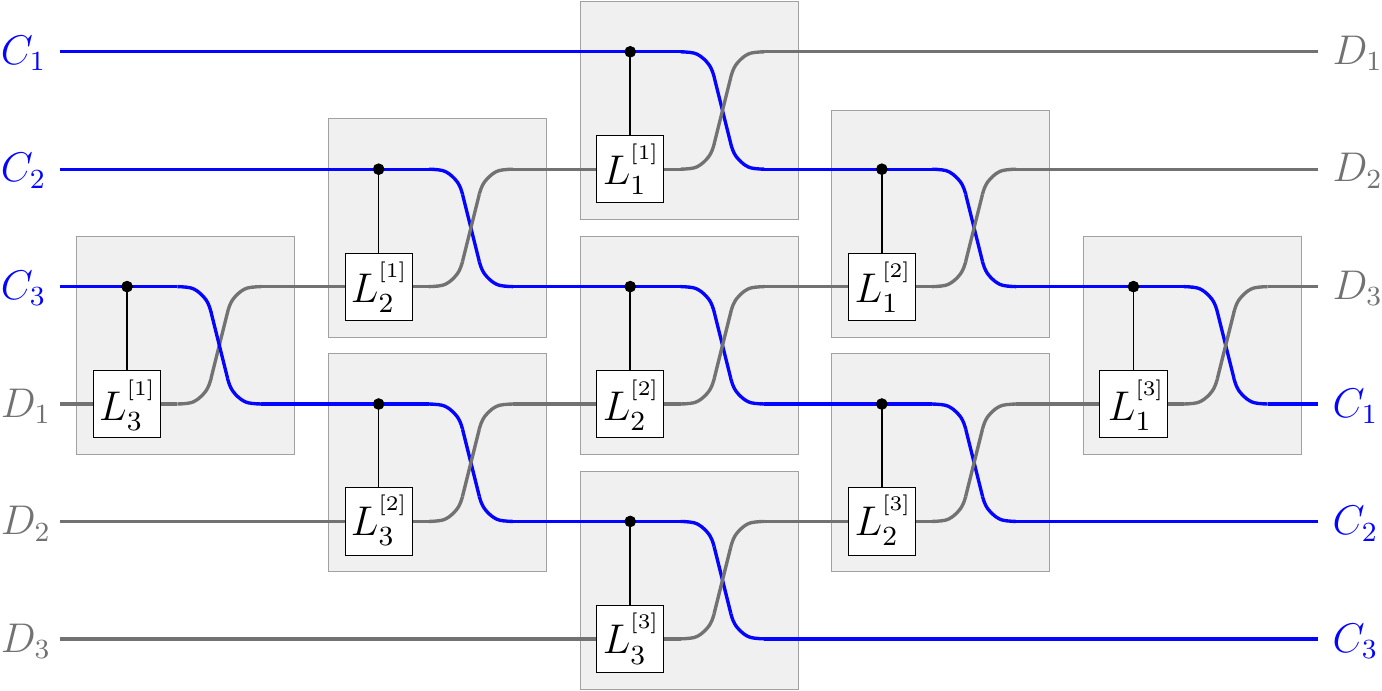}
&
\hspace*{12pt}
&
\includegraphics[height=145pt]{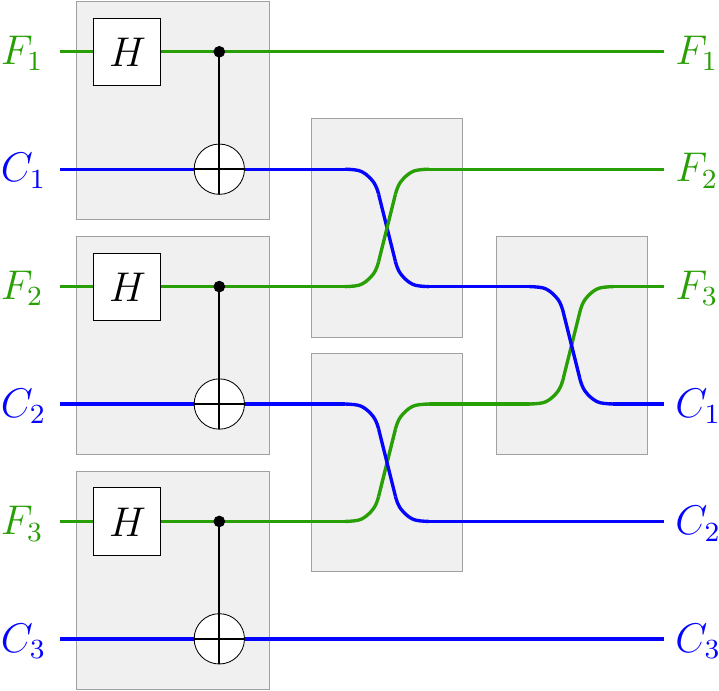}\\[4pt]
({\bf{a}}) & & ({\bf{b}})
\end{tabular}

\vspace*{8pt}
\begin{tabular}{ccccccc}
\includegraphics[height=45pt]{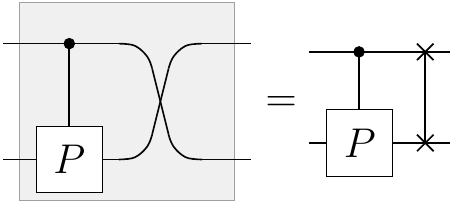}&\ &
\includegraphics[height=45pt]{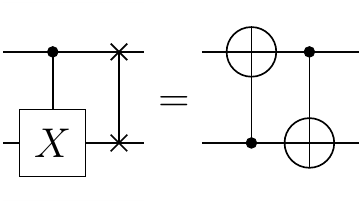}&\ &
\includegraphics[height=45pt]{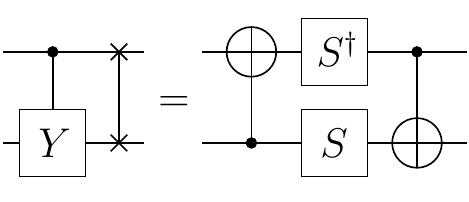}&\ &
\includegraphics[height=45pt]{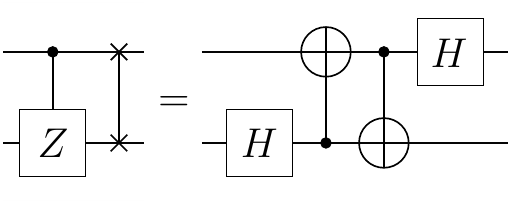}\\[0pt]
({\bf{c}}) && ({\bf{d}}) && ({\bf{e}}) && ({\bf{f}})
\end{tabular}
\caption{Efficient implementation of (a) the left check circuit, and (b) the left-side of the flags over LNN. The building block (c) consisting of a controlled single-qubit Pauli followed by a $\swapgate$ can be simplified for X, Y, and Z Paulis (d--f). A layout for implementing Pauli measurements similar to (a) was given in~\cite{PhysRevA.98.062339}.}\label{Fig:LNN}
\end{figure}

\subsection{Linear qubit connectivity}\label{Sec:LNNCircuit}

The implementation of coherent Pauli checks based on the circuit expansion shown in Figure~\ref{fig:CPC1} requires high-degree qubit connectivity. Most  contemporary quantum processors support only limited qubit connectivity and therefore require such circuits to be implemented using a set of two-qubit gates enabled by the qubit-to-qubit connectivity map. Without special care, such transformations could result in an unnecessarily large number of $\swapgate$ operations. We now present an approach that efficiently maps check circuits onto a chain of qubits with linear nearest neighbor (LNN) connectivity. We illustrate this for a three-qubit left-side check circuit on three data qubits in Figure~\ref{Fig:LNN}(a), omitting for clarify the Hadamard gates that precede the conditional Pauli operations. As a result of the linear connectivity, two-qubit gates are available only on neighboring qubit pairs. The corresponding right-side circuit mirrors this along the y-axis with left checks $L$ substituted by the appropriate right checks $R$. The initial qubit order starts with all check qubits $C_i$, and is followed by the data qubits $D_j$. The idea of the construction is to apply controlled single-qubit Pauli operations on the adjacent check and data qubits, followed by a $\swapgate$ operation to update the qubit order (illustrated by the crossing qubit lines implementing the $\swapgate$ operation). When flag qubits are desired the initial qubit order interleaves check and flag qubits, as shown in Figure~\ref{Fig:LNN}(b). We then apply all checks and repeatedly swap qubits until all check qubits are clustered together and ready for application of the left check circuit. When a quantum processor only supports operations locally equivalent to the $\cnotgate$ gate, it would seem that each of the gray blocks in Figure~\ref{Fig:LNN}(c) requires three $\cnotgate$ operations to implement the $\swapgate$, and one additional $\cnotgate$ operation, possibly combined with single-qubit gates, whenever $L_{i}^{[j]} \in \{X,Y,Z\}$.  The efficiency of the proposed approach stems from the fact that this template can be simplified, as shown in Figures~\ref{Fig:LNN}(d)--(f): Controlled X, Y, and Z operations followed by a $\swapgate$ can all be implemented using two rather than four $\cnotgate$ gates.  Interestingly, the conditional identity operation, which itself does not require any gates, followed by a $\swapgate$ is now the most expensive, requiring three $\cnotgate$ gates.

\section{Analysis}

\subsection{Asymptotic Pauli-check performance}\label{Sec:AsympototicCPC}

The performance of coherent Pauli checks is characterized by the post-selection rate (that is, the fraction of shots that pass the selection criteria), and the logical error rate in the shots that passed the selection. We now show how to compute these quantities,
both for a finite number of checks and in the asymptotic regime where the number of checks goes to infinity.

\subsubsection{Markov model}\label{Sec:MarkovModel}

In order to model the performance of the Pauli-check framework we iteratively add single checks and use a simple Markov model to update a state vector that represents the probability of being in one of the following three states: (1) detected error, (2) undetected error, and (3) no error.
The detected-error state means that, up to that point, at least one of the checks was activated, which means it will fail post-selection and we will eventually discard data from this circuit run. As such, once this state is reached, we remain in this state. The undetected-error state means that some error occurred in the payload or check gates, but none of the checks so far was activated. The no-error state indicates that no errors occurred in either the payload circuit or the checks considered so far. Denoting by $\epsilon_{pl}$ the probability that the payload circuit is affected by an error, we can write the initial state, including only errors in the payload circuit, as
\begin{equation}\label{Eq:InitialMarkovState}
\pi^{(0)} =
\left(
\begin{array}{c}
\pi_1^{(0)}\\ \pi_2^{(0)}\\ \pi_3^{(0)}
\end{array}
\right)
=
\left(\begin{array}{c}0\\\epsilon_{pl}\\ 1-\epsilon_{pl}\end{array}\right).
\end{equation}
The state is updated in a Markovian manner and the state resulting from the application of the $c$-th check is given by $\pi^{(c)} = T^{(c)}\pi^{(c-1)}$. The only way we can transition into an error-free state $\pi_3^{(c)}$ is if we started in an error-free state and had no error in the check gates, which occurs with some probability $t_{ok}^{(c)}$. We can arrive at or stay in the undetected state $\pi_2^{(c)}$ in several ways. Starting from an error-free state, we could have noise in the check that goes undetected, which happens with some probability $t_u^{(c)}$. If we are already in the undetected-error state, we can stay there if the current error commutes with the check, which is assumed to occur with probability one half, and we either have no error or an undetectable error. Alternatively, also with probability one half, the error anti-commutes with the check, but detection is then negated by a detectable error, which occurs with probability $t_d^{(c)}$. Since $t_{ok}^{(c)}+t_u^{(c)}+t_d^{(c)} = 1$ it follows that the overall transition probability is 1/2. Applying similar logic to the detected-error state and assuming that the various probabilities are independent of the check index $c$, we obtain the following transition matrix:
\begin{equation}\label{Eq:MarkovModelTransition}
T = \left(
\begin{array}{ccc}
1 & \half & t_d \\[5pt]
0 & \half & t_u\\[5pt]
0 & 0 & t_{ok}
\end{array}
\right).
\end{equation}
The post-selection and logical error rates for a state $\pi = \pi^{(c)}$ are given by:
\[
P(\mbox{postselect}) = \pi_2{+}\pi_3 \text{ \;and\; }
P(\mbox{logical error}) = \frac{\pi_2}{\pi_2+\pi_3}.
\]

\subsubsection{Errors in the check gates}

We now study the transition probabilities $t_d$ and $t_{ok}$, and consequently $t_{u} = 1 - (t_d+t_{ok})$. For this we first consider two-sided Pauli checks on a fully-connected topology.
Each check is implemented using a set of gates before and after the payload circuit, which we respectively refer to as the left-check and right-check gates. For simplicity we assume that single-qubit gates are noiseless and that the controlled-Pauli gates CP with $P \in \{X,Y,Z\}$ are affected by identical two-qubit depolarizing noise channels
\begin{equation}\label{Eq:DepolarizingNoise}
\mathcal{D}_{\epsilon}(\rho) = (1-\epsilon)\rho
+ \frac{\epsilon}{15}\sum_{i=1}^{15}P_i\rho P_i^{\dag}.
\end{equation}
Since $\mathcal{D}_{\epsilon}$ is invariant under conjugation by any two-qubit Clifford gate, we can freely choose whether the noise appears before or after the gate. In fact, we can push the noise channel through any adjacent one- and two-qubit Clifford gates and always assume the error occurs directly prior to
each left check and directly following each right check. For each CP gate in the right check the error occurs with probability $\epsilon$. The resulting Pauli term on the control qubit following the check will then be $I$ with probability $3/15$ and one of $X$, $Y$, or $Z$ with probability $4/15$ each.
We assume that any error on a check qubit affects the data qubits, and are therefore interested only in the overall Pauli term on the check qubit. Given that all errors due to the check gates can be assumed to occur before and after the left and right checks, we can logically cancel the check and combine all error terms on the check qubit. Since there are no longer any gates connecting the check and data qubits, we can disregard the payload circuit. An error is detectable if the overall Pauli term on the check qubit is either $X$ or $Y$, and is undetectable otherwise. Since an even number of detectable errors results in an undetectable $I$ or $Z$ term we can only detect an error when there is  odd number of detectable errors.
 For a check that is implemented using $k$ CP gates, each with a detectable error probability of $p = 8\epsilon/15$, the overall rate of detectable errors is given by
\begin{equation}\label{Eq:td}
t_d = \sum_{\mathrm{odd}\ \ell} \binom{k}{\ell}p^{\ell}(1-p)^{k-\ell}
= \half\left(1 - (1-2p)^k\right).
\end{equation}

For the implementation of the checks on a linear topology we can first push the error for each of the two or three gates in a single check to the beginning or end of the block. As illustrated in Figure~\ref{Fig:LNN}(a), we see that although the check qubit moves between physical qubits, it can be regarded as a fixed qubit. Therefore, there is no logical difference to the all-to-all connectivity and the above analysis continues to hold with $p = 8\epsilon/15$, albeit with a different number $k$ of CP gates. Likewise, for one-sided checks, we can push all errors to the end of the left check and consider detectable errors as a bit flip of the readout. Regardless of the setting, the probability of having no errors is given by $t_{ok} = (1-\epsilon)^k$. The probability of undetectable errors is given by $t_u = 1 - (t_d + t_{ok})$.
As a simplification, we disregarded the possibility that noise in the check circuit could leave the data qubits untouched. For instance, the only noise affecting a check could be an XI term following the final two-qubit gate in the right-side check, or noise terms on the data qubits could cancel. Some of these errors may be classified as detected errors, which decreases the model's post-selection probability. In addition, it may increase the probability of undetected errors, which increases the logical error rate. As such, this simplification may lead to a slightly pessimistic result. Under the assumption that errors commute or anticommute with the checks with probability 1/2, the model would then give a lower bound on the post-selection probability and an upper bound on the logical error rate.

We now consider the number $k$ of two-qubit gates needed to implement the checks. The expected number of gates for uniformly sampled payload operators $U$ is given in Table~\ref{Table:KValues}. Since we randomly sample the checks, there will clearly be some variation around the expected value. For two-sided checks we have the following result:

\begin{lemma}
Let $L$ be a left check sampled uniformly at random from the Pauli group, and define $R = ULU^{\dag}$, with arbitrary Clifford $U$. Then the number of gates $k$ needed for the implementation of a two-sided Pauli check associated with L and R satisfies
\[
\beta n - \delta\sqrt{2n} \leq
k
\leq \beta n + \delta\sqrt{2n}
\]
with probability at least $1 {-} 2e^{-\delta}$, where $\beta = \sfrac{3}{2}$ for all-to-all and $\beta = \sfrac{9}{2}$ for linear nearest neighbor connectivity.
\end{lemma}
\begin{proof}
Consider a left Pauli check $L=\prod_{i=1}^n L^{[i]}$.
Let $k_i$ be the number of {\sc{cnot}} gates used to implement the controlled-$L^{[i]}$ gate.
In the all-to-all settings, $k_i=0$ with the probability $1/4$ (when $L^{[i]}=I$) and
$k_i=1$ with the probability $3/4$ (when $L^{[i]}=X,Y,Z$). Thus the expected value  of $k_i$ equals $3/4$. In the LLN settings, $k_i=3$ with the probability $1/4$ (when $L^{[i]}=I$) and $k_i=2$ with the probability $3/4$ (when $L^{[i]}=X,Y,Z$). Thus the expected value of $k_i $ equals $9/4$. 
The number of {\sc{cnot}} gates used to implement the full left check (i.e. controlled-$L$ gate)
is $k=\sum_{i=1}^n k_i$. It has the expected value $\beta n/2$, where  
$\beta = 3/2$ and $\beta=9/2$ for the all-to-all and LLN settings respectively.
By Hoeffding's inequality, the random variable $k$ deviates from the expected value of $\beta n/2$ by more than 
$\delta\sqrt{n/2}$ with probability at most $e^{-\delta}$. 
Exactly the same arguments apply to the right check
$R = ULU^{\dag}$, since $R$  is distributed uniformly on the Pauli group for any fixed Clifford $U$.
Thus, on average, one needs $\beta n$ {\sc{cnot}} gates to implement a two-sided Pauli check. 
Although $R$ and $L$, and therefore the number of gates in the left and right checks, are correlated, it follows from the union bound that the probability of at least one of the gate counts exceeding the given range is bounded by $2e^{-\delta}$. It follows that their sum deviates from $\beta n$ by no more that $2\delta\sqrt{n/2} = \delta\sqrt{2n}$ with probability at least $1{-}2e^{-\delta}$, as stated.
\end{proof}

For left-only checks, we randomly sample the right checks uniformly at random from the Pauli-Z or identity. The weight of the left checks, and consequently the number of gates needed to implement the check, depends on the payload circuit. When the payload circuit implements a random permutation on $n$ data qubits the weight of the left checks matches that of the right checks and has an expected value of $n/2$. In order to characterize bounds on the number of gates, we therefore also need to assume that $U$ is sampled uniformly at random from the group of $n$-qubit Clifford operators. The result then holds with $\beta = 3/4$ for all-to-all and $\beta = 9/4$ for linear nearest neighbor connectivity.

\begin{table}
\centering
\begin{tabular}{lcc}
\hline
{\bf{Description}} && $k$\\
\hline
Fully connected, left \& right checks && $3n/2$\\[4pt]
Fully connected, left checks only && $3n/4$\\[4pt]
Linear, left \& right checks && $9n/2$\\[4pt]
Linear, left checks only && $9n/4$\\
\hline
\end{tabular}
\caption{The number of {\sc{cnot}} gates for a single check on $n$ data qubits is concentrated around the given $k$ values.}\label{Table:KValues}
\end{table}

\subsubsection{Asymptotic logical error rate}

For the asymptotic logical error rate, we can disregard the part of the Markov model that is associated with the probability mass of the detected errors. In particular, we consider only the lower-right block of the transition matrix:
\[
T' = \left(
\begin{array}{cc}
\frac{1}{2} & t_u\\[5pt]
0 & t_{ok}
\end{array}
\right)
\]
The case $t_u = 0$ occurs only when $\epsilon = 0$, in which case we have $t_d = 0$ and $t_{ok} = 1$. For each added check the number of undetected errors in post-selection can only decrease, while the fraction of circuit instances with no error remains the same. Consequently, the logical error rate will decrease to zero. In the general case where $t_u > 0$, scaling the relevant part of the state vector is needed, we can write
\begin{equation}\label{Eq:TPrime}
\left(
\begin{array}{c}
\half + \alpha t_u\\
\alpha t_{ok}
\end{array}
\right)
=
T' \left(\begin{array}{c}1\\\alpha\end{array}\right).
\end{equation}
A fixed point occurs whenever $(1,\alpha)^T$ is a (scaled) eigenvector of $T'$, which implies
\begin{equation}\label{Eq:Alpha}
\frac{\alpha}{1} = \frac{\alpha t_{ok}}{\half + \alpha t_u}.
\end{equation}
This is satisfied for $\alpha = 0$ or $\alpha = (t_{ok} - \half) / t_u$. Since the state vector cannot have negative entries, we must have $\alpha = 0$ for $t_{ok} \leq \half$, implying an asymptotic logical error rate of 1. For $t_{ok} > \half$ it follows from the upper-diagonal form of $T'$ that $t_{ok}$ is the dominant eigenvalue of $T'$. When $\epsilon_{pl} < 1$, the initial $\alpha$ will be strictly positive,
and it follows from Eq.~\ref{Eq:Alpha} that $\alpha = (t_{ok}-\half) / t_u$. This means we converge to a logical error rate of
\[
\frac{\pi_2}{\pi_2 + \pi_3}=
\frac{1}{1+\alpha} =
\frac{t_u}{t_u + t_{ok} - \half} =
\frac{t_u}{\half - t_d}.
\]
Putting everything together, we can concisely express the asymptotic logical error rate as
\begin{equation}\label{Eq:AsymptoticError}
E_{\mathrm{asymp.}} =
\begin{cases}
t_u / \big(\half - t_d\big) & \mbox{if $t_{ok} > \half$ and $\epsilon_{pl}< 1$}\\
1 & \mbox{otherwise}.
\end{cases}
\end{equation}
We can expand the first case as follows:
\[
t_u/(\half-t_d)
= \frac{1-(t_d+t_{ok})}{\half - t_d}
= 1 + \frac{1-2t_{ok}}{1-2t_d}
= 1 + \frac{1 - 2(1-\epsilon)^k}{(1 - \frac{16}{15}\epsilon)^k}
= \frac{14}{15}k\epsilon + \mathcal{O}(k^2\epsilon^2),
\]
where the last expression follows from Taylor series expansion around $\epsilon=0$, and $k$ is as given by Table~\ref{Table:KValues}. For left-right checks on a fully connected topology this gives an approximate logical error rate of $7n\epsilon /5$.
As for post-selection, it follows from Eq.~\ref{Eq:TPrime} that the ratio of successive post-selection rates is given by
\[
s(\alpha) := \frac{\half + \alpha t_u + \alpha t_{ok}}{1+\alpha}.
\]
For $t_{ok} > \half$ and $\alpha \geq 0$ we have
\[
\frac{ds}{d\alpha}(\alpha)
= \frac{t_u + t_{ok}}{1+\alpha} - \frac{\half + \alpha(t_u+t_{ok})}{(1+\alpha)^2}
=
\frac{t_u+{t_{ok}} - \half}{(1+\alpha)^2} \geq 0,
\]
which means that $s(\alpha)$ is monotonically non-decreasing for $\alpha \,{\geq}\, 0$. As a result, it follows that
\[
\half = s(0) \leq s(\alpha) \leq \lim_{\alpha'\to\infty}s(\alpha') = t_u+t_{ok}.
\]
In other words, at worst the post-selection rate is halved at every iteration; at best is it multiplied by $t_u + t_{ok}$.
The logical error rate is invariant under the normalization of $(\pi_2',\pi_3') = (1,\alpha)$ and can therefore be expressed as $1/(1+\alpha)$. With decreasing logical error rate or, equivalently, increasing $\alpha$, the decrease in post-selection rate slows down.
Finally, it follows from Eq.~\ref{Eq:TPrime} and the asymptotic value of $\alpha = (t_{ok}-\half)/t_u$, that the post-selection rate asymptotically decreases by a factor
\[
s\big((t_{ok}-\half)/t_u\big)
=
\frac{\half + (t_{ok} - \half) + t_{ok}(t_{ok}-\half)/t_u}{(t_u + t_{ok}-\half)/t_u}
=
\frac{t_{ok}(t_u + t_{ok} - \half)}{t_u + t_{ok}-\half}
= t_{ok}.
\]

\subsubsection{Numerical simulations}

\begin{figure}
\centering
\begin{tabular}{ccc}
\includegraphics[width=0.32\textwidth]{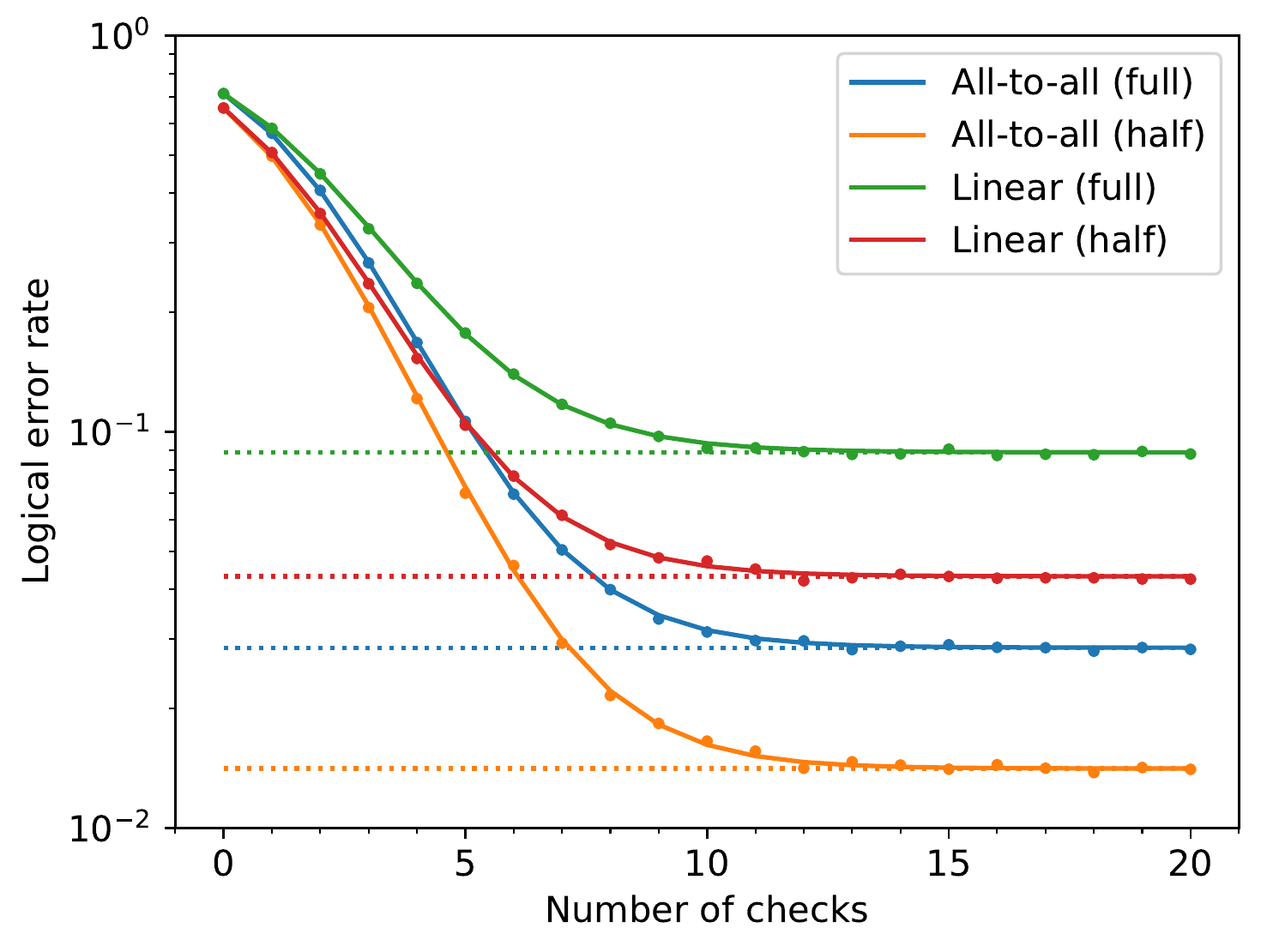}&
\includegraphics[width=0.32\textwidth]{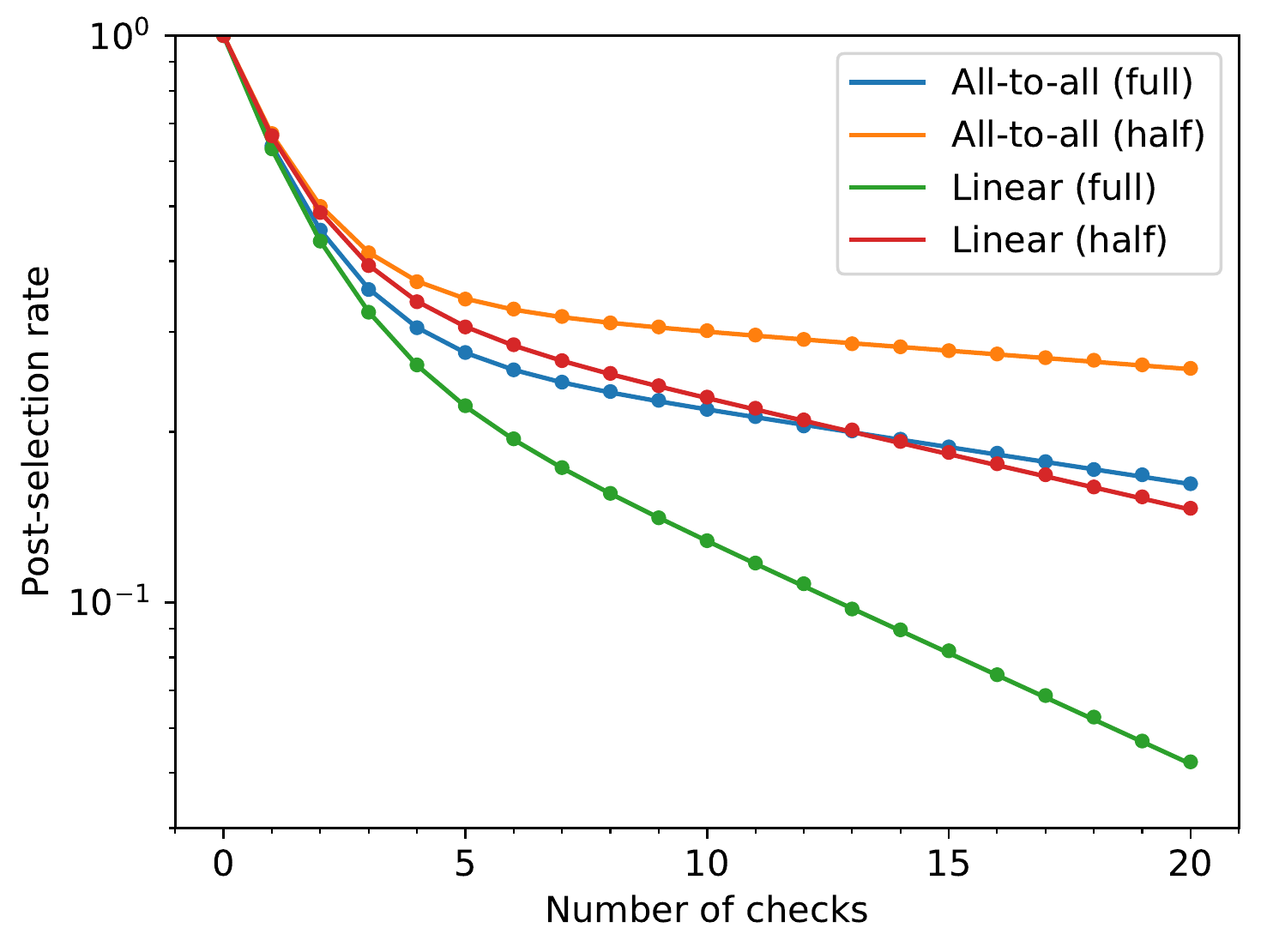}&
\includegraphics[width=0.32\textwidth]{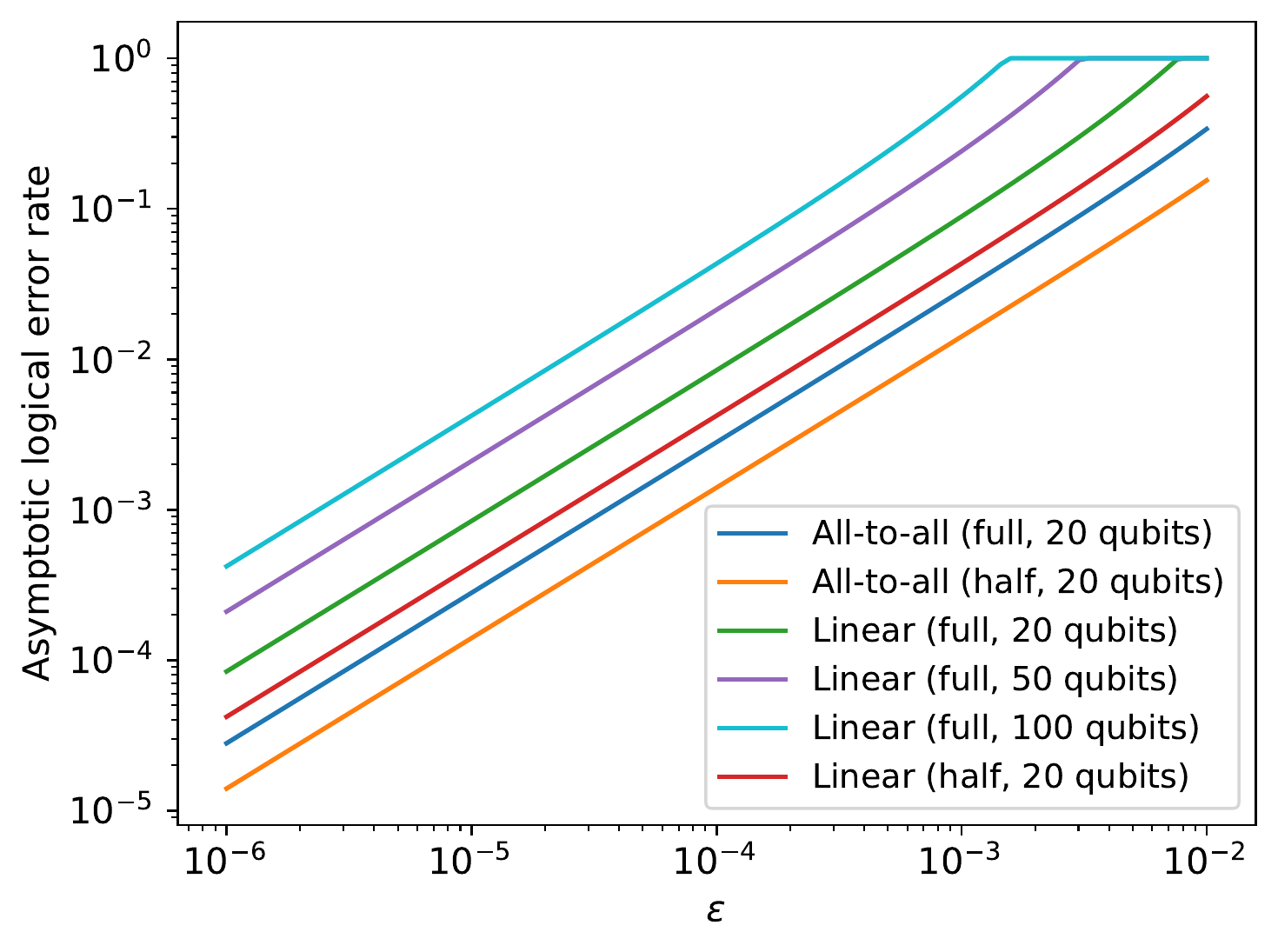}\\
({\bf{a}}) & ({\bf{b}}) & ({\bf{c}})
\end{tabular}
\caption{Plots of (a) the logical error rate and (b) the post-selection rate as a function of the number of Pauli checks based on simulated data (dots) and evaluation of the theoretical model (solid lines) for a randomly sampled 20-qubit Clifford circuit with depolarizing noise ($\epsilon=0.001$) on the two-qubit gates. The asymptotic logical error rates for the settings in plot (a) are indicated by the horizontal dotted lines. The simulated results are obtained by combining data from $10^5$ noisy circuit instances each for 20 randomly sampled check instances. Plot (c) shows the asymptotic logical error rates as a function of $\epsilon$ for different settings.}\label{Fig:Bounds001}
\end{figure}

For a better understanding of the performance of CPC and the derived asymptotic floor values of the logical error rate we numerically simulate the method. As the first step, we sample a 20-qubit Clifford operator uniformly at random \cite{bravyi2021hadamard} and map it to a quantum circuit that is optimized for the LNN connectivity (see~\cite{bravyi2021hadamard, maslov2022cnot} for more details). We then generate checked circuits with randomly sampled one- or two-sided checks over either all-to-all or LNN architecture, giving four distinct settings. The resulting circuits are all Clifford, which allows to simulate them in a compact tableau representation based on the stabilizer formalism~\cite{aaronson2004improved}. Instead of tracking the state as it evolves by successive application of the gates, we represent by each row of the tableau the accumulated error of a single circuit instance. Application of a gate then amounts to pushing the error through the gate. For a noisy gate we sample a random Pauli term according to the associated Pauli channel and multiply it by the existing noise term.
When all gates in the circuit have been processed we end up with sampled error terms as they would occur just prior to measurement. In our simplified setting, we assume that state preparation, readout, and all single-qubit operations are noiseless. We further assume that all two-qubit gates are affected by depolarizing noise channels (see Eq.~\ref{Eq:DepolarizingNoise}) with identical $\epsilon$ values. Without loss of generality we can assume that the initial state of the data qubits is given by $U^{\dag}\ket{0}$, which means that, in the absence of gate errors, we should measure the all-zero state. Based on this assumption, we can process the errors captured by the tableau and, possibly after classical application of the right check, determine whether the sample is accepted during post-selection, and whether an error occurred on the data qubits. For the one-sided check, we disregard any Pauli-$Z$ components in the errors since these do not affect the measurements. For the two-sided checks, we assume that the final state is not yet measured but instead participates in further computations. Any non-identity Pauli terms on the data qubits are therefore considered to be an actual error.

For the simulations, we allow up to 20 checks, and for each setting, we determine the number of correct and post-selected shots as the average over 20 random check instances, each with $10^5$ shots. The resulting logical error and post-selection rates, based on the depolarizing strength $\epsilon=0.001$, are shown as dots in Figure~\ref{Fig:Bounds001}. We superimpose as solid lines the values predicted by the Markov model using the $k$ values from Table~\ref{Table:KValues} and a payload error rate $\epsilon_{pl}$ as estimated by the numerical simulation with zero checks. Finally, we indicate the asymptotic logical error rate as given by Eq.~\ref{Eq:AsymptoticError} by a horizontal dotted line. Despite the simplifying assumption, we observe that the theoretically predicted values are remarkably close to the simulated values.

\subsection{Readout-error mitigation using checks}

We now consider an instance of the readout-error mitigation scheme described in Section~\ref{Sec:ReadoutIntro}. As mentioned in Section~\ref{Sec:OneSided}, measurement errors that occur during one-sided Pauli-Z checks can be considered to be errors associated with the payload circuit. By defining an empty payload at the end of the circuit, just prior to measurement, we can therefore use a Pauli-Z check to detect measurement errors. Repeating the same on the measurement of the check itself, we obtain a quantum circuit with nested checks, as illustrated in Figure~\ref{Fig:RepeatedReadoutSims}(a). The purpose of the circuit is to obtain an accurate readout of the object qubit at the bottom using the ancillary qubits above it. The dashed boxes model the locations where independent bit-flip errors are expected. By applying the right-side checks classically it can be verified that an outcome is accepted only when all measured bits match. As can also be seen directly from the circuit itself, we effectively encode the object qubit using a repetition code prior to readout, and accept only valid code words during decoding.

\begin{figure}[!b]
\centering
\begin{tabular}{ccc}
\raisebox{45pt}{\includegraphics[width=0.25\textwidth]{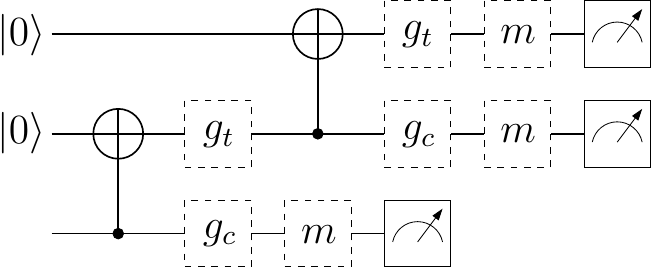}}&
\includegraphics[width=0.345\textwidth]{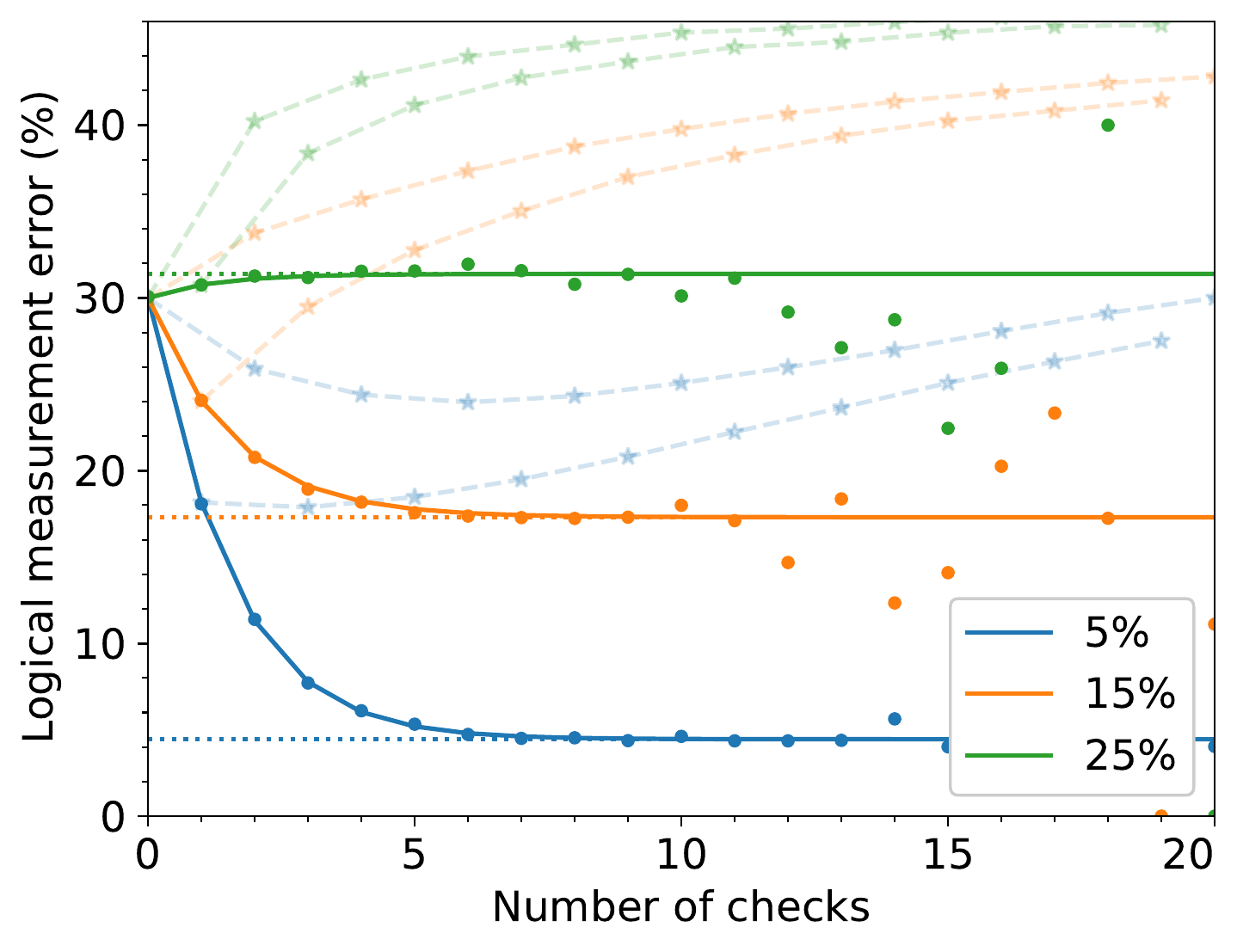}&
\includegraphics[width=0.345\textwidth]{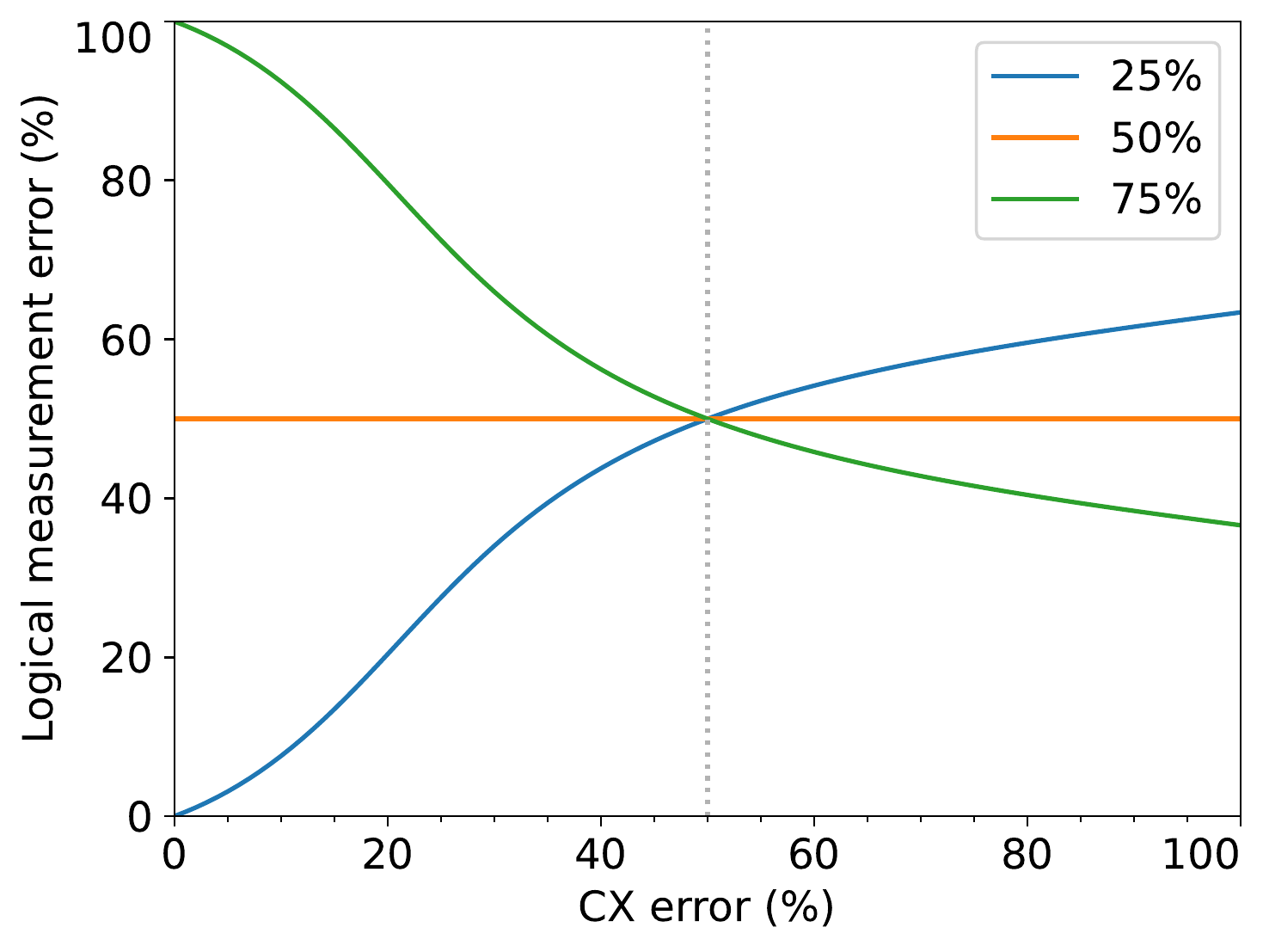}\\
({\bf{a}}) & ({\bf{b}}) & ({\bf{c}})
\end{tabular}
\caption{(a) Instance of a quantum circuit for reading out the bottom qubit with two checks. The dotted boxes indicate the location and probabilities $c$ and $m$ of bit-flip errors in the model due to $\cnotgate$ gates and measurements, respectively. (b) The logical measurement error obtained using simulation of the model based on $10^6$ samples (dots) as well as using the theoretical model (solid lines) when we fix the physical readout error rate $m$ to 30\% and use {\sc{cnot}} gates with error rates $g_c=g_t$ chosen from the set 5\%, 15\%, 25\%. The asymptotic logical measurement error is indicated by the horizontal dotted line. For comparison we also plot the results obtained using majority voting of the measured bits for even and odd numbers of checks (asterisks and light dashed lines). (c) The asymptotic logical measurement error as a function of {\sc{cnot}} error rate $c$ for three different readout error rates $m$.}\label{Fig:RepeatedReadoutSims}
\end{figure}

\subsubsection{Error modeling}

For modeling of the logical readout error, we first assume that each measurement is affected by independent symmetric bit-flip channels, each with a bit-flip probability $m$. As seen in Figure~\ref{Fig:RepeatedReadoutSims}(a), we model the measurement errors prior to measurement. This is merely for convenience and we could equivalently have modeled them as classical noise following an ideal measurement. Each {\sc{cnot}} gate is followed by independent symmetric bit-flip channels on the control and target qubits, with transition probabilities $g_c$ and $g_t$, respectively.

If there are no checks, we accept all measurements and therefore have a logical readout error rate equal to $m$. For the remainder of this discussion we assume there is at least one check. This means that the overall measurement error on the objective qubit combines the {\sc{cnot}} control error and the measurement error, resulting in a combined bit-flip channel with transition probability $m' = g_c(1-m) + (1-g_c)m$. Given the symmetry of the noise channels, we can assume, without loss of generality, that the object qubit is in the $\ket{0}$ state. Measurement of the object qubit therefore results in $0$ with probability $1{-}m'$ and $1$ with probability $m'$.

We analyze the performance of the Pauli-checked readout using a Markov model with each step representing an additional check. For the measurement outcome to be accepted we require successive checks to match the outcome of the object qubit. We also keep track of the state of the (ancillary) qubit prior to the {\sc{cnot}} control and measurement errors, as this is the state that will propagate to the next ancillary qubit. We represent the current state and the required measurement value as a tuple. Finally, we need an error state that indicates that a mismatch in the measured values was encountered. Using these components we represent the initial state as
\[
\mbox{init}
= \left(\begin{array}{c} \mbox{\ket{0}, measure 0}\\ \mbox{\ket{1}, measure 0}\\
\mbox{\ket{0}, measure 1}\\
\mbox{\ket{1}, measure 1}\\
\mbox{error}\end{array}
\right)
= \left(\begin{array}{c}1-m'\\0\\m'\\0\\0\end{array}\right).
\]
Application of the {\sc{cnot}} gate on ancillary qubit initialized to $\ket{0}$ results in a state that matches the previous state. However, we then need to apply noise on the target qubit (that is, the current ancillary qubit), which can flip the current state. This amounts to the multiplication of the current state by the transition matrix
\[
G = 
\left(\begin{array}{cc|cc|c}
1-g_t & g_t & 0 & 0 & 0\\
g_t & 1-g_t & 0 & 0 & 0\\
\hline
0 & 0 & g_t & 1-g_t & 0\\
0 & 0 & 1-g_t & g_t & 0\\
\hline
0 & 0 & 0 & 0 & 1
\end{array}\right).
\]
At this point we leave the current state unaffected and merely determine the probability with which the (noisy) measurement matches the desired result, and with which probability it fails to match, which gives a state transition to the error status. Depending on whether we are dealing with an intermediate or the final check, the measurement error is given by $m'$ or $m$. Denoting this error by $\alpha$ we have a transition matrix
\[
M(\alpha) = 
\left(\begin{array}{cc|cc|c}
1-\alpha & 0 & 0 & 0 & 0\\
0 & \alpha & 0 & 0 & 0\\
\hline
0 & 0 & \alpha & 0 & 0\\
0 & 0 & 0 & 1-\alpha & 0\\
\hline
\alpha & 1-\alpha & 1-\alpha & \alpha & 1\\
\end{array}\right).
\]
When applying $k$ checks we have $k{-}1$ intermediate checks with combined measurement error $m'$, and one final check with measurement error $m$. The final state after $k {\geq} 1$ checks can therefore be expressed as
\begin{equation}\label{Eq:FinalState}
\mathrm{final}(k) = M(m)G\Big(M(m')G\Big)^{k-1}\mathrm{init}.
\end{equation}
Given a final state we can express the post-selection and correct measurement probabilities respectively by
\begin{align*}
\langle \mbox{postselect},\mbox{final}\rangle &\quad\mbox{with}\quad
\mbox{postselect} = (1,1,1,1,0)^T\\
\langle \mbox{correct},\mbox{final}\rangle& \quad\mbox{with}\quad \mbox{correct} = (1,1,0,0,0)^T.
\end{align*}
Here and below we write $\langle a,b\rangle\equiv \sum_{i=1}^5 a_i b_i$ for the inner-product of five-dimensional vectors $a$ and $b$.
We can consequently write the logical measurement success rate as  $\langle \mbox{correct},\mbox{final}\rangle/ \langle \mbox{postselect},\mbox{final}\rangle$.

\subsubsection{Asymptotic measurement error}

In order to compute the asymptotic measurement success rate, it helps to consider the eigendecomposition $M(m')G = V\Lambda V^{-1} = \sum_i \lambda_iv_i w_i^T$, 
where $\Lambda$ is a diagonal matrix containing the eigenvalues $\lambda_i$, the columns $v_i$ of V represent the associated right eigenvectors, and the left eigenvectors $w_i$ are given by the columns of $(V^{-1})^T$.  Using the definition of the success rate and Eq.~\ref{Eq:FinalState}, we have
\begin{equation}\label{Eq:MeasurementSuccess}
\frac{\langle \mbox{correct}, \mbox{final}(k)\rangle}{\langle \mbox{postselect}, \mbox{final}(k)\rangle}
=
\frac{\sum_i \lambda_i^{k-1}\langle \mbox{correct}, M(m)Gv_i\rangle\cdot
\langle w_i, \mbox{init}\rangle}{\sum_i \lambda_i^{k-1}\langle \mbox{postselect}, M(m)Gv_i\rangle\cdot
\langle w_i, \mbox{init}\rangle}
\end{equation}
The largest eigenvalue of $M(m')G$ is $\lambda_0 = 1$ with corresponding eigenvector $w_0 = (0,0,0,0,1)^T$.
Observe, however, that we can completely ignore this term in Eq.~\ref{Eq:MeasurementSuccess}, since $\langle w_0, \mbox{init}\rangle = 0$. As the number of checks $k$ goes towards infinity, the only remaining $\lambda_i$ terms of relevance are those that match the second largest eigenvalue $\lambda_{\mathrm{mid}}$. Denoting the indices $i$ for which $\lambda_i = \lambda_{\mathrm{mid}}$ and observing that the scalar term $\lambda_{\mathrm{mid}}^{k-1}$ appears  in both the enumerator and the denominator, we find that
\begin{equation}\label{Eq:AsymptoticMeasurementSuccess}
\lim_{k\to\infty}
\frac{\langle \mbox{correct}, \mbox{final}(k)\rangle}{\langle \mbox{postselect}, \mbox{final}(k)\rangle}
=
\frac{{\sum_{i \in\mathcal{I}}} \langle \mbox{correct}, M(m)G v_i\rangle\cdot
\langle w_i, \mbox{init}\rangle}%
{\sum_{i\in\mathcal{I}}
\langle \mbox{postselect}, M(m)Gv_i\rangle\cdot
\langle w_i, \mbox{init}\rangle}.
\end{equation}
Keep in mind, however, that the asymptotic post-selection rate will go to zero with the number of checks, unless the measurement error rate $m$ and $\cnotgate$ gate error rates $g_c$ and $g_t$ are all zero.

\subsubsection{Simulation of measurement checks}

\begin{figure}[!t]
\centering
\begin{tabular}{cc}
\includegraphics[width=0.35\textwidth]{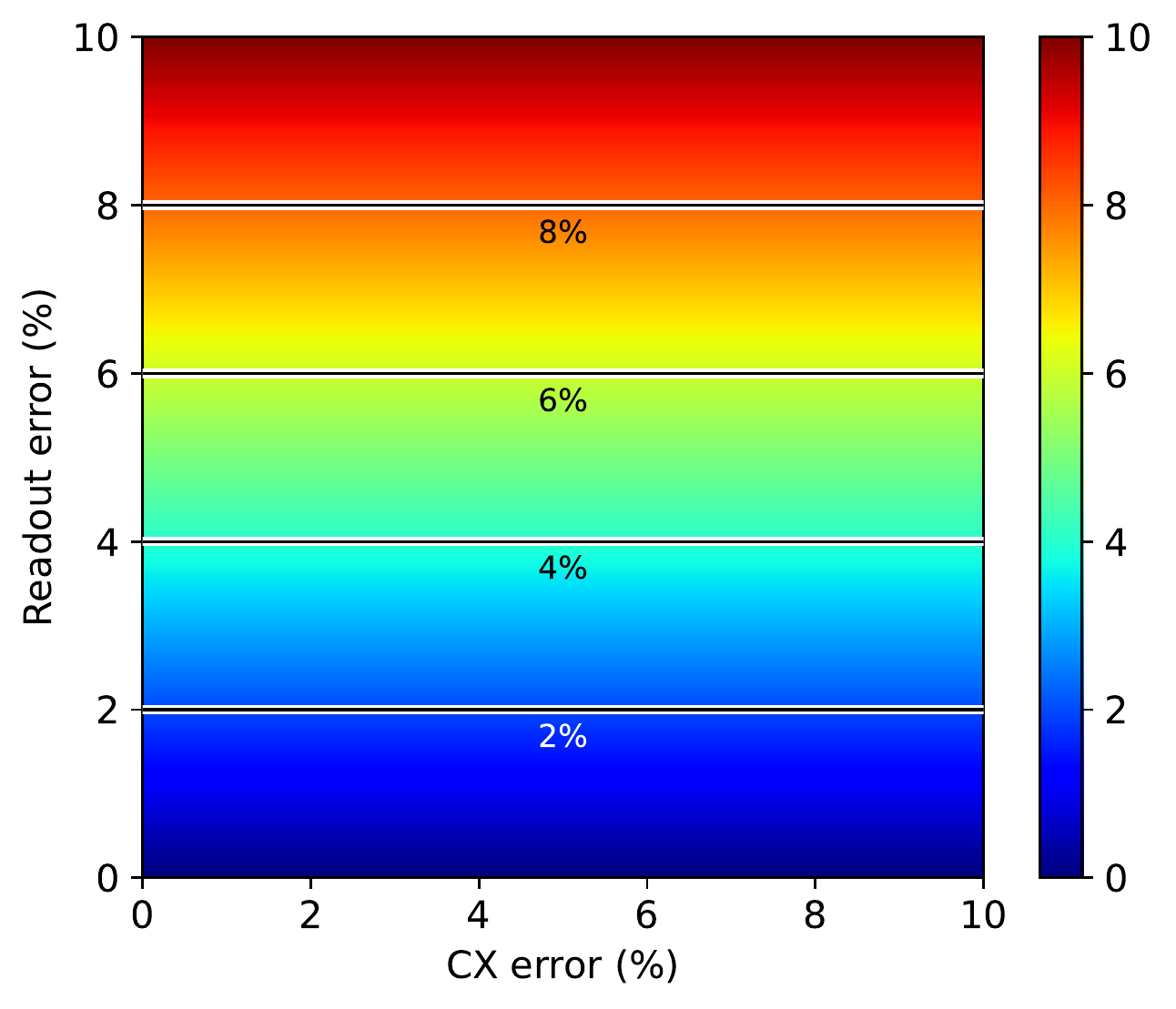}&
\includegraphics[width=0.35\textwidth]{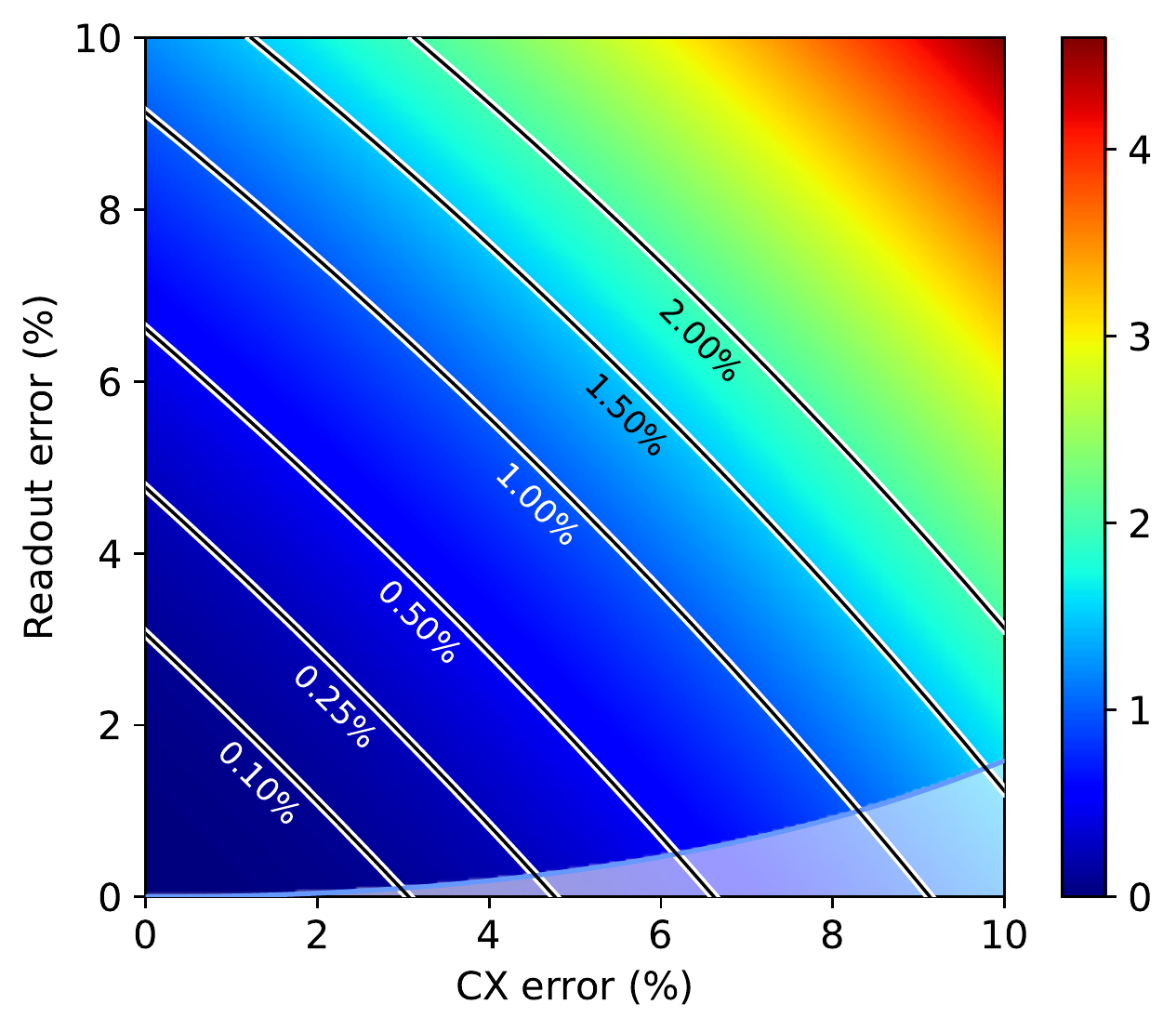}\\
({\bf{a}}) Direct readout & ({\bf{b}}) Single check\\[8pt]
\includegraphics[width=0.35\textwidth]{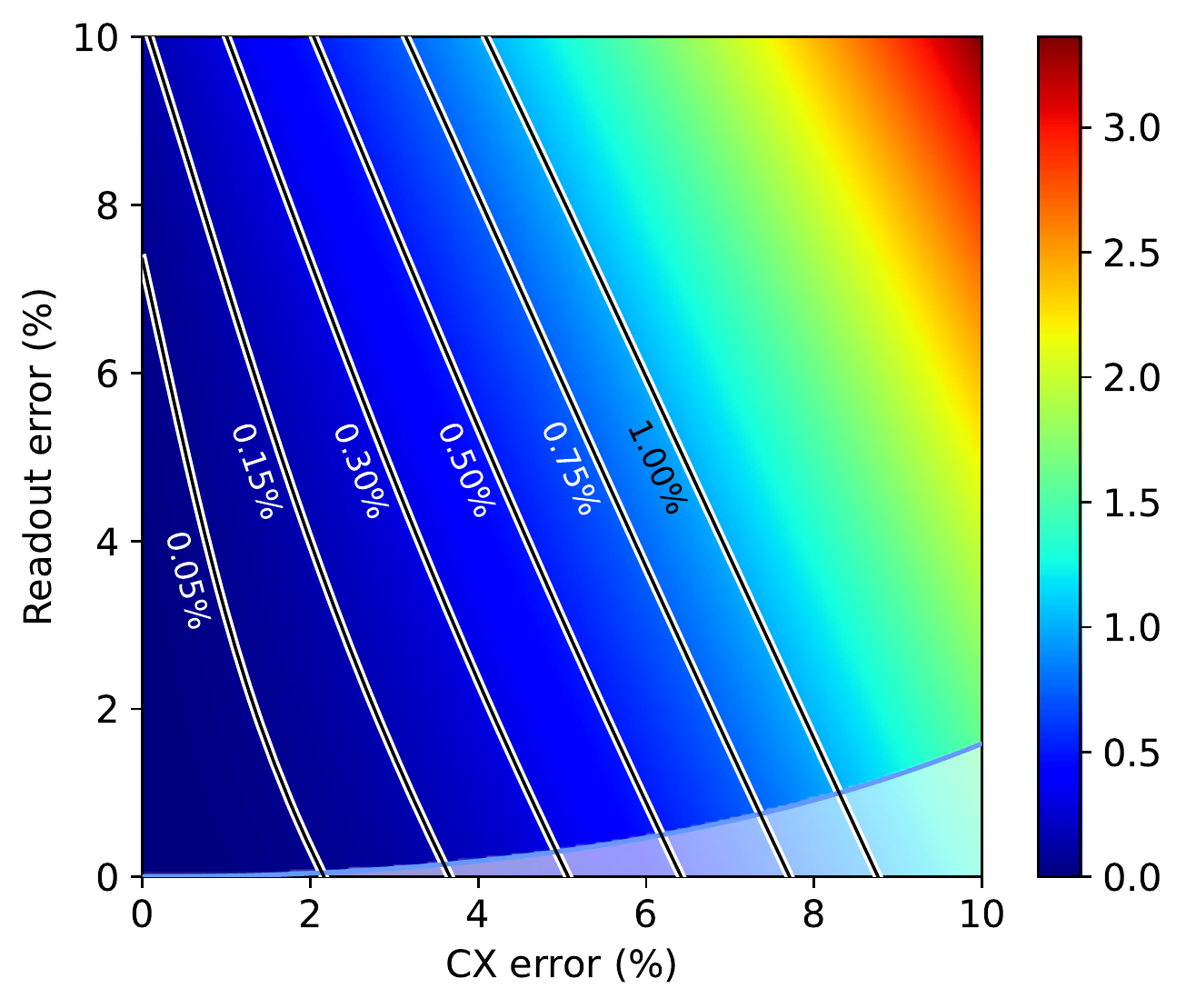}&
\includegraphics[width=0.35\textwidth]{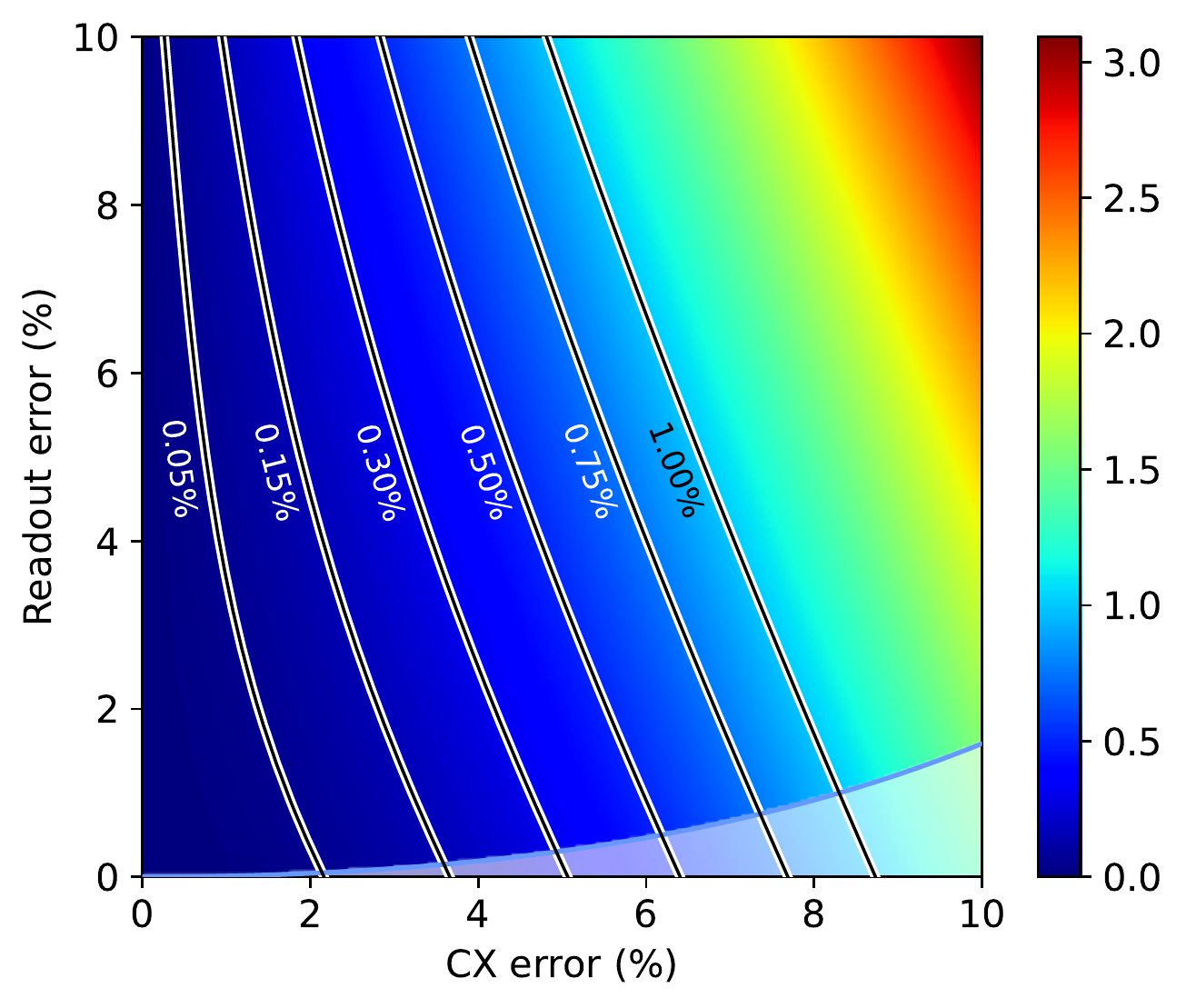}\\
({\bf{c}}) Two checks & ({\bf{d}}) Three checks
\end{tabular}
\caption{Logical measurement error rates as a function of both {\sc{cnot}} and physical measurement error rates along with iso-contours indicating a fixed logical error rate. The curves and scale of the error bar changes as the number of checks is increased from zero in (a) to three in (d). The saturated regions in the bottom right of plots (b)--(d) indicate parameter combinations for which the checked result is worse than direct measurement.}\label{Fig:PracticalReadout}
\end{figure}

The output distribution of the circuit in Figure~\ref{Fig:RepeatedReadoutSims}(a) is easily sampled using numerical simulation. This allows us to evaluate the performance of the measurement checks in various settings. We first validate the Markov model used to predict the error rates. For this, we fix the physical readout error $m$ to $30\%$ and use {\sc{cnot}} gates with errors $g_c{=}g_t$ chosen from the set $\{5\%$, $15\%$, $25\%\}$. We generate $10^6$ samples for different numbers of checks and plot the computed logical measurement error rates as dots in Figure~\ref{Fig:RepeatedReadoutSims}(b). For small numbers of checks, the sampled values closely match those generated by the model, indicated by solid lines. As the number of checks increases, the post-selection rate decreases, leading to larger variations in the sampled values. The asymptotic error rates are shown as horizontal dotted lines. For comparison, we also show the logical measurement error rates obtained using majority voting of the measurement for even and odd numbers of checks, shown as asterisks connected by light dashed lines.
For odd numbers of checks, we reject any samples where the number of 0 and 1 bits in the measurement matches. When the number of checks is zero or one, the two schemes are equivalent. For larger numbers of checks, we see that the logical measurement error for majority voting is significantly higher than that obtained using the more stringent requirement that all measurement bits match, as used in the Pauli-check approach. (Similar conclusions regarding the difference in majority and unanimous decoding were found in~\cite{PhysRevA.105.012419}.) Figure~\ref{Fig:RepeatedReadoutSims}(c) shows the asymptotic error rate as a function of the {\sc{cnot}} error rates $g_c{=}g_t$ for three initial values of $m$. When $m < \half$, the asymptotic error rate goes towards zero as the {\sc{cnot}} error decreases. When $w=\half$, the logical error rate remains at a half regardless of the {\sc{cnot}} error rate. The same applies irrespective of the measurement error rate when the {\sc{cnot}} error rate is one half.
The asymptotic readout error rate for $w > \half$  goes towards one with decreasing {\sc{cnot}} error rate.

\begin{figure}
\centering
\begin{tabular}{cc}
\includegraphics[width=0.425\textwidth]{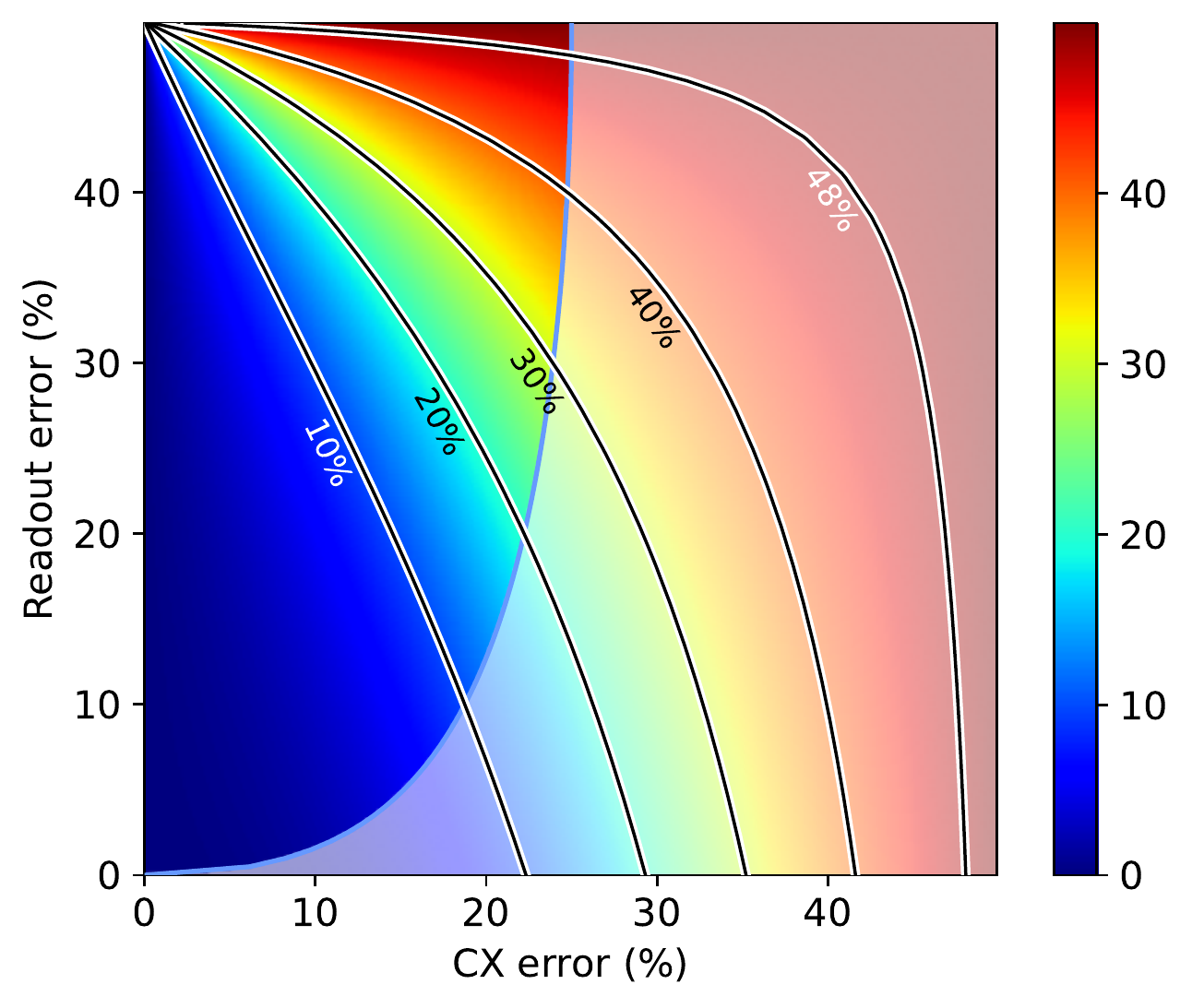}&
\includegraphics[width=0.425\textwidth]{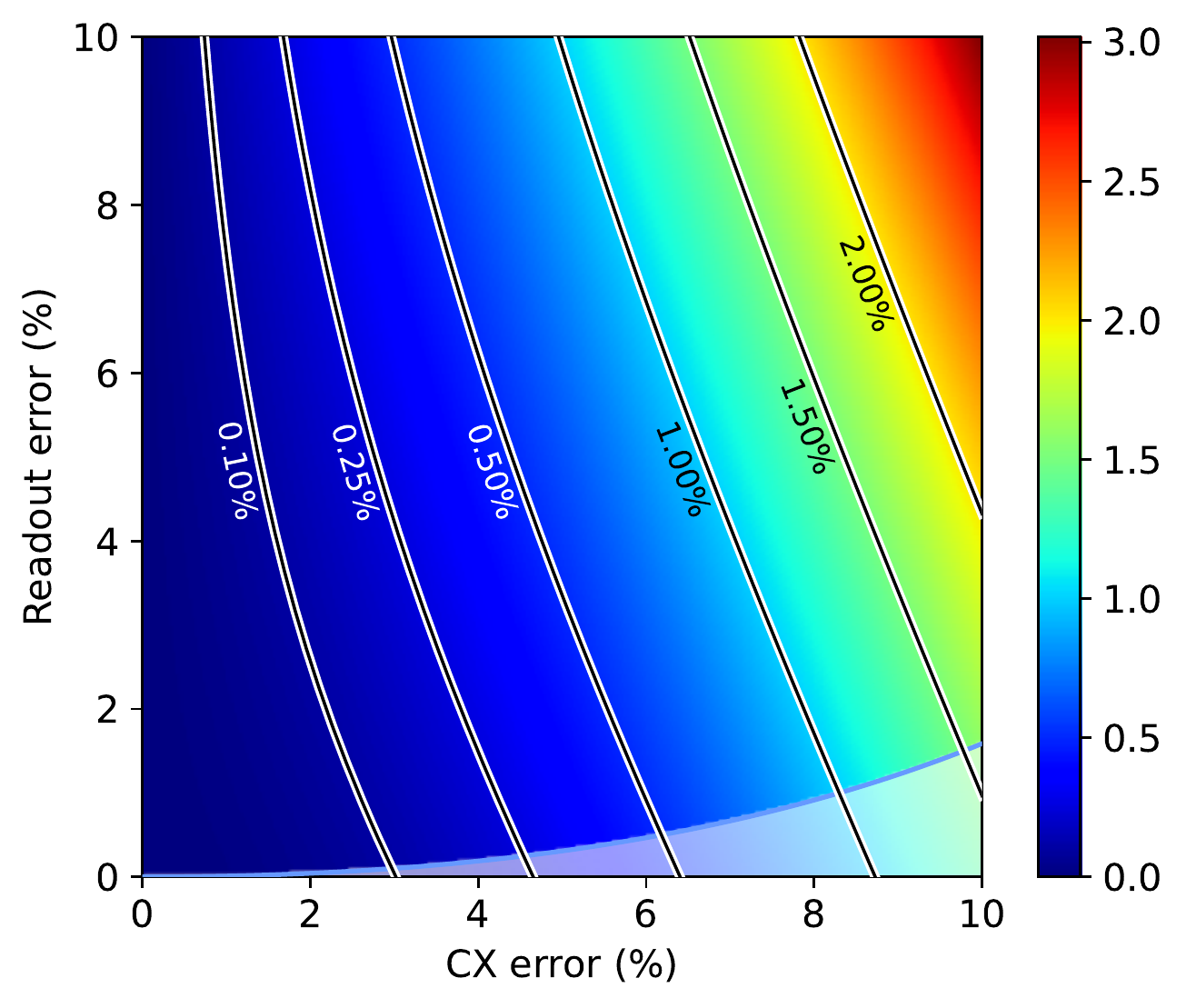}
\end{tabular}
\caption{Asymptotic logical readout-error rate after mitigation using a repetition code implemented with a chain of {\sc{cnot}} gates of length going towards infinity (in the limit, the post-selection rate will be zero). The readout errors are given as a bit-flip channel with equal probability of transitioning from 0 to 1 and vice versa. The {\sc{cnot}} error rates are the $g_c$ and $g_t$ values, each indicating a bit-flip probability. Measurements are accepted only when all bits match. The saturated region indicates parameter combinations where the asymptotic result is worse than direct measurement without checks.}\label{Fig:AsymptoticReadout}
\end{figure}

In Figure~\ref{Fig:PracticalReadout} we show the logical measurement error rates obtained for different combinations of {\sc{cnot}} and measurement error rates when using up to three checks. Compared with the asymptotic results in 
Figure~\ref{Fig:AsymptoticReadout} we see that the logical error rate obtained using a limited number of checks quickly approaches the asymptotic value. Contour lines in the figures show parameter combinations resulting in the same logical error rate. The bottom-right regions of the plots with saturated colors indicate parameter combinations for which the logical error rate exceeds the physical error rate. In those regions, there is clearly no advantage in using Pauli checks as they will only deteriorate the measurement accuracy.

\section{Non-asymptotic performance}\label{Sec:NonAsymptotic}

The theoretical model derived in Section~\ref{Sec:AsympototicCPC} allows us to evaluate the asymptotic performance of CPC under simplifying assumptions. In order to evaluate the performance in a more practical setting, the model requires a number of changes. First, we need to consider noise on gates other than {\sc{cnot}} gates and replace the uniform depolarizing noise with more general Pauli channels that are specific to each gate. Moreover, given that check and flag qubits in two-sided checks are generally idle for at least the duration of the payload circuit, errors due to thermal relaxation can be substantial and should therefore be included as well, along with readout errors. Second, in order to evaluate the performance of CPC with additional flag qubits, we need to extend the Markov model and keep track of the exact Pauli term on the check qubits. Finally, instead of assuming the payload error to be known in advance, we need to estimate or bound the error rates in a tractable manner. In this section, we discuss these extensions along with other techniques needed to obtain more accurate performance estimates of CPC in the non-asymptotic regime.

\subsection{Circuit preparation}

When preparing a quantum circuit for execution on a quantum processor we need to make sure that all gates are supported. Typically, the processor only provides a limited set of elementary gate types such as {\sc{cnot}} and {\sc{rz}} gates. In addition, certain two-qubit gates may be defined only on qubits that are physically connected, thereby further limiting the set of available gates. We assume that the payload circuit has been provided in such a way that it can run directly on the selected qubits. We therefore only need to make sure we appropriately implement the check and flag circuits. Here we assume that the controlled Pauli gates and swap operations already follow the qubit topology and that no additional swap operations are needed. Successive single-qubit operations can be combined and converted into an appropriate sequence of elementary single-qubit gates. For further analysis by the model described later in this section we need to keep track of which gates belong to which part of the circuit. Most notably, we keep track of the Pauli-swap blocks illustrated in Figure~\ref{Fig:LNN} and their inverses for the left and right checks. For each check, we keep track of the blocks used to implement it, and likewise for the flags. Gates have different durations, and the next step is to schedule the operations and assign a start and stop time for each gate. This allows us to identify qubit idle times which are padded with delays for simulation to capture thermal noise. When submitting to quantum processors, idle times can be replaced by appropriate dynamical decoupling~\cite{viola1998dynamical,ezzell2022survey} sequences.

\subsection{Tableau simulation}

Clifford circuits can be efficiently simulated using binary tableaus, where each row represents an $n$-qubit Pauli operator as a bit string of length $2n$. The initial $n$-qubit state $\ket{0}\bra{0}$ can be expressed in the Pauli-Z basis, resulting in a tableau of size $n\times 2n$. Applying gates such as {\sc{cnot}} and {\sc{s}} to the current state amounts to simple predetermined updates to the tableau that only affect few columns, thus enabling simulation of deep circuits. Sampling a single measurement outcome, however, may require updating the entire tableau and therefore forms a major computational expense, despite being polynomial in the number of qubits.
Fortunately, we can avoid simulating measurements if we assume that measurement errors can be modeled as the product of single-qubit symmetric bit-flip channels. Doing so ensures that the measurement noise is independent of the state, and allows us to model measurement errors as a Pauli-X noise channel just prior to an ideal measurement.
Instead of simulating an initial state and sampling at the end, we use the tableau representation to simulate the evolution of Pauli error strings as they change and accumulate throughout the circuit.
That is, each row in the tableau represents the noise term for a single run of the circuit, initialized to the identity operator and possibly updated by a Pauli noise channel representing state preparation. Per gate we then update the tableau as before, which conjugates the noise terms, effectively pushing them through the gate. Per row, we then sample a Pauli term from the noise channel associated with the gate and update the existing term using the exclusive-{\sc{or}} operator on the two bit strings. (If gate noise is modeled to occur prior to the gate we first sample and update the tableau before applying the ideal gate.) Once all gates have been applied, possibly including the noise channels for state preparation and measurement, we are left with the Pauli noise terms that apply just prior to measurement.  As for the outcomes themselves, we know that in the absence of noise the check and flag qubits will be zero. We would therefore measure a one if and only if the corresponding Pauli error term is X or Y. In the two-sided Pauli-check scheme we assume no measurements are made on the data qubits, and therefore only need to know whether an error occurred on these qubits or not. That is, we only need to check whether the Pauli string has identity terms on the data qubits. For the one-sided Pauli check we ignore the Pauli-Z component of the errors and perform the right-check classically based only on the Pauli-X component.

\subsection{Performance model}\label{Sec:PerformanceModel}

With the increased complexity of the noise model and the addition of timing information to capture thermal relaxation, it is no longer feasible to find a closed-form expression for the logical error and post-selection rates. Nevertheless, we can still obtain estimates of these rates by numerically evaluating an extended version of the Markov model used in Section~\ref{Sec:MarkovModel}. The global state vector is modified to include the probabilities of four states that indicate whether or not there is an error on the data qubits, and whether or not a check or flag was raised so far. This state is updated one (flagged) check at a time, and is initialized based on the error probability of the payload circuit. We no longer assume this probability as given and will derive upper and lower bounds in the next section that can be evaluated for a given payload circuit.

For a single step in the Markov model, we consider the probability of introducing different Pauli error terms on the check and flag qubits for a single check and whether the check introduced any new error on the data qubits. After taking into account any symmetrized readout errors, we determine if the check was raised by checking if the Pauli terms contain any X or Y components. The probabilities differ in case any existing errors are present on the data qubits, and we, therefore, evaluate the probabilities both with and without existing errors. Given these two vectors of probabilities, it is straightforward to update the global state vector.

We separately consider the left and right sides of each check and will refer to these as segments. For each segment, we first evaluate the probability of having an I, X, Y, or Z Pauli term on the check qubit, combined with a state that indicates whether the segment introduced any error term on the data qubits. This gives a state vector with eight probabilities. For two-sided checks, we evaluate the probabilities in the state vector by pushing all errors toward the beginning of the left segment and towards the end of the right segment. Logically, we can then cancel both sides of the check and combine the two state vectors by combining the Pauli terms, assuming that any newly introduced error terms on the data qubits do not cancel. At this point, if we assume the data qubits have  existing errors, we introduce an I or an X term on the check qubit with equal probability, reflecting the assumption that any error on the data qubits (anti)commutes with the check with probability one half. This is conveniently implemented by averaging the I and X and the Y and Z probabilities.

When flag qubits are present, we extend the Pauli terms by an identity term on the flag qubit, and extend the model to cover all 16 two-qubit Pauli terms, for a total of 32 probabilities. For all-to-all connectivity, we can easily deal with flag qubits, since there is never any interaction between qubits associated with different Pauli checks. For the linear nearest neighbor setting, however, this does not generally hold. For instance, in the example circuit shown in Figure~\ref{Fig:LNN}(b), the right-most \swapgate\ gate acts on the check qubit $C_1$ for the first check, and the flag qubit $F_3$ for the third check, thus breaking the Markovian assumption. For two-sided checks with flags on the linear nearest neighbor topology, we therefore factor the noise for the \swapgate\ gates in the flag circuit by looking at the marginal probabilities of the Pauli terms on each of the two qubits separately. This effectively decouples the noise terms and restores Markovianity at the expense of modifying the noise channel.

\begin{figure}[!t]
\centering
\setlength{\tabcolsep}{12pt}
\begin{tabular}{ccccc}
\includegraphics[height=48pt]{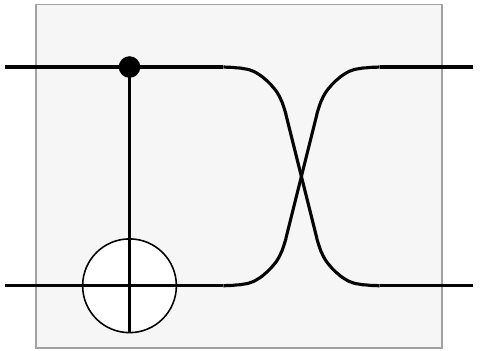}&
\includegraphics[height=48pt]{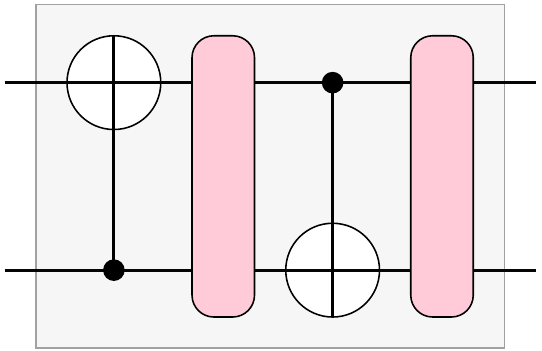}&
\includegraphics[height=48pt]{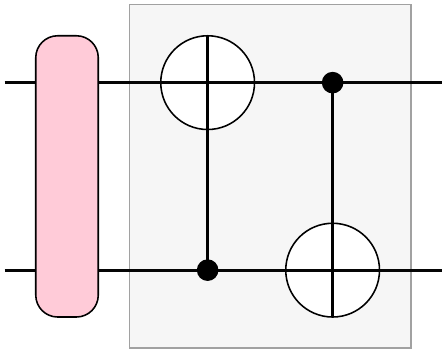}&
\includegraphics[height=48pt]{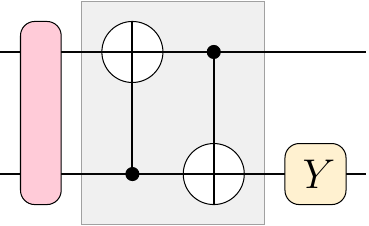}&
\includegraphics[height=48pt]{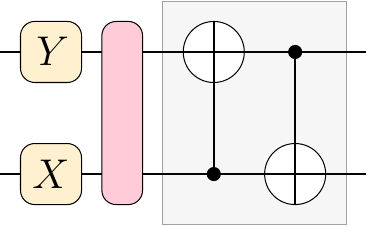}\\[5pt]
({\bf{a}}) & ({\bf{b}}) & ({\bf{c}})
& ({\bf{d}}) & ({\bf{e}})
\end{tabular}
\caption{Example of (a) a single Pauli-X element from a left check along with (b) its optimized noisy implementation. Gate noise (indicated by the rounded red boxes) can be pushed towards the beginning of the block (c). Any existing noise term on the check qubit (d) can be pushed through the gate block (e) for combination with the noise channel.}\label{Fig:XCheckBlock}
\end{figure}

For a given left or right check we evaluate the state vector using a Markov model that iterates over the individual Pauli elements that constitute the check. These elements are implemented using a single controlled-Pauli gate in the all-to-all case, and using the blocks shown in Fig.~\ref{Fig:LNN} in the case of linear nearest neighbors. Consider, for instance, a Pauli-X element from a left check, as shown in Fig.~\ref{Fig:XCheckBlock}(a). By combining the controlled-Pauli and \swapgate\ gates, the implementation would consists of two {\sc{cnot}} gates with noise, indicated by the red boxes in Fig.~\ref{Fig:XCheckBlock}(b). Since each check element is a quantum circuit on two qubits, we can explicitly evaluate the individual noise channels and push them through to the beginning or end of the element by conjugating the Pauli terms. We then combine the Pauli channels using convolution (efficiently implemented by element-wise multiplication of the Pauli fidelities) to obtain a single noise channel associated with the element (see Fig.~\ref{Fig:XCheckBlock}(c)).
With this we can consider the update process of the state vector. Each state has a given Pauli noise term on the check qubit, for instance, a Pauli-Y term in Fig.~\ref{Fig:XCheckBlock}(d). We push this through the element, as shown in  Fig.~\ref{Fig:XCheckBlock}(e) and can then determine the contributions to the next state vector. Conceptually, suppose the red noise channel yields a Pauli ZY term. Then, following the multiplication of YX and ZY, the overall noise term would be XZ. This has an X term on the check qubit, and given that the second term is not the identity, this introduces an error on the data qubits. In case the second term would be the identity,
we would simply maintain the current state of whether or not an error occurred. We can process all of the 8 elements of the current state this way to determine the updated state. Note that the updated state only records whether or not a Pauli error occurred on the data qubits; the exact terms are discarded and their individual probabilities are consolidated in the overall state vector.

For left-only checks, we compute the state vector by 
pushing noise towards the end of the checks, processing elements from left to right. Since left-only checks are only applicable when data qubits are measured at the end, we can disregard Pauli-Z errors on those qubits at the end of the circuit. In order to account for this to some extent, we slightly modify the processing of individual elements when updating the state vector. Instead of just looking at the Pauli term that appears on the data qubit, we augment it with identity terms and push it through the payload circuit. Whenever the resulting term contains only I and Z terms we treat it as not introducing any new error onto the data qubits. For simplicity, and indeed tractability, we assume that X or Y errors on the data qubits do not cancel or combine to Pauli Z terms. Finally, note that readout errors on the data qubits are included in the overall error probability of the payload circuit and that left-only checks are never combined with flag qubits.

\subsection{Bounds on the payload circuit error probability}\label{Sec:PayloadErrorBounds}

The logical error and post-selection probabilities depend on the error rate of the payload circuit. For a given Clifford payload circuit, it is therefore important that we can estimate or at least bound this error rate using the available information on the individual the gate errors. Assuming that gate errors are Pauli channels, it is in principle possible to propagate all gate errors to the end of the circuit and form the overall Pauli noise channel affecting the payload. However, this approach scales exponentially in the number of qubits and is therefore impractical for all but the smallest payload circuits. What we can do is  compute aggregated error channels $C_i$ for successive gates on small subsets of qubits and use these combined channels in further calculations.
We characterize each channel $C_i$ by two scalar values. The first, $s_i$, denotes the `success probability' of the channel, namely the probability that no error occurred. The second, $\alpha_i$, denotes  the largest coefficient of a non-identity Pauli coefficient in
the channel. Both values are invariant under conjugation of the noise channel with Clifford operations, and without loss of generality, we can therefore assume that all intermediate noise channels occur at the end of the circuit. We now define an aggregated noise channel $\mathcal{C}_{\ell}$ that combines channels $C_1$ through $C_{\ell}$. The overall success probability $S_{\ell}$ is now defined as the probability that the aggregated channel appears noiseless; either because none of the sub-channels had any noise, or because noise terms canceled. We obtain a lower bound $L_{\ell}$ on the overall success probability if we discount error cancellation and require that all sub-channels are error-free:
\[
L_{\ell} = \prod_{i=1}^{\ell} s_i.
\]
For an upper bound we need to consider a more optimistic scenario where errors cancel. For a pair of Pauli channels with coefficients $u_i$ and $v_i$ respectively for Pauli $P_i$, the total probability of canceling errors is given by their inner product $\langle u,v\rangle = \sum_i u_iv_i$ since only matching coefficients cancel. It is then natural to ask which channel coefficients $v$ maximize this probability for a given $u$. In case the noise-free probabilities $u_0$ and $v_0$ are fixed, and denoting by $\bar{u}$ and $\bar{v}$ the non-identity Pauli coefficients, this amounts to solving
\[
\mathop{\mathrm{maximize}}_{\bar{v} \geq 0}\quad \langle \bar{u},\bar{v}\rangle \quad \mbox{subject to}\quad \Vert \bar{v}\Vert_1 = 1-v_0
\]
This expression is closely related to the definition of the dual norm  of the one norm, and the optimum is given by $\Vert \bar{u}\Vert_{\infty} = (1-v_0)\cdot \max_i\{ \bar{u}_i\}$. Applying this to the combination of channels $C_1$ with $C_2$, we obtain the upper bound $s_1s_2 + \alpha_2(1-s_1)$, which can be rewritten as $(s_2 - \alpha_2)s_1 + \alpha_2$. By repeatedly adding single channels to previously combined channels, we can obtain an upper bound on the success probability for $\mathcal{C}$.
In case $s_i \geq \alpha_i$, which is always true if the error rate $1-s_i \leq 1/2$, we can define the upper bound
\[
U_{i} := (s_i-\alpha_i)U_{i-1} + \alpha_i,
\]
starting with $U_0 = 1$.
Since Pauli channels commute, we permute the order in which we specify the channels. This does not affect the actual noise channel, but can have an effect on the the upper bound $U_k$. As such, it is possible to further lower the upper bound by carefully selecting the channel order.

In the special case where all channels satisfy $s_i=s$ and $\alpha_i=\alpha$, we have
\begin{equation}\label{Eq:UpperBound1}
U_{k} = (s-\alpha)^{k} + \alpha\sum_{i=0}^{k-1}(s-\alpha)^i
= (s-\alpha)^{k} + \alpha\frac{1 - (s-\alpha)^k}{1 - (s-\alpha)}.
\end{equation}
 For a payload circuit consisting of $k$ two-qubit gates, each affected by a depolarizing channel with error probability $\epsilon$, we have $s=1-\epsilon$ and $\alpha = \epsilon/15$ and it therefore follows from Eq.~\ref{Eq:UpperBound1} that
\begin{align}
U_k
&= (1 - (\epsilon+\alpha))^k + \frac{\alpha}{\epsilon+\alpha}(1 - (1 - (\epsilon+\alpha))^k)\notag\\
&= \frac{\epsilon}{\epsilon+\alpha}(1 - (\epsilon+\alpha))^k + \frac{\alpha}{\epsilon+\alpha}
 = \frac{15}{16}\left(1 + \frac{16\epsilon}{15}\right)^k + \frac{1}{16}.
 \label{Eq:DepolarizingSuccessBound}
\end{align}

\begin{figure}
\centering
\begin{tabular}{cc}
\includegraphics[width=0.4\textwidth]{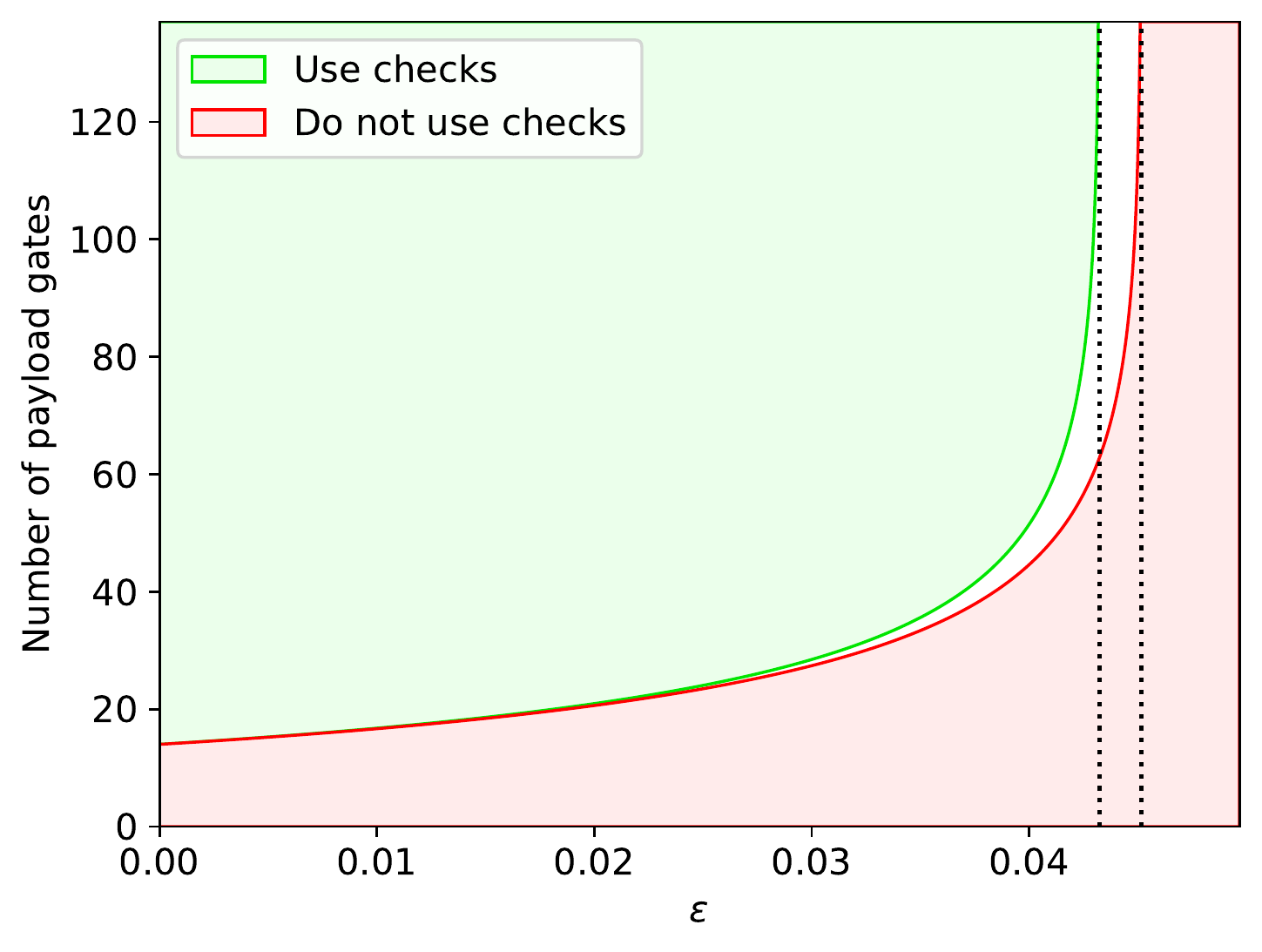}&
\includegraphics[width=0.4\textwidth]{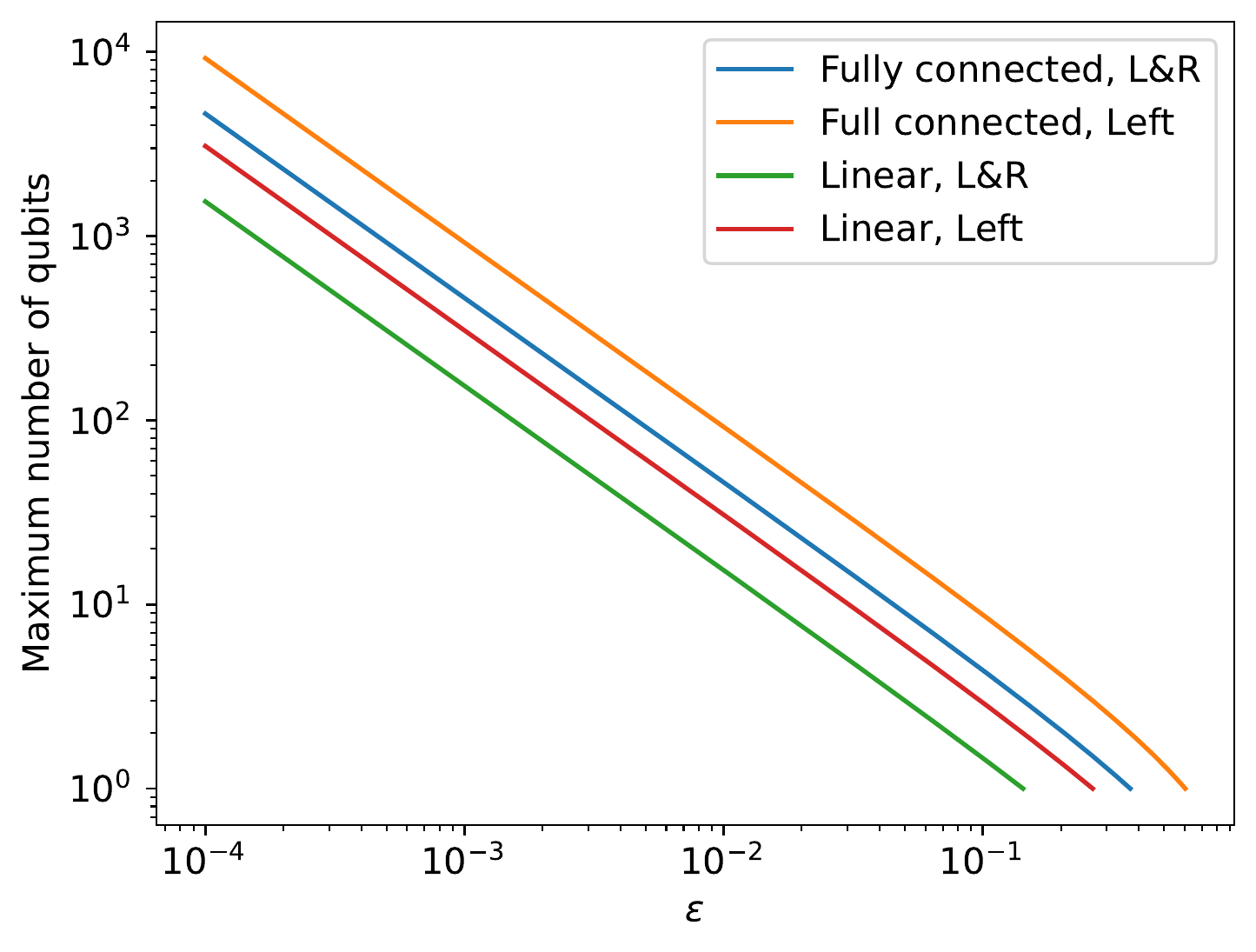}\\
({\bf{a}}) & ({\bf{b}})
\end{tabular}
\caption{Assuming depolarizing channels with error rate $\epsilon$ for all two-qubit gates, plot (a) shows the regions, in terms of the number of two-qubit gates in a 10-qubit  payload circuit, for which two-sided Pauli checks with all-to-all connectivity will improve or deteriorate performance, along with critical threshold based on upper and lower bounds on the payload error rate. For the region in between the two boundaries, performance improvements depend on the exact error rate of the payload circuit. The left dotted line corresponds to the limit of $P_{\min}$ as the number of gates goes to infinity, and the right dotted line gives the value of $\epsilon$ for which $t_{ok}$ reaches 1/2. Plot (b) illustrates the maximum number of qubits in the payload circuit for which $L_k < 1/2$. This is the largest number of qubits for which $t_{ok}$ may not reach the critical value of 1/2. For a guaranteed performance we could find the largest number of qubits for which the number of check gates $k$ satisfies $t_{ok} \leq U_k < 1/2$.
}\label{Fig:Bounds}
\end{figure}

Given the final upper and lower bounds on the success probability, $L_k$ and $U_k$, we immediately obtain bounds on the error rate by setting $P_{\min} := 1 - U_k$ and $P_{\max} := 1 - L_k$. In the depolarizing case, whenever the number of gates $k$ is large enough for $P_{\min}(k)$ to exceed the critical error rate $P_{\mathrm{critical}}$, we can improve the logical error rate by applying Pauli checks. On the other hand, when $k$ is sufficiently small such that  $P_{\max}(k) < P_{\mathrm{critical}}$ it does not make sense to apply Pauli checks, since doing so increases the logical error rate. Earlier we also noted that the asymptotic logical error rate goes to one as $t_{\mathrm{ok}} \leq 1/2$. The bounds derived in this section can be used to bound $t_{\mathrm{ok}}$ based on the number of gates $k$ used to implement the check (see also Table~\ref{Table:KValues}). For a given error rate $\epsilon$, the lower bounds derived in this section allow us to determine the maximum number of qubits in the payload circuit for which $t_{\mathrm{ok}} < 1/2$ may still hold. If the lower bound attains or exceeds one-half, we are guaranteed that the logical error rate increases to one. We illustrate the various bounds in Fig.~\ref{Fig:Bounds}.

When analyzing the payload errors in the context of left-only checks, we can disregard all Pauli-Z components in the error. In this case, we can still follow the same derivation as above, albeit with some minor changes. In order to obtain the noise channels $C_i$ we still propagate errors to the end of a contiguous group of gates on a subset of the qubits, but now need to push all Pauli terms through the remaining gates to the end of the payload circuit. We then discard the Pauli-Z component of the channel terms and form a new channel consisting only of Pauli-X operators. In addition to this, we generally want to include one additional noise channel to model the readout errors. The combination of successive noise channels into aggregated upper and lower bounds remains unchanged.

\subsection{Simulation and model results}

Given the simulation and modeling tools developed earlier in this section, we can now compare the performance of coherent Pauli checks with and without flags, as well as the relative performance of one- and two-sided checks in terms of their logical error and post-selection rates. In the case of two-sided checks, the expanded noise model also allows us to study the effect of thermal relaxation during the idle times between the left and right checks. 
Simulation allows us to consider both all-to-all and linear nearest neighbor architectures, and enables
us to evaluate the accuracy of the model, which can be substantially faster than sampling, certainly in regimes where the post-selection rate is low.

\subsubsection{Noise modeling}

Before we can run the simulation, we need to construct a noise model that better captures noise in actual processors. For this, we start with the error model for IBM's 127-qubit superconducting quantum processor \ibmwashington, which is periodically updated and available through Qiskit Aer~\cite{qiskit}. Gate noise estimates are provided as local noise channels, which we simplify to Pauli channels based on the Pauli fidelities. Measurement errors are modeled as classical bit flips occurring after qubit measurements and represented by $2\times 2$ stochastic matrices, one for each qubit. For simplicity, we average the off-diagonal elements and renormalize to obtain symmetric bit-flip channels. Finally, following~\cite{sarvepalli2009asymmetric}, we model the combined effect of amplitude damping and dephasing on idle qubits as single-qubit Pauli channels with coefficients
\[
p_x(t) = p_y(t) = \sfrac{1}{4}\big(1 - e^{-t / T_1}\big),\quad
p_z(t) = \sfrac{1}{4}\big(1 + e^{-t/T_1}
- 2e^{-t/T_2}\big),
\quad
p_i(t) = 1 - (p_x + p_y + p_z)(t),
\]
where $T_1$ and $T_2$ represent the qubit thermal relaxation and dephasing times, and $t$ denotes the duration of the qubit idle time.
To control the strength of the noise, we introduce a scaling factor $\alpha$ for each of the noise types. For Pauli noise channels associated with gates, we change the channel such that each Pauli fidelity $f$ is mapped to $f^{\alpha}$. For $\alpha = 1/2$ this means that we halve the noise level in the sense that we need to apply the Pauli channel twice to obtain the original noise level. For measurement errors, we multiply the off-diagonal elements of the stochastic transition matrices by $\alpha$, followed by renormalization. Finally, the noise associated with qubit idle time is scaled by dividing both the $T_1$ and $T_2$ times by $\alpha$.

The qubits used in the linear nearest neighbor setting are easy to embed in the heavy-hex connectivity. This means that all gates and their associated noise channels are directly available. Given the heavy-hex connectivity of the processor and correspondingly, the restricted set of gates, we need to artificially extend the noise model if we are to simulate circuits that assume a fully-connected topology.
Given a target number of qubits, we sample gates, along with their noise model and duration, uniformly at random with replacement, from the gates defined on the selected qubit chain. The same is done for the $T_1$ and $T_2$ times and the measurement errors. For two-qubit gates we ensure that both orientations of the gate share the same noise and duration.

\subsubsection{All-to-all connectivity}

With the noise model in place, we are now in the position to compare the performance of the one- and two-sided Pauli checks with all-to-all qubit connectivity. As the payload, we randomly sample a 10-qubit Clifford payload circuit that is implemented over LNN~\cite{bravyi2021hadamard}. The circuit contains 274 {\sc{cnot}} gates and, considering only two-qubit gates, has a circuit depth of 88. With the {\sc{cnot}} noise scaled by a factor of 0.3 and ignoring idle time within the payload circuit, the simulated payload error rate using 10 million samples is around 81.95\% when considering all Paulis errors and 76.21\% when disregarding Pauli-Z errors. The corresponding bounds evaluated using the model given in Section~\ref{Sec:PayloadErrorBounds} are 79.82--81.96\% and 72.77--76.66\%, respectively.

We scale the noise on delay gates, which represent idle time, by factors 0.2, 0.6, and 1.0 and evaluate the performance using both simulation and the performance model described in Section~\ref{Sec:PerformanceModel}, with error rate set to the modeled upper bound. In both cases we sample up to 20 random Pauli checks in such a way that per set of checks, each qubit has a non-identity Pauli in at least one check. Using 10 million shots per problem instance, we obtain the results shown in Figure~\ref{Fig:Model003_10Q}.
With limited noise on qubit idle time, we see from
Figure~\ref{Fig:Model003_10Q}(a) that all three methods (two-sided checks with and without flags and one-sided checks) significantly reduce the logical error rate with added checks. Performance stabilizes around ten checks; using additional checks does not lower the logical error rate and only reduces the post-selection. Among the three methods, the one-sided check has both the lowest logical error rate and the highest post-selection rate.  A two-sided check with flags has a lower error rate compared to a two-sided check without flags, at the cost of a much lower post-selection rate. This trend becomes even clearer when increasing the noise on idle qubits in Figures~\ref{Fig:Model003_10Q}(b) and~\ref{Fig:Model003_10Q}(c). With increased noise between the two segments of the checks, the logical error rate increases and eventually causes the flags and checks to deteriorate performance rather than improve it. The performance of one-sided checks is largely unaffected by these changes in the noise, as a result of the limited idle time present in the circuits for one-sided checks. The modeled logical error and post-selection rates, indicated by the solid lines in Figure~\ref{Fig:Model003_10Q}, closely match those obtained using the simulation. The computational time for modeling is only a small fraction of the time needed for the simulation, especially for settings where the post-selection rate is low and many samples are required for accurate estimates.

\begin{figure}
\centering
\begin{tabular}{ccc}
\includegraphics[width=0.3\textwidth]{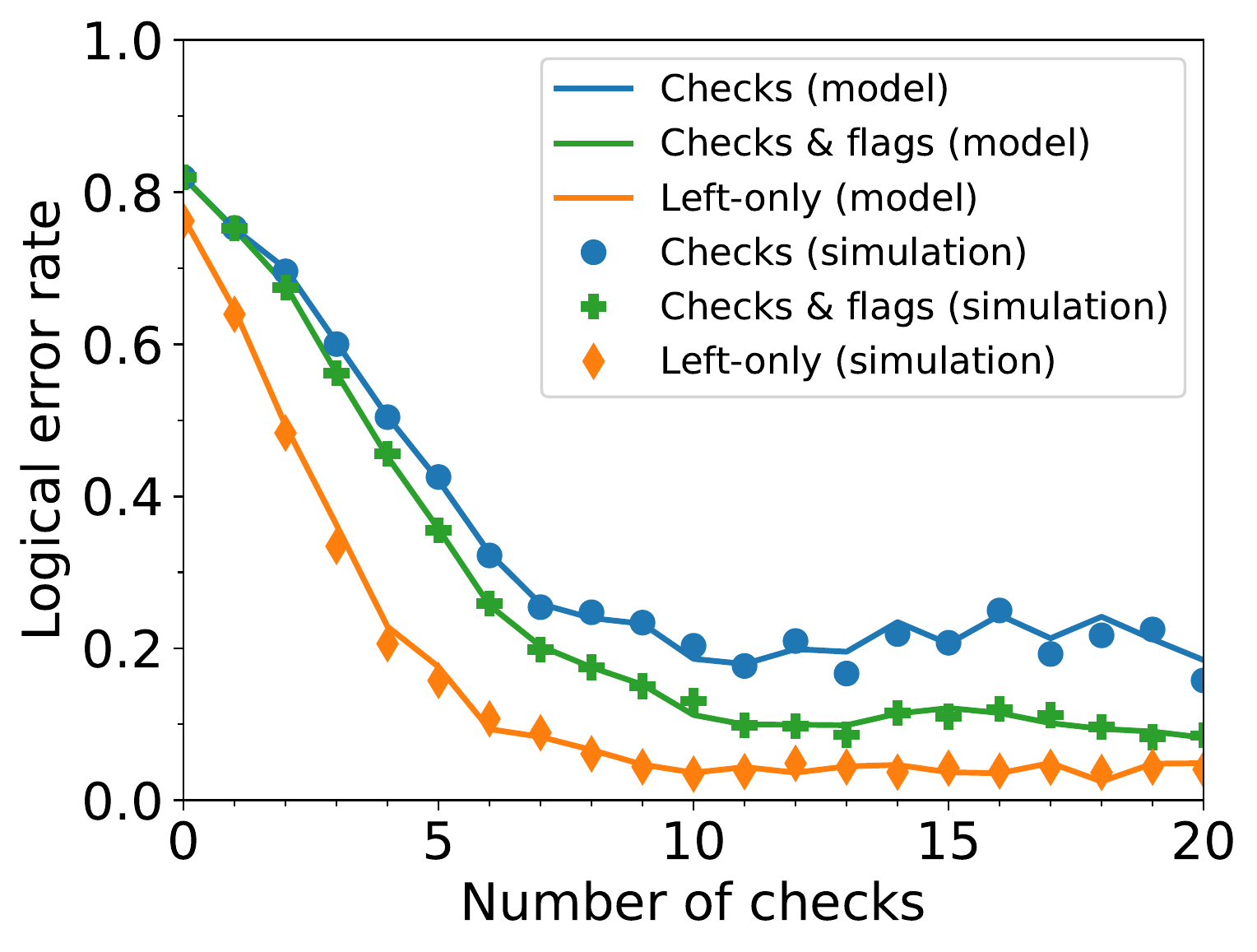}&
\includegraphics[width=0.3\textwidth]{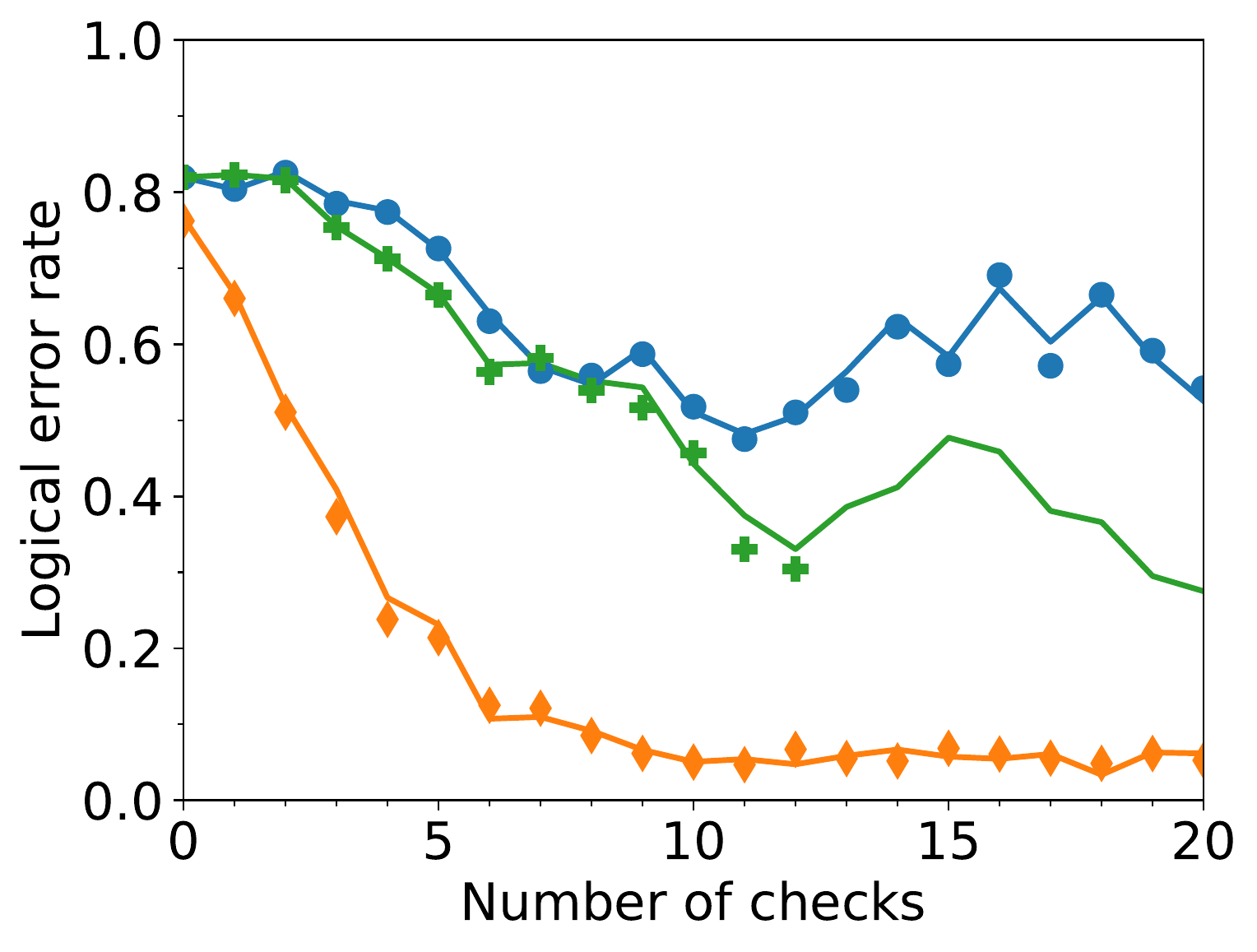}&
\includegraphics[width=0.3\textwidth]{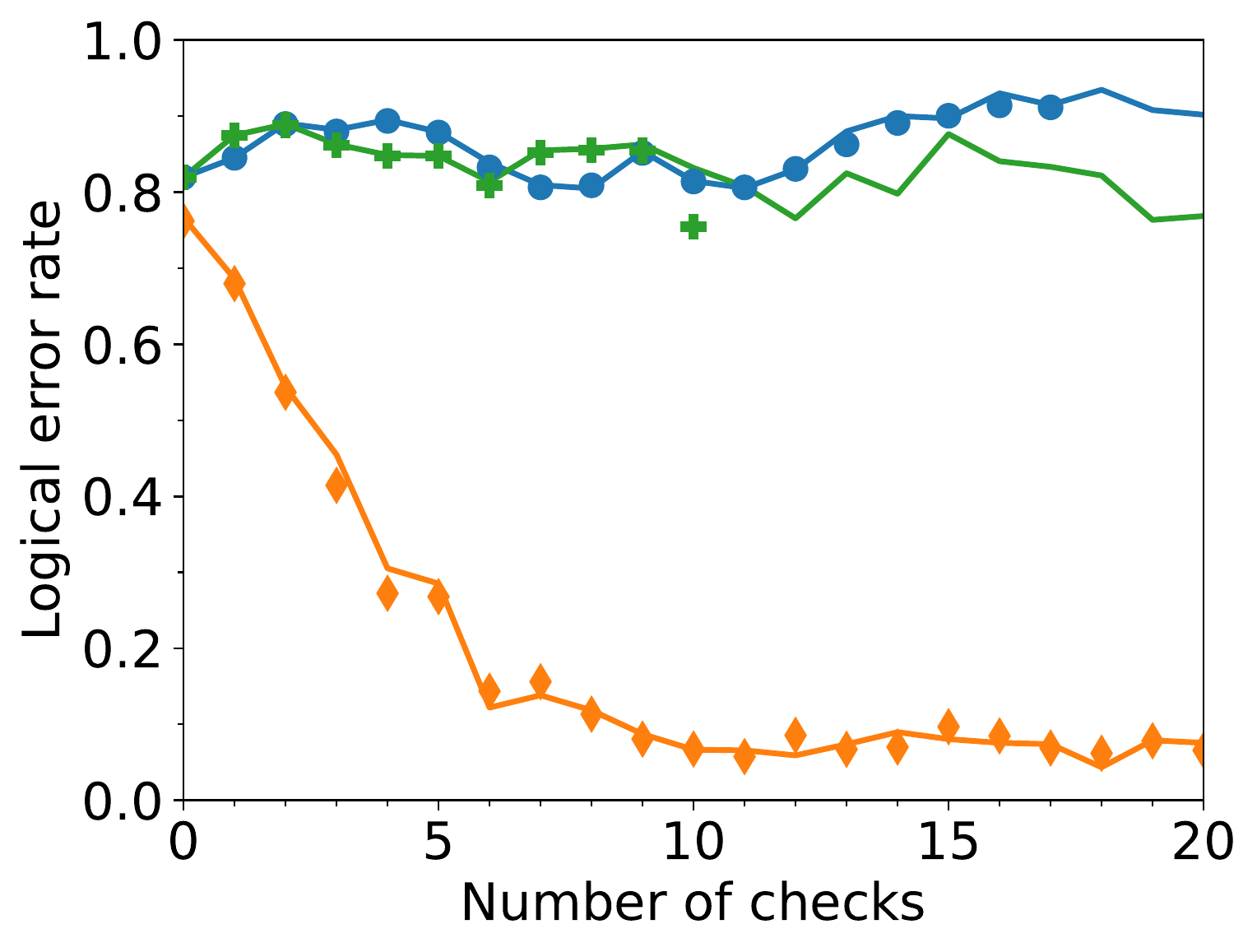}\\[3pt]
\includegraphics[width=0.3\textwidth]{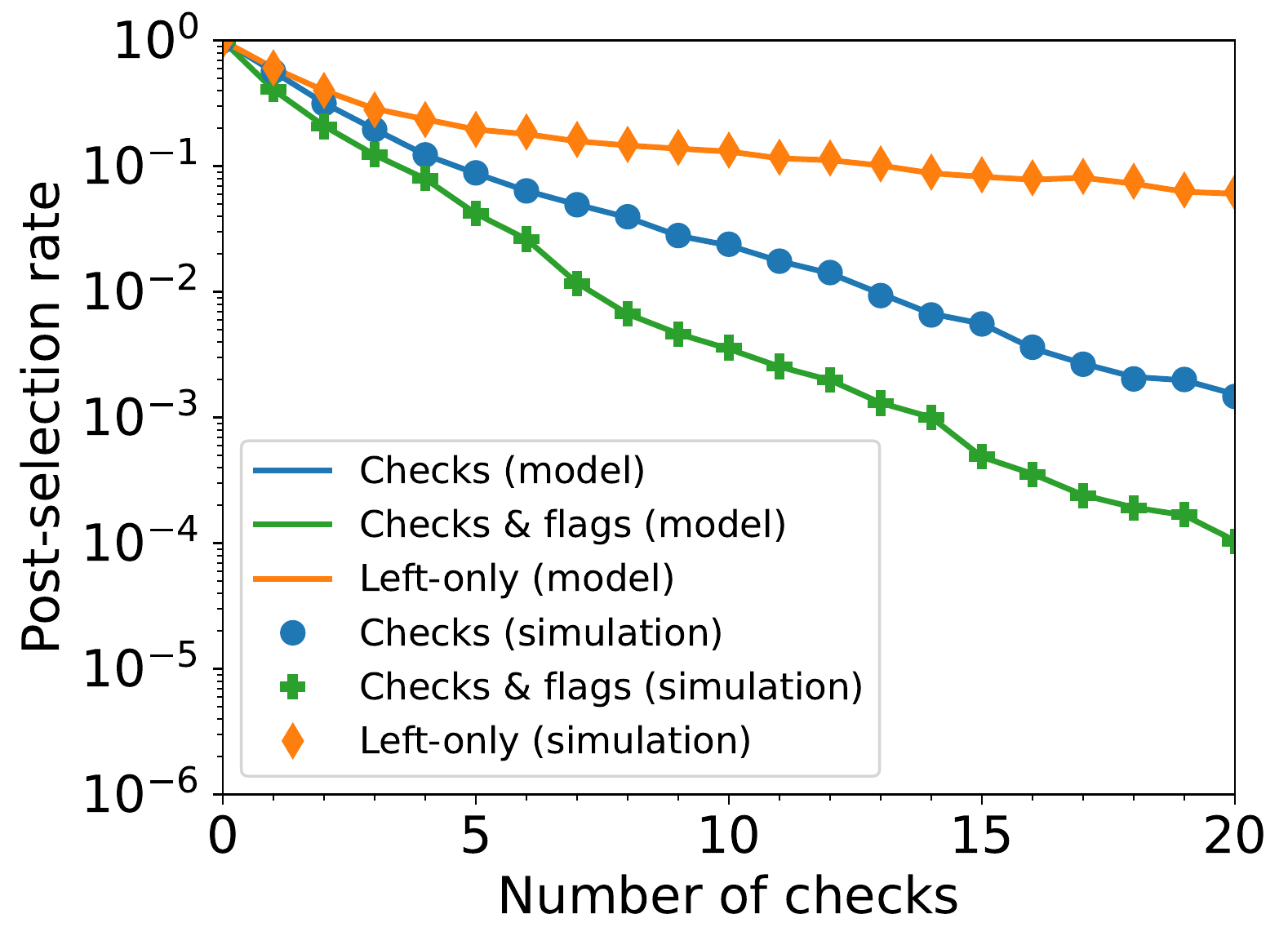}&
\includegraphics[width=0.3\textwidth]{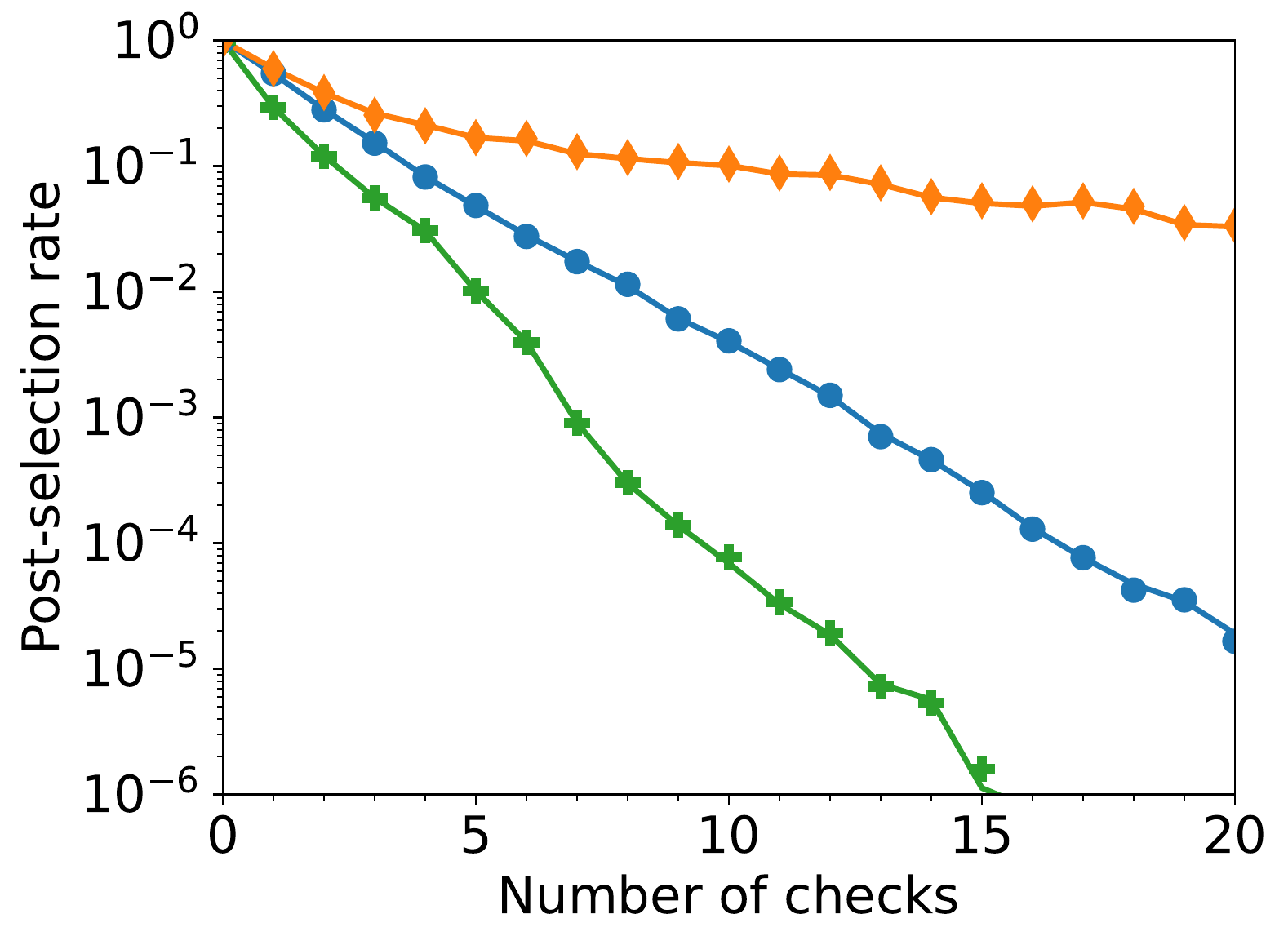}&
\includegraphics[width=0.3\textwidth]{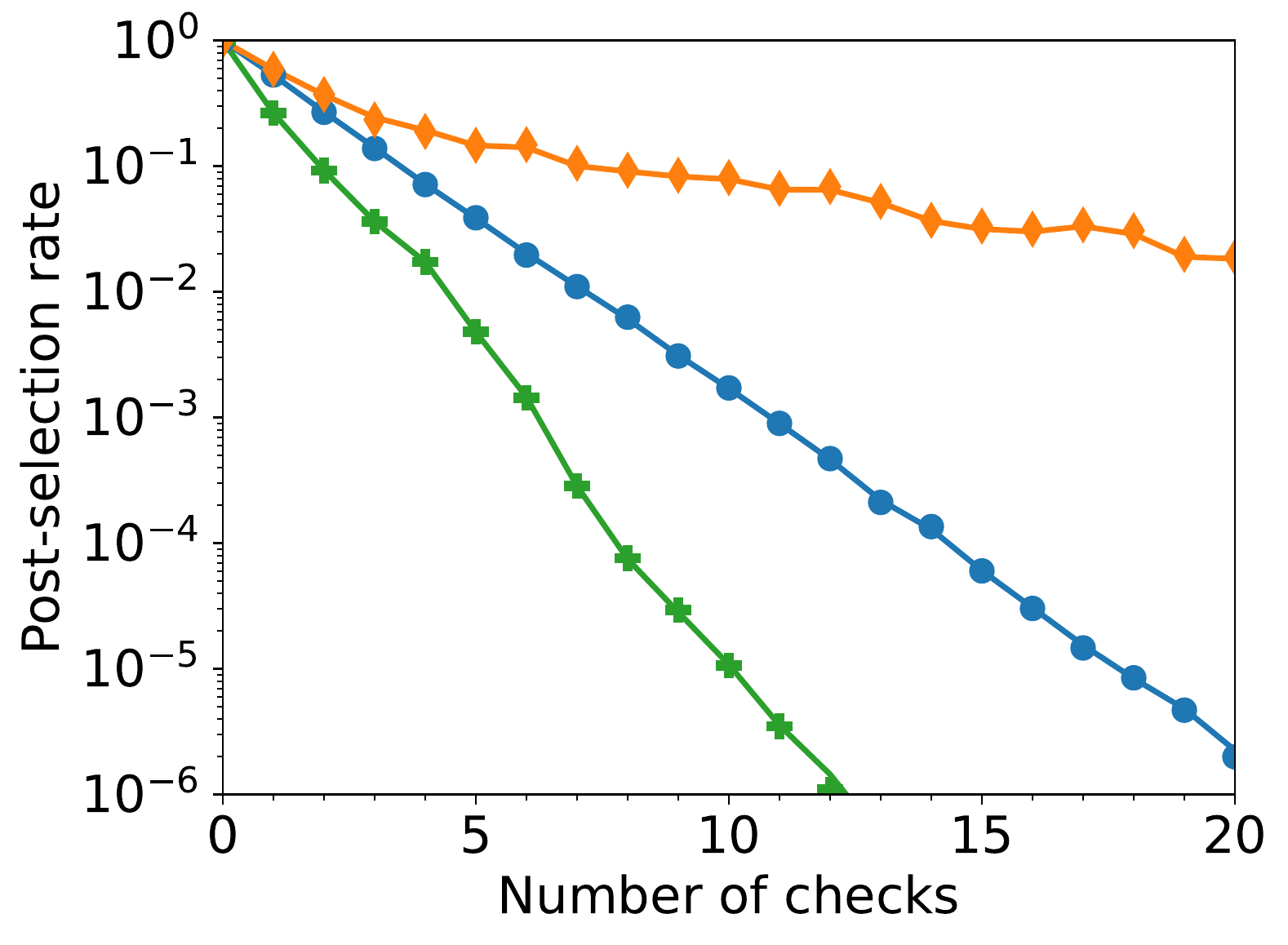}\\[3pt]
({\bf{a}}) & ({\bf{b}}) & ({\bf{c}})
\end{tabular}
\caption{The effects of scaling the $T_1$ and $T_2$ values by factors (a) 0.2, (b) 0.6, and (c) 1.0, on simulated and modeled logical error (top) and post-selection rates (bottom) for a randomly sampled Clifford payload circuit on 10 qubits with all-to-all connectivity. Simulation data points are based on 10 million samples each. Simulated logical error rates with fewer than ten post-selected samples are omitted.}\label{Fig:Model003_10Q}
\end{figure}

As a second setup for all-to-all connectivity we consider a randomly sampled 30-qubit Clifford operator. The operator assumes only LNN connectivity and the resulting circuit has 2897 {\sc{cnot}} gates and a {\sc{cnot}}-depth of 328. This depth could be optimized further using recent constructions described in~\cite{bravyi2021hadamard,maslov2022cnot}, but we did not pursue this.  Given the large number of qubits involved we scale down the {\sc{cnot}} and delay noise by a factor of 10, and multiply the readout measurement error by a factor of 0.3. This gives an estimated payload error rate of 99.63\%, which reduces to 98.94\% when ignoring all Pauli-Z errors. The corresponding bounds obtained from the model are 96.75--99.79\% and 92.93--99.37\%, respectively. We simulate up to 20 checks and plot the resulting logical error and post-selection rates in Figure~\ref{Fig:Model003_30Q} along with the modeled results. In this case, the modeled results are not as close to the simulated results as before. However, replacing the modeled upper bound on the payload error rate with the sampled one (zero checks) again yields quite a good agreement. Using 20 one-sided checks reduces the logical error rate from 98.95\% down to 4.44\%, with a post-selection rate of 0.23\%. Two-sided checks have a much lower post-selection rate, especially when flags are added. The best logical error rate obtained without flags is around 31\% when using 29 checks, and 33\% with 15 checks and flags. These numbers cannot be expected to be very accurate due to the low post-selection rates and should therefore be taken as rough estimates. For small numbers of checks, we see that adding flag qubits helps to reduce the logical error rate at the cost of a modest decrease in the post-selection rate.

\begin{figure}[!t]
\centering
\begin{tabular}{cc}
\includegraphics[width=0.45\textwidth]{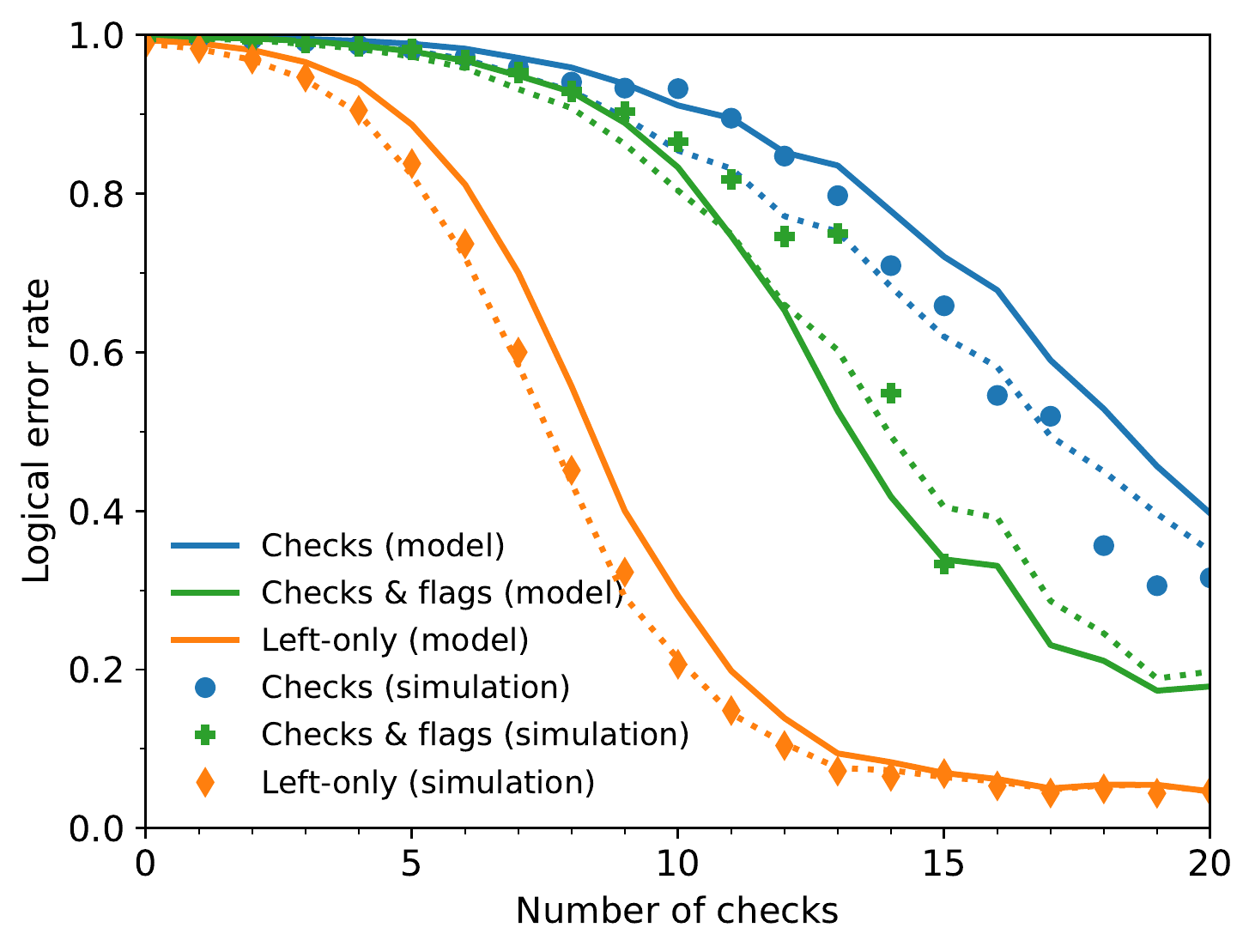}&
\includegraphics[width=0.45\textwidth]{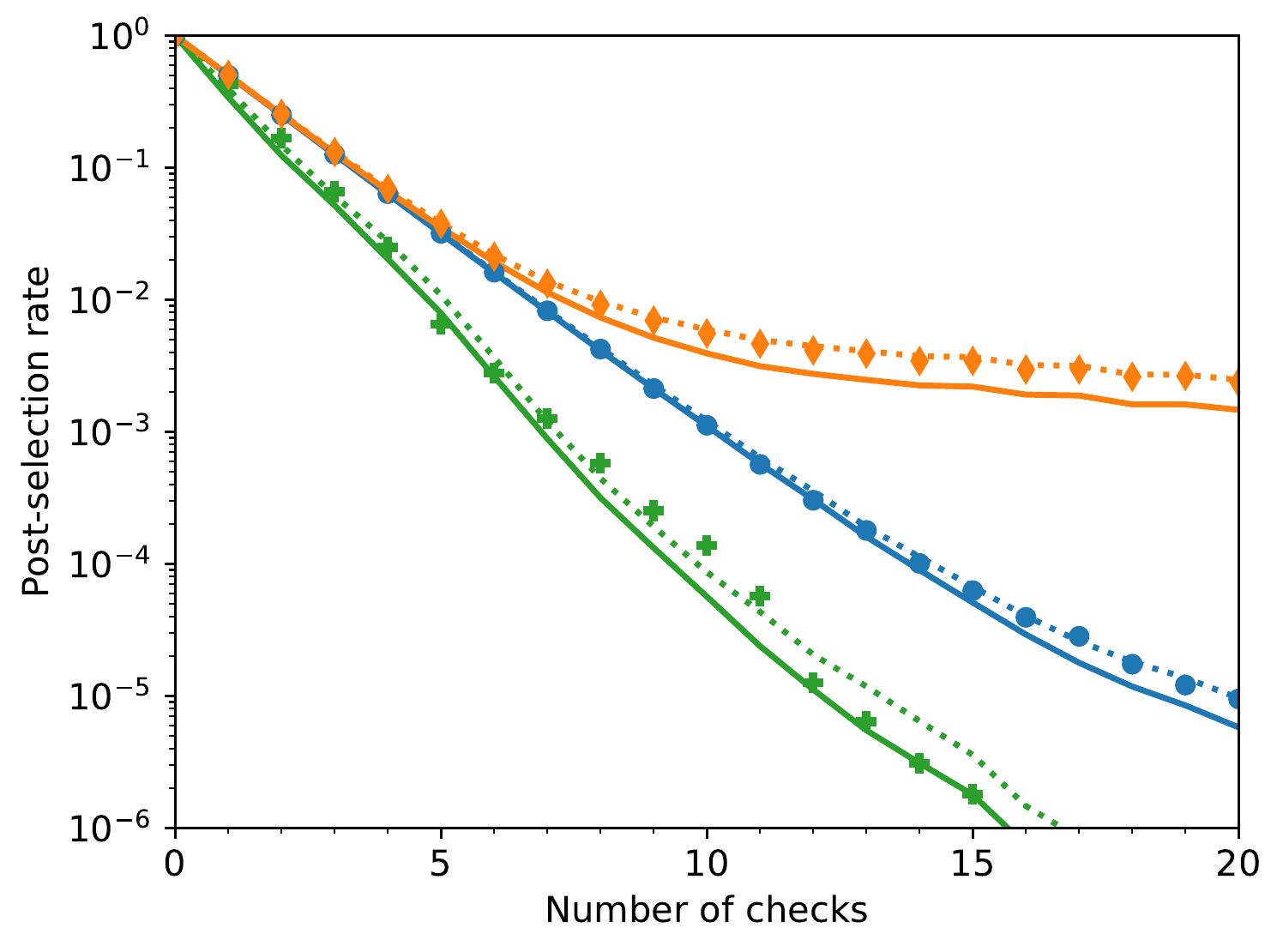}\\
({\bf{a}}) & ({\bf{b}})
\end{tabular}
\caption{The (a) logical error rate and (b) post-selection rate using all-to-all checks on a 30-qubit random Clifford circuit with LNN implementation. Noise on {\sc{cnot}} and delay gates are reduced by a factor of 10, the measurement error rate is multiplied by a factor $0.1$. The markers indicate simulated results based on 10 million samples. Results with fewer than 10 post-selected samples are omitted.
The solid and dotted lines represent the modeled results using payload error rates given respectively by the modeled upper bound and the estimate obtained from simulation without checks.}\label{Fig:Model003_30Q}
\end{figure}

\begin{figure}[!t]
\centering
\begin{tabular}{ccc}
\includegraphics[width=0.3\textwidth]{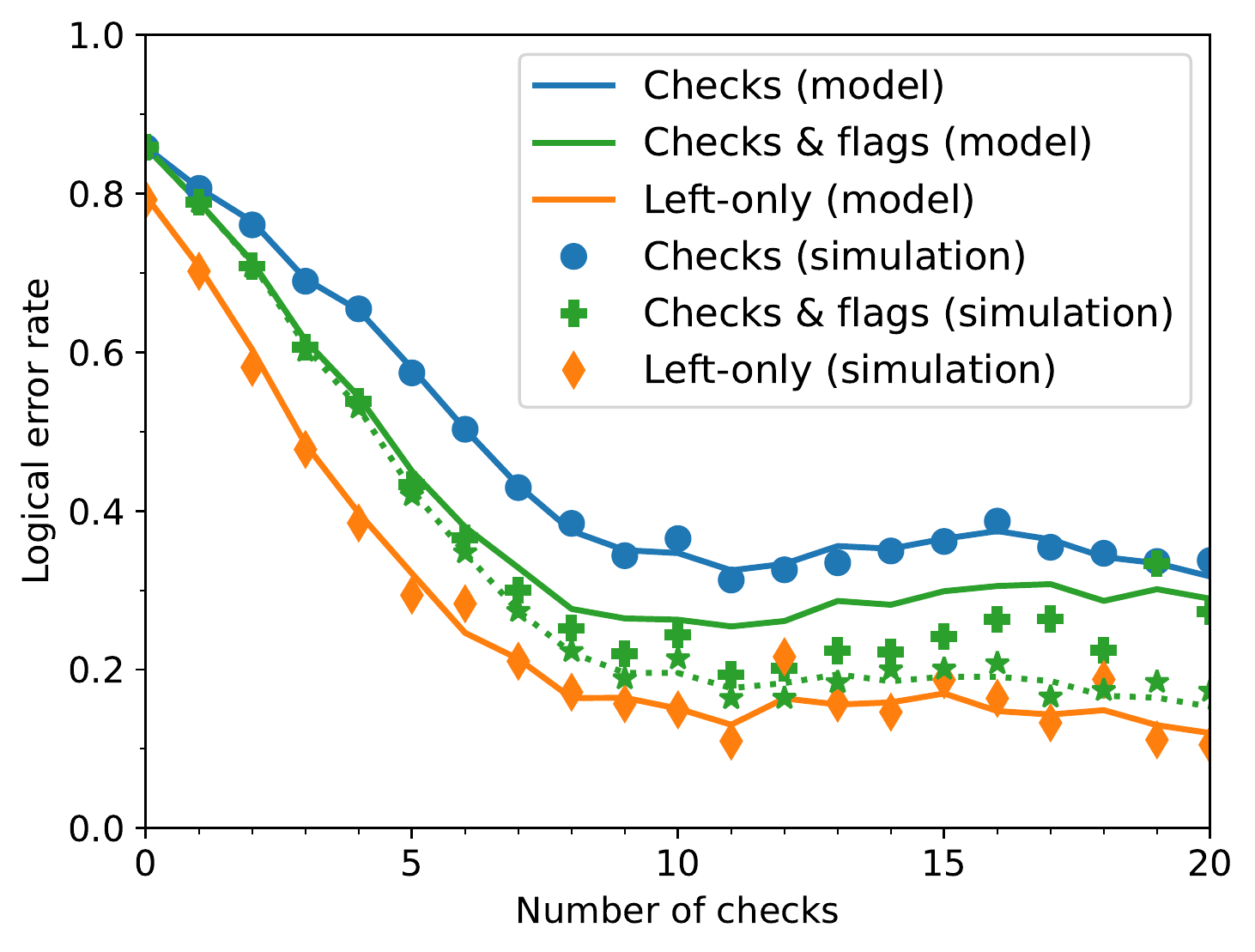}&
\includegraphics[width=0.3\textwidth]{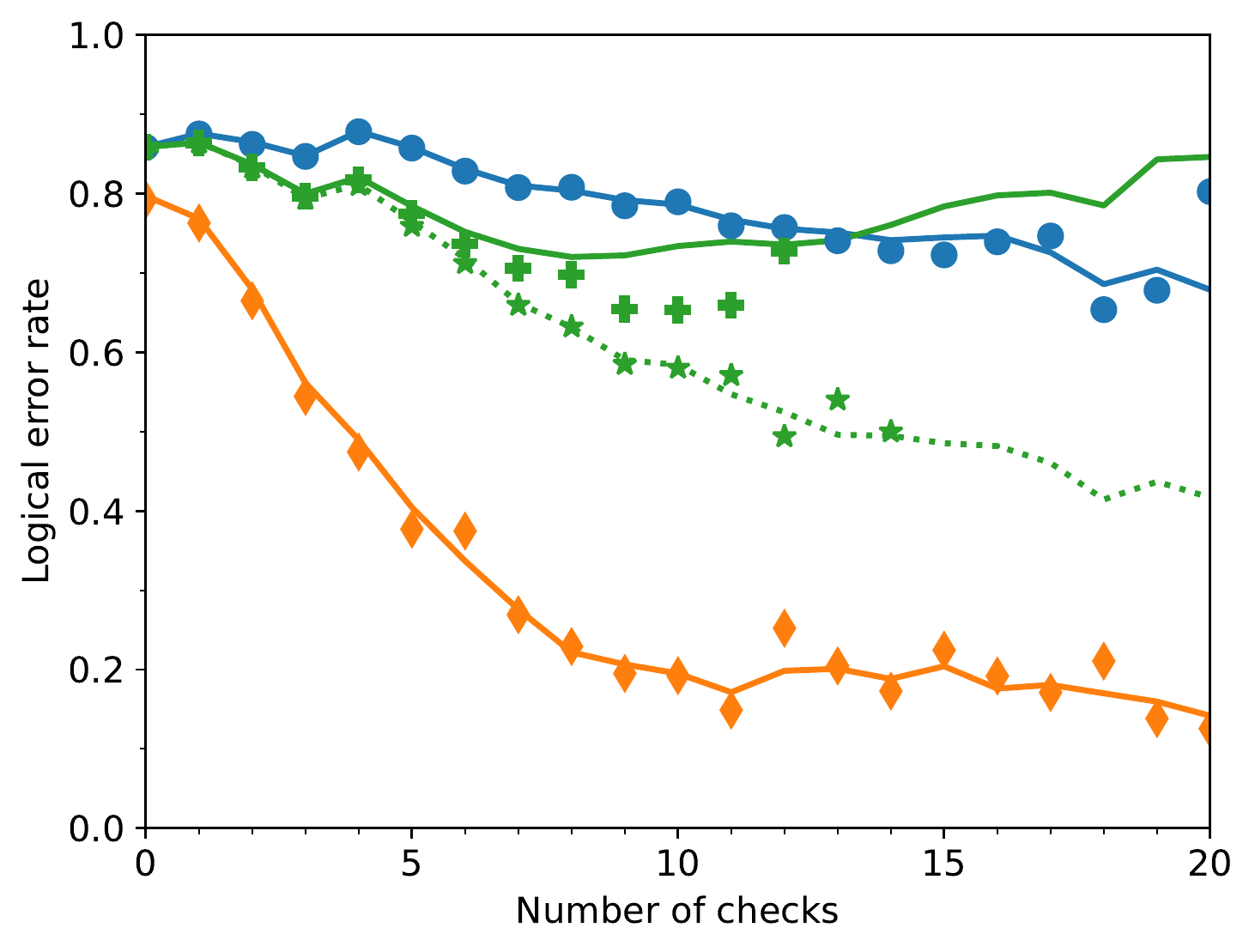}&
\includegraphics[width=0.3\textwidth]{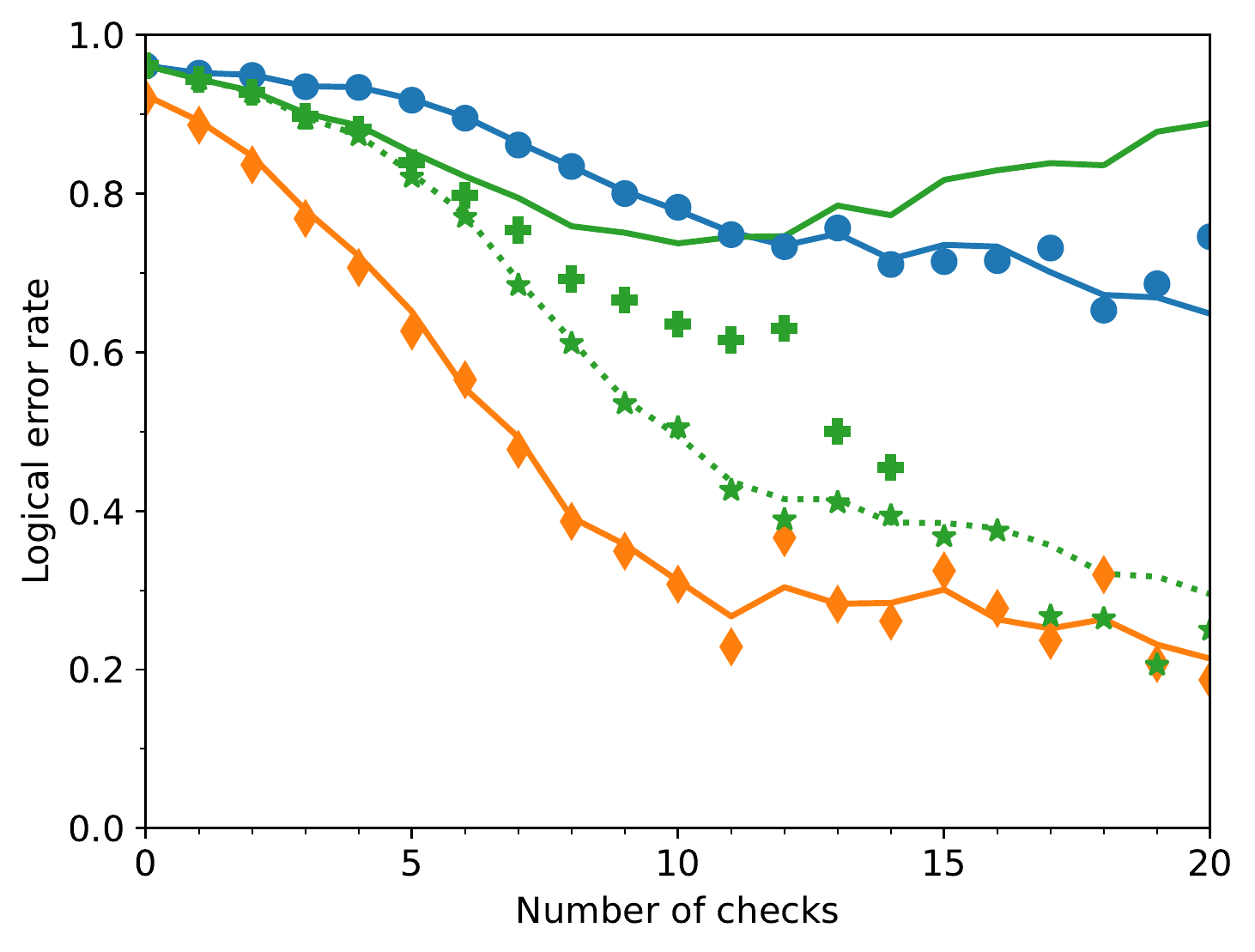}\\
\includegraphics[width=0.3\textwidth]{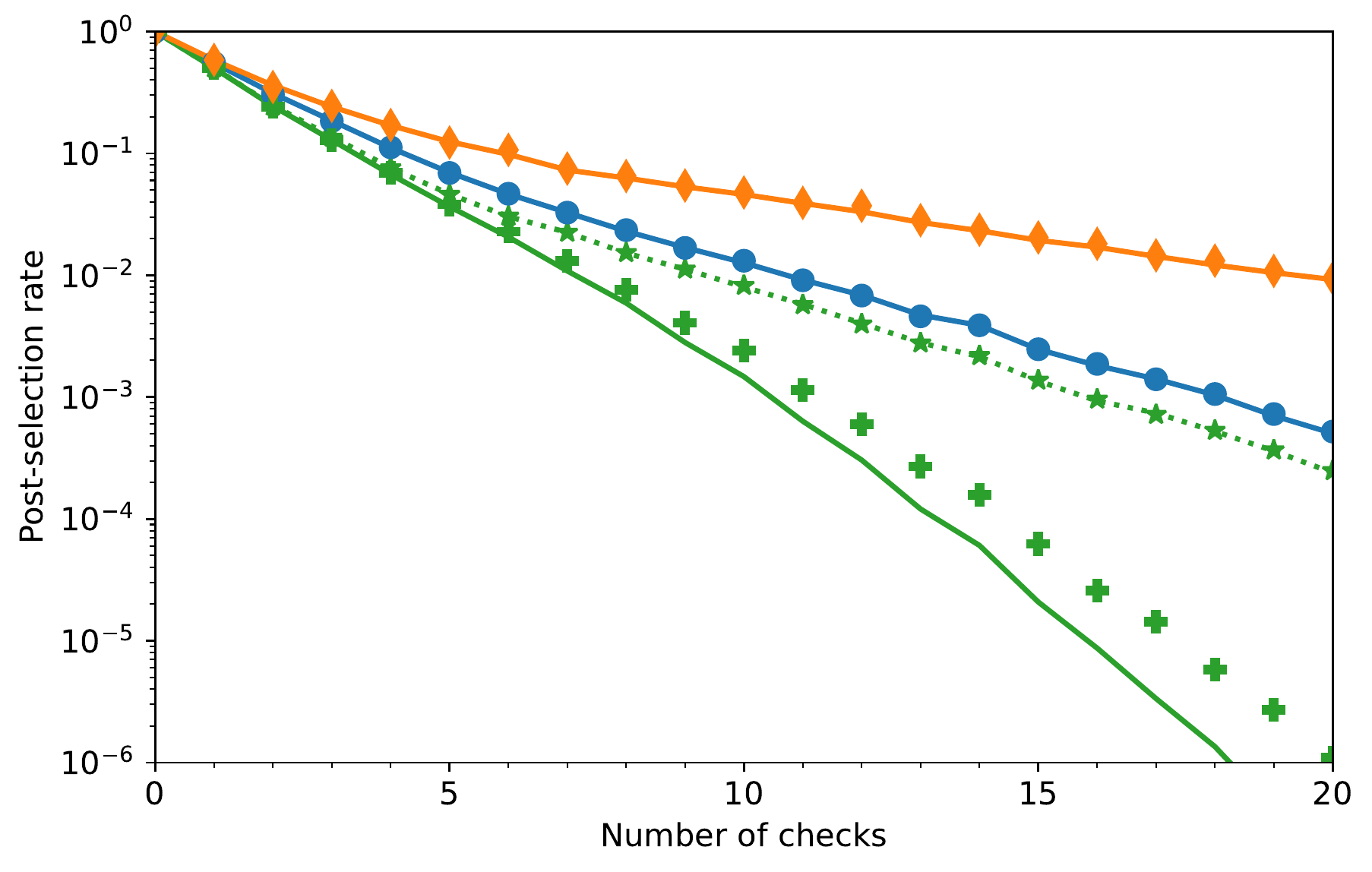}&
\includegraphics[width=0.3\textwidth]{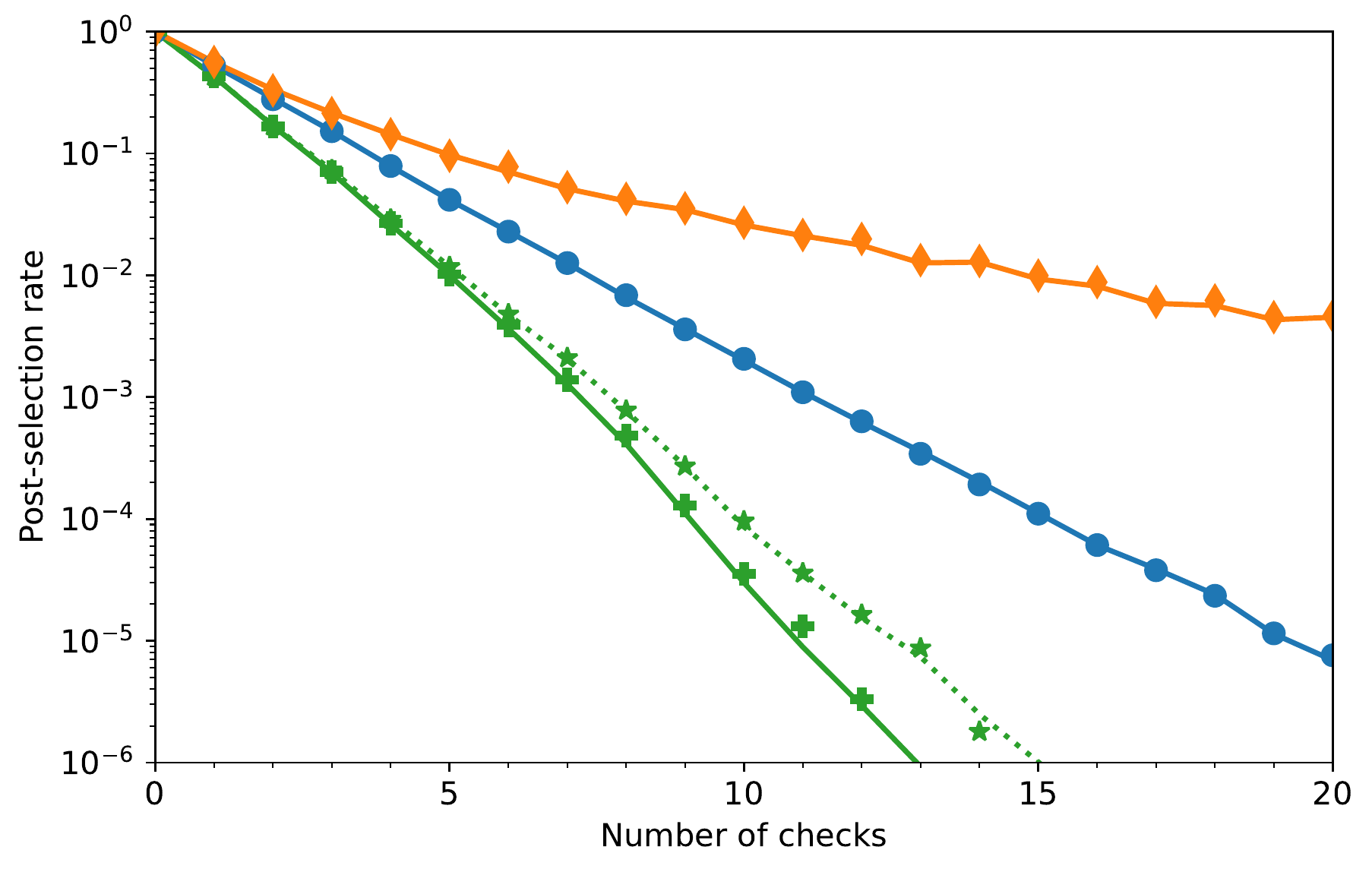}&
\includegraphics[width=0.3\textwidth]{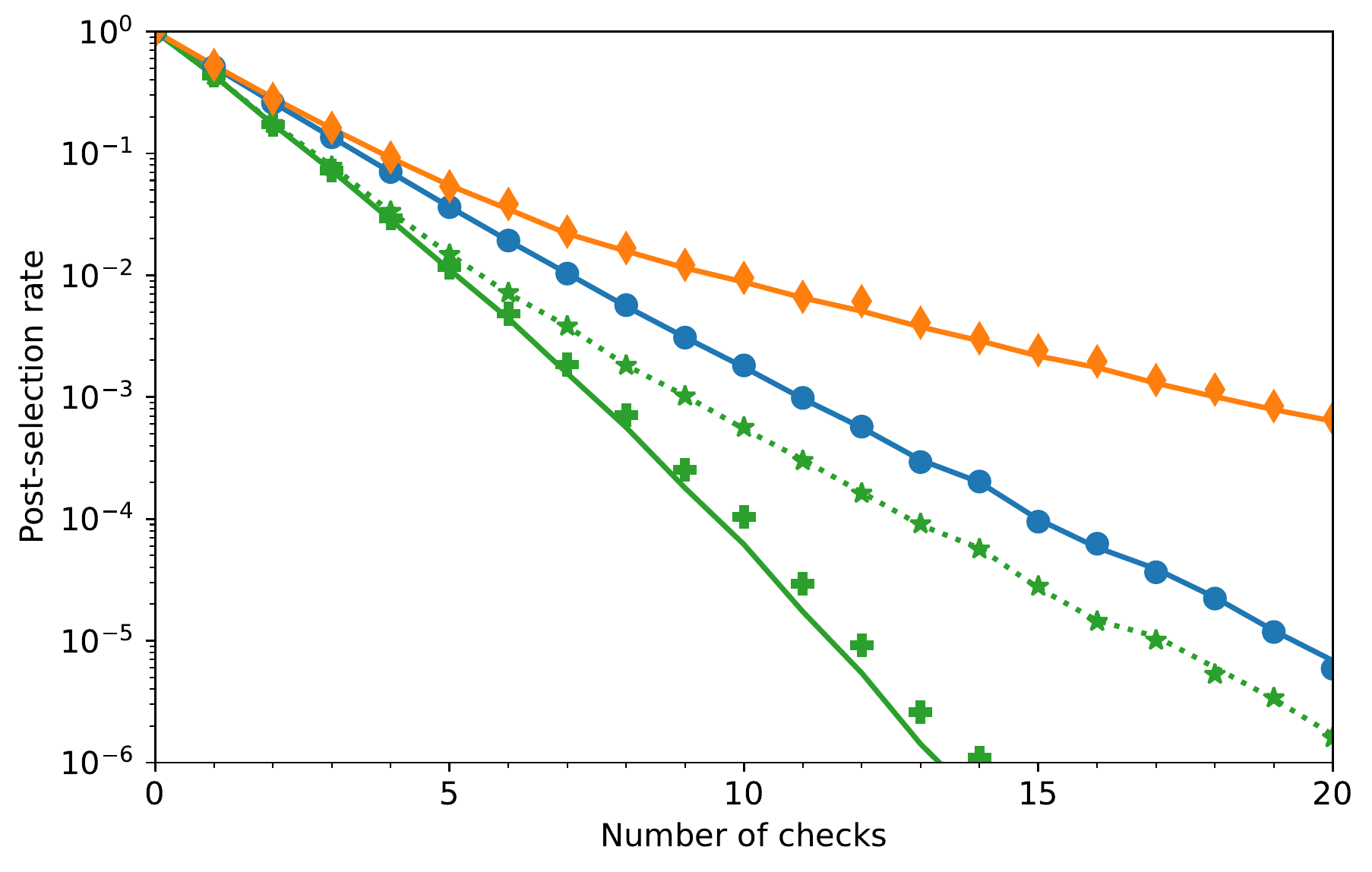}\\
({\bf{a}}) & ({\bf{b}}) & ({\bf{c}})
\end{tabular}
\caption{The logical error rates (top) and post-selection rates (bottom) for a random 10-qubit Clifford circuit with up to 20 checks on an LNN topology. The {\sc{cnot}} and delay noise terms are scaled respectively by (a) 0.3 and 0, (b) 0.3 and 0.2, and (c) 0.5 and 0. The markers represent simulated data based on 10 million samples, solid lines show the modeled results using the modeled upper bound on the payload error rate. The dotted line and asterisks show the modeled and simulated rates when using flagged checks when no noise is present on the \swapgate\ gates used to implement the flags. Data points with fewer than 10 post-selected samples are omitted.}\label{Fig:Model002_10Q}
\end{figure}

\subsubsection{Linear nearest neighbor connectivity}

The qubit connectivity on \ibmwashington\ is restricted to the heavy-hex pattern, which means that any circuit with {\sc{cnot}} gates on arbitrary qubit pairs first needs to be expressed using only those gates that match the actual qubit connectivity. Doing so may require a potentially large number of $\swapgate$ operations, and we therefore proposed the optimized check implementation in Section~\ref{Sec:LNNCircuit}. For our next setup, we find a linear chain of 50 qubits on \ibmwashington, for which the gate noise is available directly from the noise model of the backend, provided by Qiskit Aer. We again sample a random 10-qubit Clifford operator and consider the performance of coherent Pauli checks with increasing numbers of checks. The results obtained using simulation and modeling for this setting are plotted in Figure~\ref{Fig:Model002_10Q} for various levels of noise.

For the flagged setting, we see that the modeled results start to deviate from the simulated results as the number of checks increases. As mentioned in Section~\ref{Sec:PerformanceModel}, this is likely due to the assumption of independent noise on the qubits for each {\sc{swap}} gate used in the LNN implementation of the flags. To verify this, we also ran the simulation and modeling without any noise on the {\sc{swap}} gates, aside from any single-qubit delay noise that follows the gates. Under this assumption the noise becomes separable and the model assumptions hold. This is reflected in the results of the two methods which now closely match, as seen by the dotted line and asterisks in Figure~\ref{Fig:Model002_10Q}.

\section{Experiments}

So far, we have restricted the performance evaluation of coherent Pauli checks using simulations and modeling. However, the real purpose of coherent Pauli checks is, of course, to improve the logical error rate of payload circuits that are run on actual noisy quantum processors.

\begin{figure}
\centering
\begin{tabular}{c|c}
\raisebox{-0.5\height}{\includegraphics[width=0.32\textwidth]{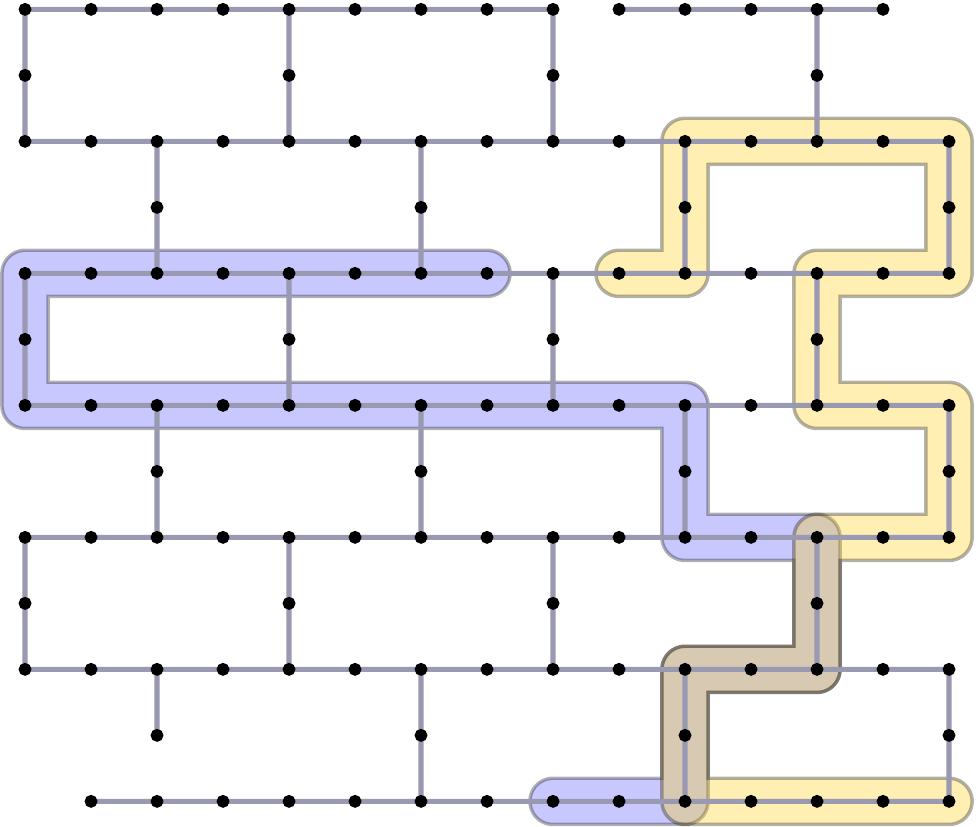}}\ \ \ %
&
\ \ \ {\small\begin{tabular}{lp{3.25cm}p{3.25cm}}
& Chain 1 & chain 2\\
\hline
Qubits & 
{\small{28, 29, 30, 31, 32, 35, 36, 46, 47, 49, 50, 51, 55, 68, 69, 70, 74, 87, 88, 89, 93, 104, 105, 106, 111, 122, 123, 124, 125, 126}}
&
{\small{37, 38, 39, 40, 41, 42, 43, 44, 52, 56, 57, 58, 59, 60, 61, 62, 63, 64, 65, 66, 73, 85, 86, 87, 93, 104, 105, 106, 111, 120, 121, 122}}
\\[3pt]
\hline
$T_1$ ($\mu$s)
& 32.45 -- 160.27 (94.96)
& 51.44 -- 147.63 (93.92)
\\
$T_2$ ($\mu$s)
& 3.23 -- 219.60 (96.80)
& 15.76 -- 249.81 (101.53)
\\
Meas.~error (\%)
& 0.17 -- 3.26 (0.91)
& 0.17 --17.96 (1.86)
\\
{\sc{cnot}} fidelity (\%)
& 89.4 -- 99.4 (98.4)
& 95.5 -- 99.4 (98.5)
\\
{\sc{cnot}} time (ns)
& 313 -- 1088 (517)
& 356 -- 981 (527)
\\
{\sc{sx}}, {\sc{x}} time (ns)
& 35.6
& 35.6
\\
\hline
\end{tabular}}
\end{tabular}
\medskip
\caption{The topology of \ibmwashington\ with qubits 0--126 labeled left-to-right from top to bottom along with selected qubit chains 1 (yellow) and 2 (blue). The table on the right lists the qubit indices of the chains and provides a snapshot of the device properties on the selected qubits, providing the minimum, maximum, the average values (in brackets). The exact $T_1$ and $T_2$ times and error rates vary over time and are therefore only illustrative.}\label{Fig:IBM_Washington}
\end{figure}

\subsection{Pauli checks}

The heavy-hex topology of \ibmwashington, illustrated in Figure~\ref{Fig:IBM_Washington}, means that we have to focus on the linear nearest neighbor implementation of Pauli checks.
For this we select two chains of qubits with high-fidelity readout and {\sc{cnot}} gates, as shown in Figure~\ref{Fig:IBM_Washington}. We then define a simple payload circuit on even numbers of qubits, consisting of 48 successive {\sc{cnot}} gates on alternating qubits pairs, such that all qubits are either the control or the target of a {\sc{cnot}} gate.
Starting with one-sided checks, we add up to 15 randomly sampled checks to the payload and then accumulate 250,000 shots for each circuit, starting from the zero initial state. We repeat this process 10 times, each time with newly sampled checks, and plot the resulting logical error and post-selection rates in Figure~\ref{Fig:Exp029B}.
For both qubit chains we see a marked reduction in the logical error rate on payload circuits on up to 10 qubits, summarized by the table in Figure~\ref{Fig:Exp029B}(e). Increasing the number of checks beyond 15 may still further reduce the logical error rate for our payloads of size 6, 8, and 10 qubits, but the number of shots required to accurately show this may become prohibitive. As seen in the plots, the variance in the logical error and post-selection rates estimates grows as the number of checks increases and the rates decrease. As shown in Figure~\ref{Fig:Exp029B}, the post-selection rates for both qubit chains follow very similar trajectories.

\begin{figure}[!tb]
\centering
\begin{tabular}{cc}
\includegraphics[width=0.45\textwidth]{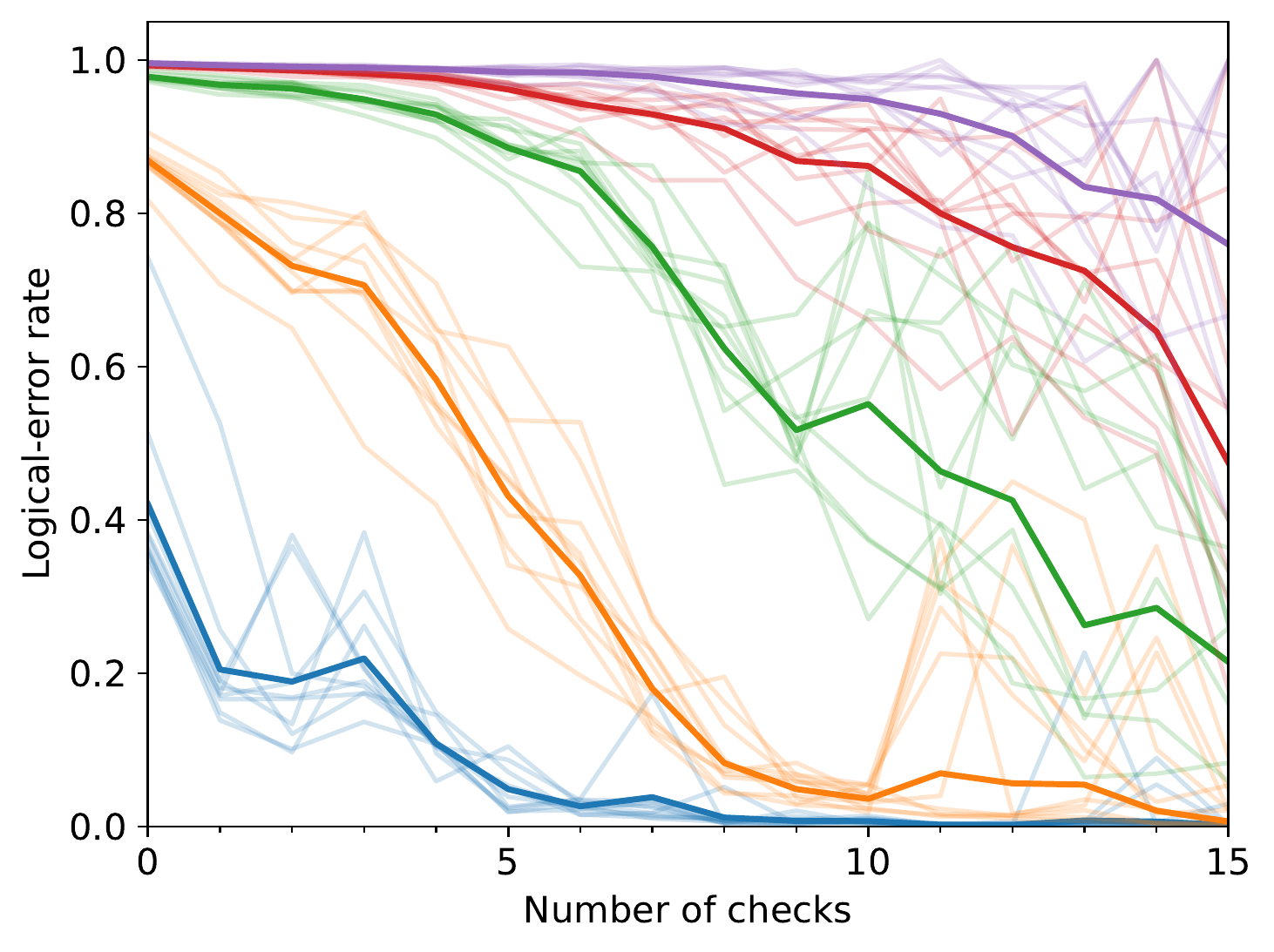}&
\includegraphics[width=0.45\textwidth]{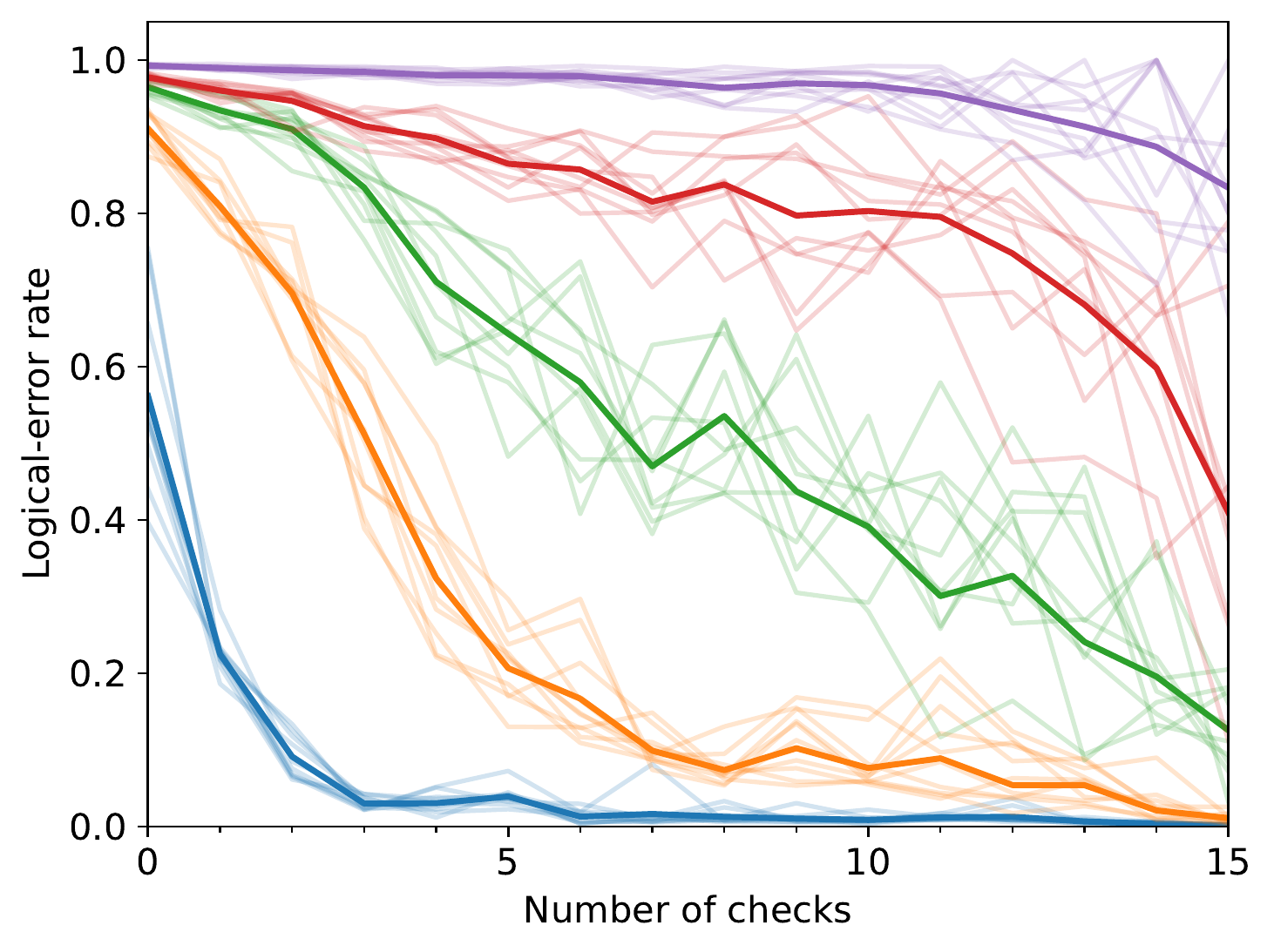}\\[-4pt]
({\bf{a}}) & ({\bf{b}}) \\[3pt]
\includegraphics[width=0.45\textwidth]{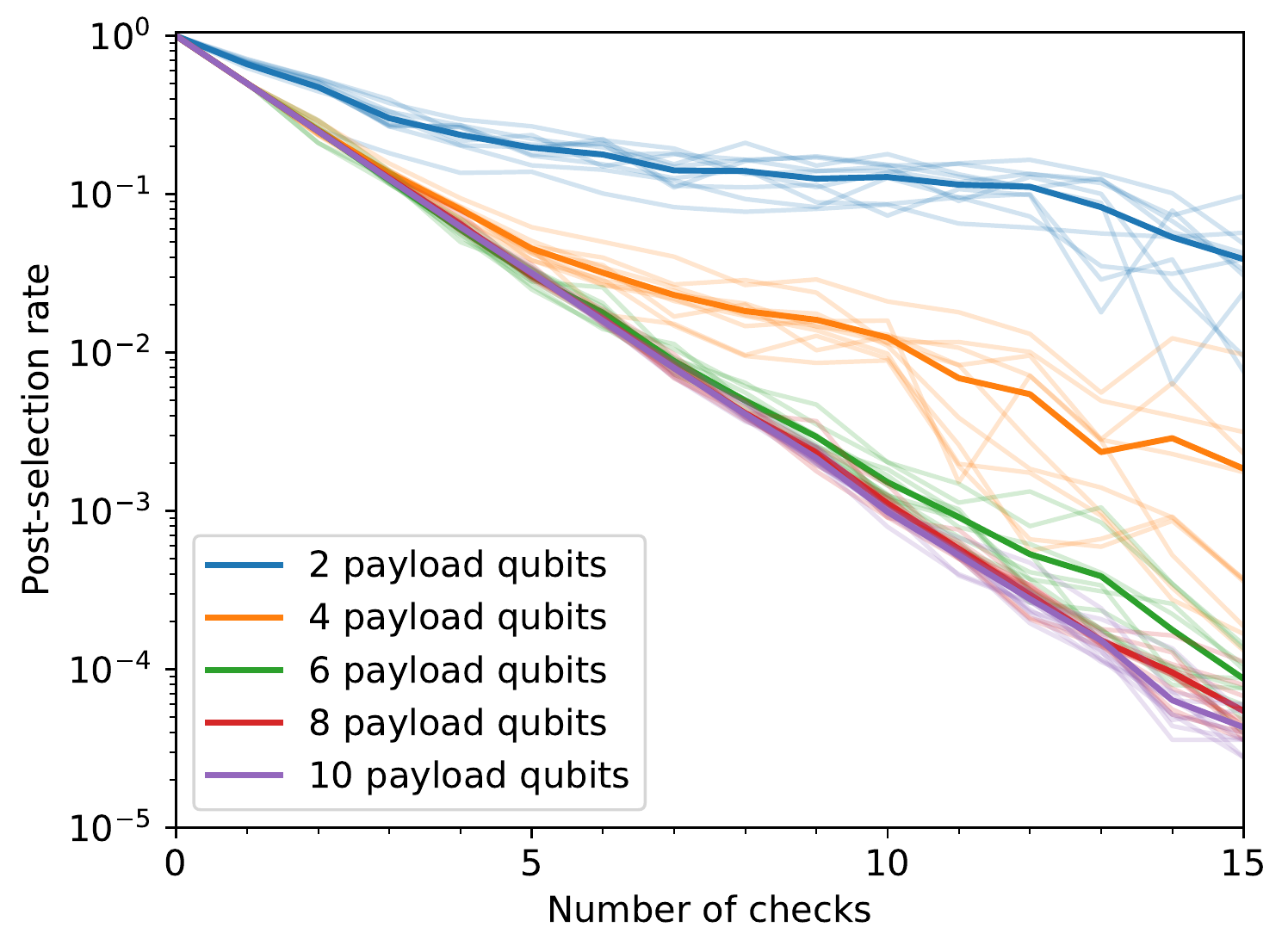}&
\includegraphics[width=0.45\textwidth]{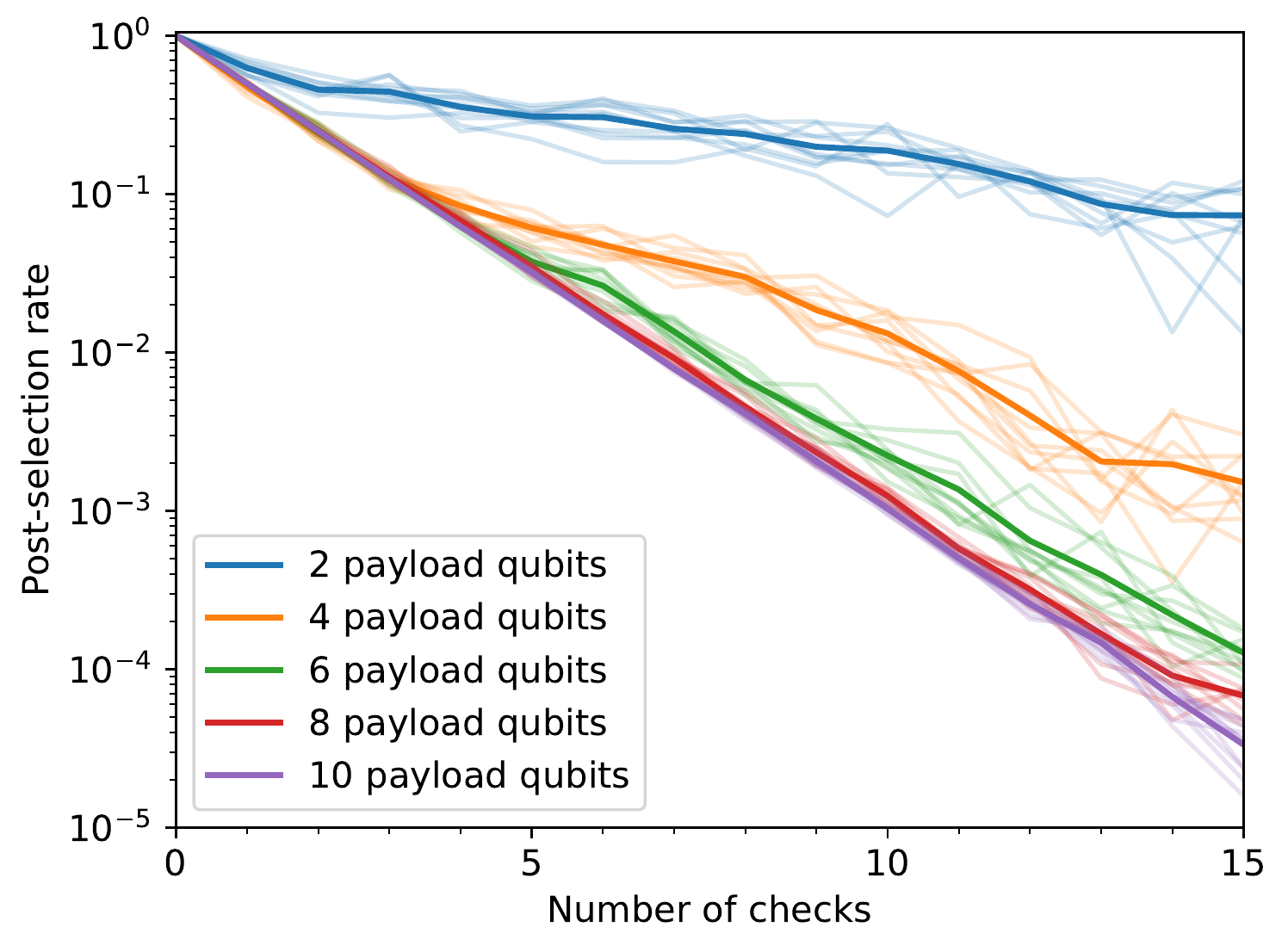}\\[-4pt]
({\bf{c}}) & ({\bf{d}})\\[10pt]
\multicolumn{2}{c}{
\small
\begin{tabular}{lrrrrrrrrrrr}
&\multicolumn{5}{c}{Chain 1} & & \multicolumn{5}{c}{Chain 2}\\ 
\cline{2-6}
\cline{8-12}
\\[-10pt]
Payload qubits &  2&4&6&8&10 & & 2&4&6&8&10  \\
\hline
Logical error rate (\%), no checks
&  42.02&86.88&97.82&99.32&99.59 &
& 56.17&91.06&96.46&97.75&99.28  \\
Logical error rate (\%), 15 checks
&  0.20&0.67&21.46&47.45&75.93 &
&  0.15&1.10&12.54&40.94&83.33  \\
Post-selection rate (\%), 15 checks
&  3.89& 0.18& 0.009& 0.005& 0.004 &
 & 7.35& 0.15& 0.013& 0.007& 0.003\\
\hline
\end{tabular}}\\
\\[-10pt]
\multicolumn{2}{c}{({\bf{e}})}
\end{tabular}
\caption{Performance of one-sided Pauli checks on  \ibmwashington\ for payload circuits of 2 to 10 qubits, consisting of 48 repeated {\sc{cnot}} gates on alternating pairs of qubits, with (a,b) the logical error rates and (c,d) the post-selection rates for qubit chains 1 and 2. Each point on the faint curves is based on randomly sampled checks and represents the estimated rate based on 250,000 shots with zero initial state. The solid line represents the rate obtained by combining the shots from ten such instances (each of which is shown as a separate faint curves of the same color). The table in (e) lists the initial and final logical error rates and the final post-selection rates.}\label{Fig:Exp029B}
\end{figure}

While useful for performance evaluation, a payload consisting of repeated {\sc{cnot}} gates simply implements an identity operation and therefore has no use in practical applications. On the other hand, randomly sampled Clifford operators are a much more representative choice for the payload operator. However, for such operators it is generally difficult to determine whether an individual measured bit string on the data qubits is affected by errors or not. Take for instance the payload circuit that applies a Hadamard gate on each of the data qubits. In this case each bit string would represent a valid outcome, making it impossible to detect errors unless we look at the distribution, which quickly becomes intractable. As a trade-off, we therefore generate Clifford circuits that implement random qubit permutations. Although these circuits are typically shallower in depth than general Clifford circuits, they have the advantage that we can easily determine the desired outcome. To make the experiments more interesting, we entangle alternating pairs of neighboring data qubits by initializing them as EPR pairs. Combined with the permutation circuit, this creates possibly long-distance entanglement between certain pairs of qubits. We therefore know that measurements of qubits at permuted indices should match. In this case, we only miss errors that simultaneously affect both qubits of one or more pairs. For good overall performance we sample 10 random payload circuits and Pauli check instances, for up to 15 checks. The results for the two qubit chains are shown in Figure~\ref{Fig:Exp033B}. The logical error rate decreases for payloads over 4 and 6 qubits, with checked circuit {\sc{cnot}} depth reaching 72 and 90, respectively. The error rate for the two-qubit payload circuits was small to start with and after an initial increase for a small number of checks it returns to roughly the same value as that without checks. The logical error rate for the 8-qubit payloads gradually decreases and may require more checks to reduce significantly. However, given the small post-selection rate at that point, this may require a prohibitively large number of shots. The table shown in Figure~\ref{Fig:Exp033B}(d) provides a detailed overview of the setting and the results.

\begin{figure}
\centering
\begin{tabular}{ccc}
\includegraphics[width=0.305\textwidth]{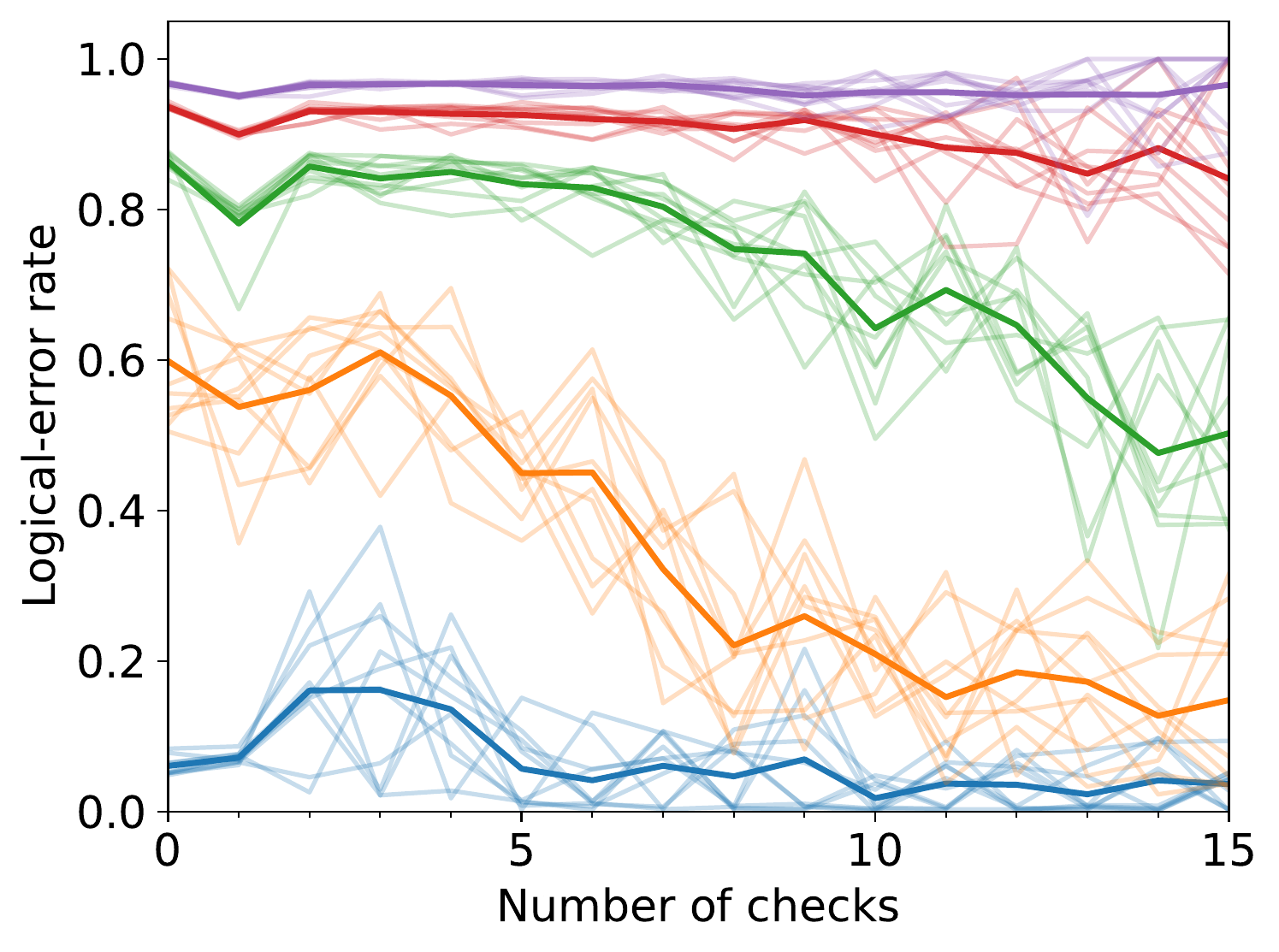}&
\includegraphics[width=0.305\textwidth]{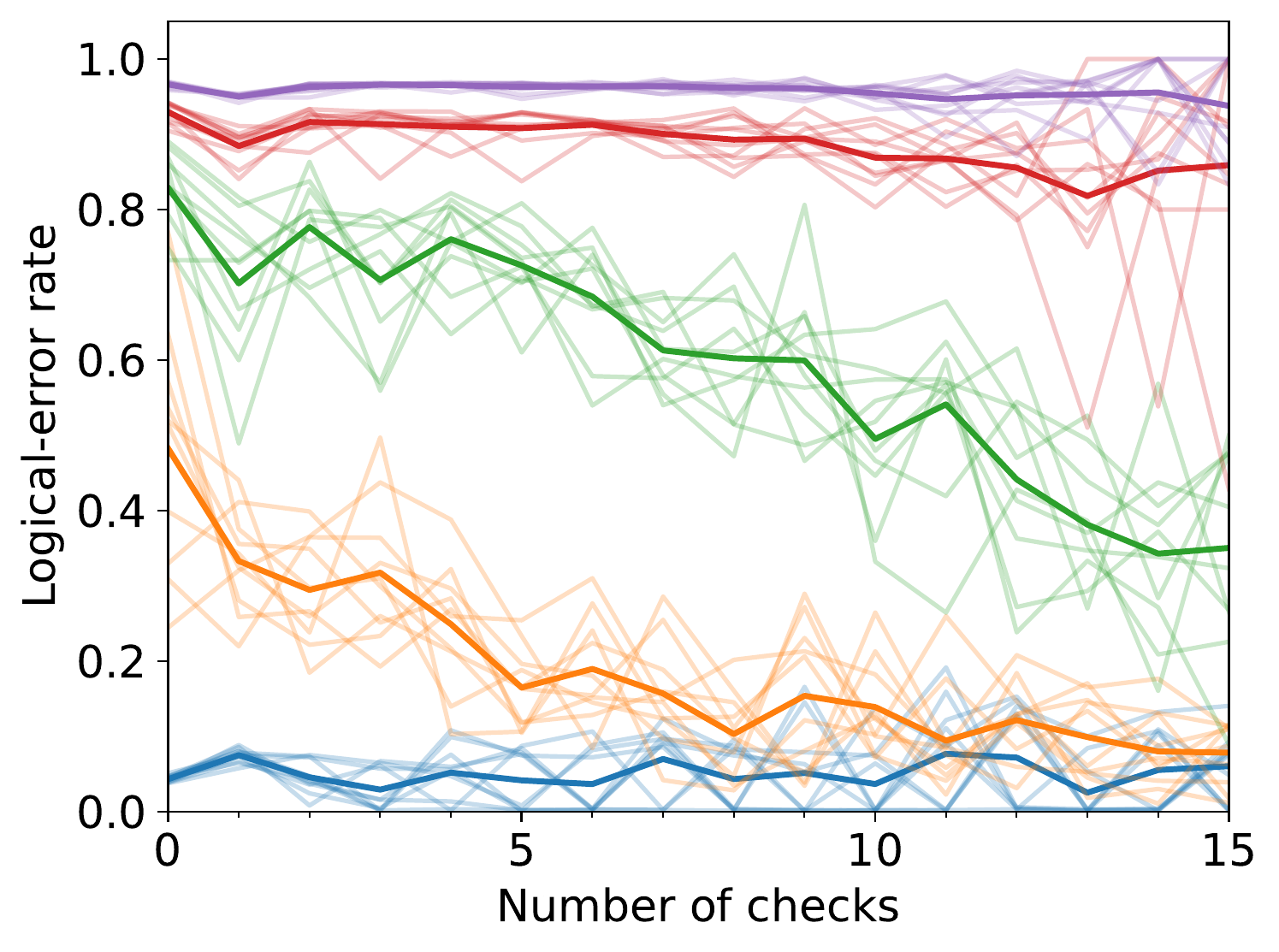} &
\includegraphics[width=0.305\textwidth]{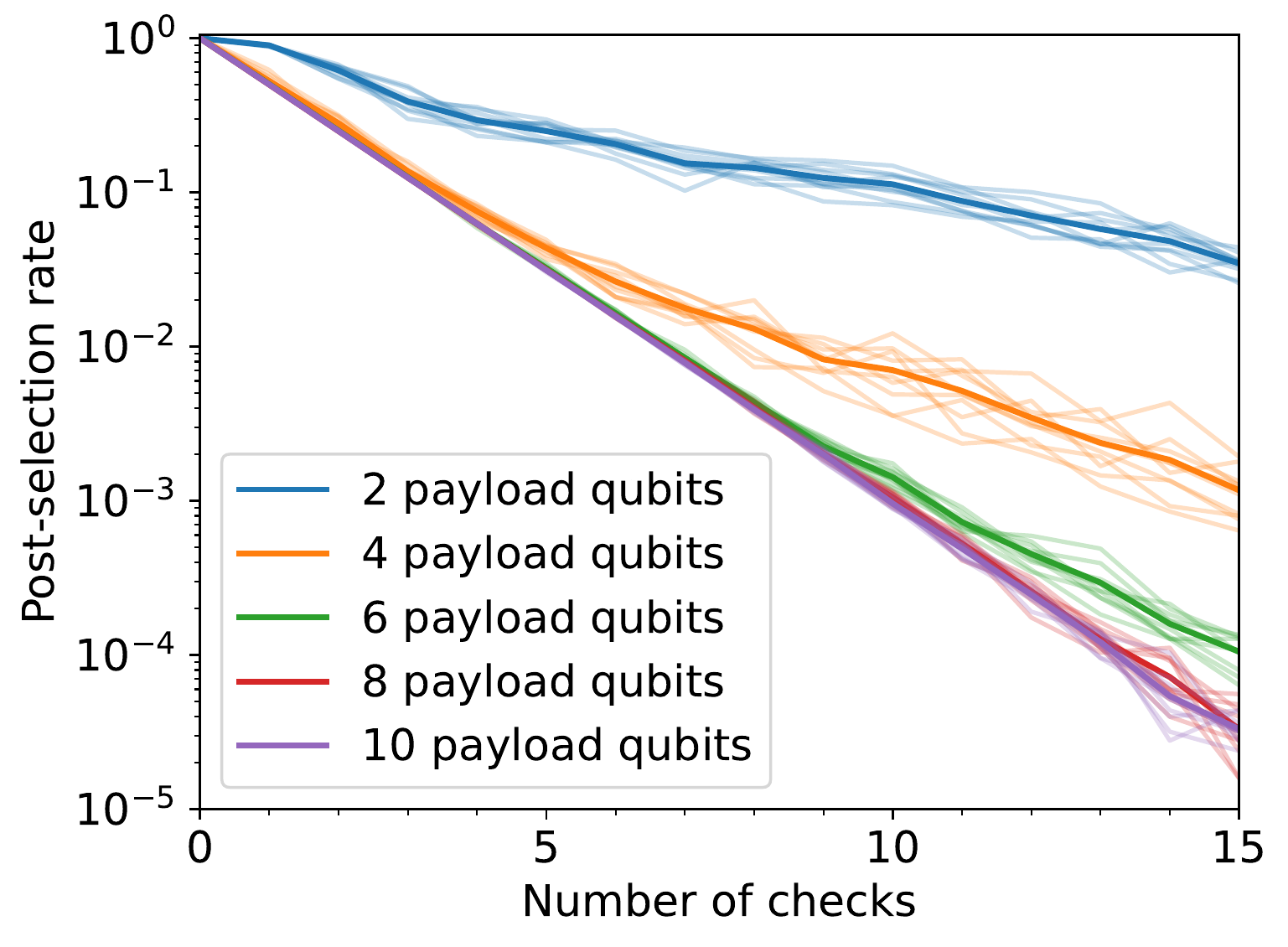}\\
({\bf{a}}) & ({\bf{b}}) & ({\bf{c}})\\
\\
\multicolumn{3}{c}{
\small
\begin{tabular}{lrrrrr}
\hline
\\[-9pt]
\multicolumn{6}{c}{Payload and circuit}\\[2pt]
\hline
Payload qubits & 2 & 4 & 6 & 8 & 10 \\
Payload {\sc{cnot}} depth
& 3--6 (5) & 10--22 (19) & 25--33 (29) & 27--46 (41) & 45--60 (52)\\
 Payload {\sc{cnot}} count
 & 3--6 (5) & 15--30 (24) & 48--64 (55) & 76--115 (101) & 131--173 (157)\\
Circuit {\sc{cnot}} depth, 15 checks
& 42--49 (45) & 60--72 (68) & 82--90 (85) & 87--111 (104) & 114--126 (120)\\
Circuit {\sc{cnot}} count, 15 checks
 & 72--79 (76) & 165--180 (174) & 269--295 (282) & 372--423 (405) & 516--555 (535)\\
\hline
\\[-9pt]
\multicolumn{6}{c}{Qubit chain 1}\\[2pt]
\hline
Logical error rate (\%), no checks  
& 6.08 & 59.88 & 86.36 & 93.66 & 96.73 \\
Logical error rate (\%), 15 checks
& 3.67 & 14.83 & 50.29 & 84.09 & 96.59 \\
Post-selection rate (\%), 15 checks
& 3.47 & 0.12 & 0.011 & 0.003 & 0.003 \\
\hline
\\[-9pt]
\multicolumn{6}{c}{Qubit chain 2}\\[2pt]
\hline
Logical error rate (\%), no checks
& 4.33 & 48.27 & 82.89 & 92.96 & 96.62\\
Logical error rate (\%), 15 checks
& 6.05 & 7.84 & 35.05 & 85.91 & 93.77\\
Post-selection rate (\%), 15 checks
& 5.90 & 0.29 & 0.015 & 0.004 & 0.003 \\
\hline
\end{tabular}}\\
\\[-8pt]
\multicolumn{3}{c}{({\bf{d}})}
\end{tabular}
\caption{The performance of Pauli checks on random 10-qubit permutation operators, with successive qubit pairs initialized as EPR pairs. For each combination of number of checks and qubits we independently sample 10 payload circuits and Pauli checks. We connect results with the same instance number by lines for better visibility of the fluctuations in the rates. The average logical error rate over the different instances is indicated by the dark solid line. Plots (a) and (b) show the logical error rate for qubit chains 1 and 2, plot (c) shows the post-selection rate for the first qubit chain. The results for the second qubit chain are overall slightly better but otherwise similar and are omitted.}\label{Fig:Exp033B}
\end{figure}

We now look at the Pauli check performance on randomly sampled, but otherwise fixed permutations on 4, 6, and 8 qubits. We randomly sample four series of checks and acquire 250,000 shots for each checked circuit instance. The resulting logical error rates, along with the average obtained by summing all post-selection and success counts in Figure~\ref{Fig:Exp48B_50B}(a). The average logical error rates reduce from 27.20\% to 1.86\% on 4 qubits, from 84.23\% to 38.61\% on 4 qubits, and from 91.74\% to 83.78\% on 8 qubits. For comparison, we also ran two-sided checks with the right Pauli-Z checks implemented as part of the quantum circuit and plot the results in Figure~\ref{Fig:Exp48B_50B}(b). As with the simulated and modeled results in Figure~\ref{Fig:Model003_30Q}, we see that the logical error rates for the two-sided checks deteriorate with increasing numbers of checks. Given that the one- and two-sided experiments were run several days apart, we ascribe the slight differences in the error rates without checks to automatic recalibration of the gates and gradual changes in the noise levels. In Figure~\ref{Fig:Exp48B_50B}(c) we show the post-selection rates for both settings.

\begin{figure}
\centering
\begin{tabular}{ccc}
\includegraphics[width=0.30\textwidth]{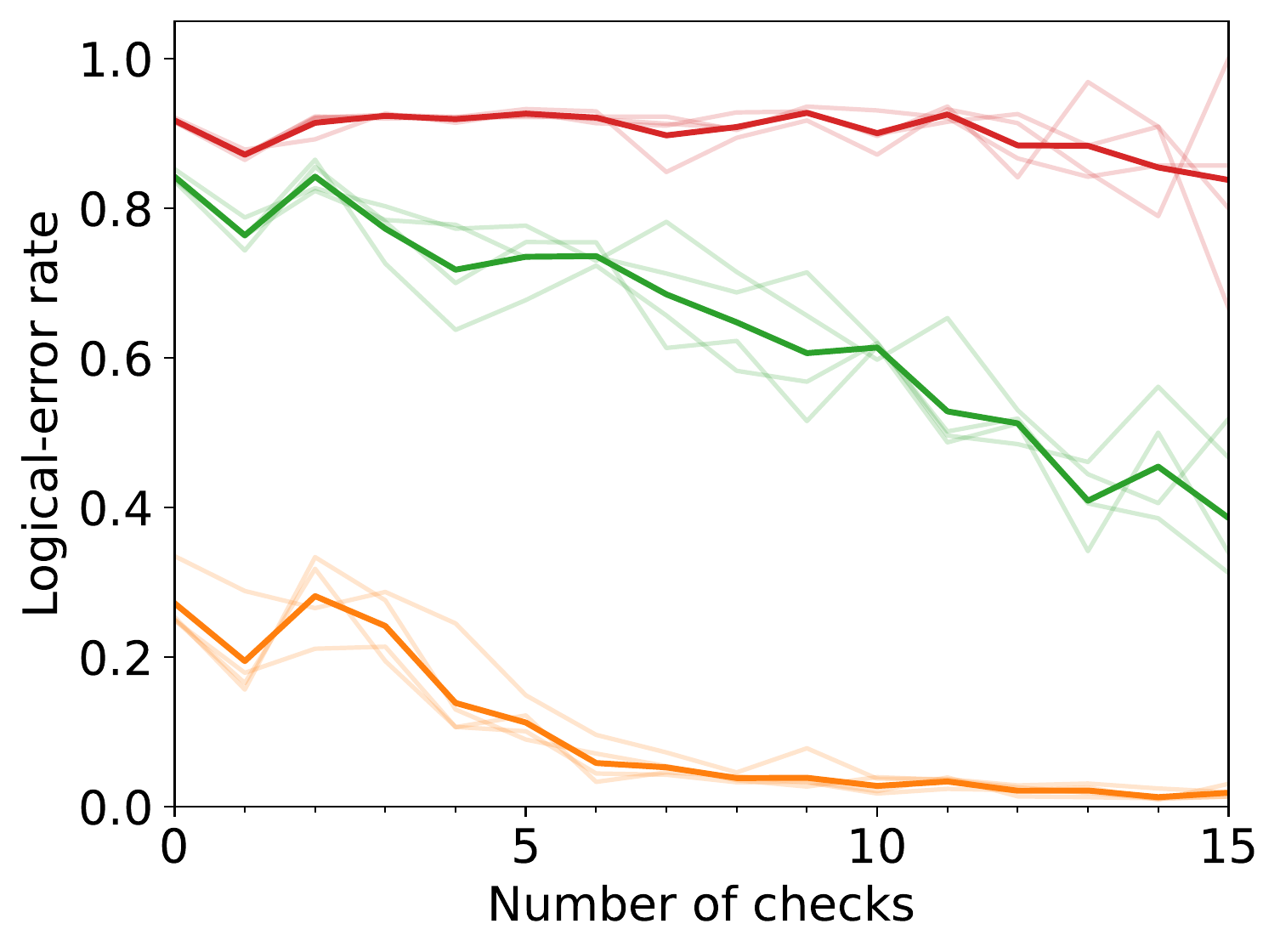} & \includegraphics[width=0.30\textwidth]{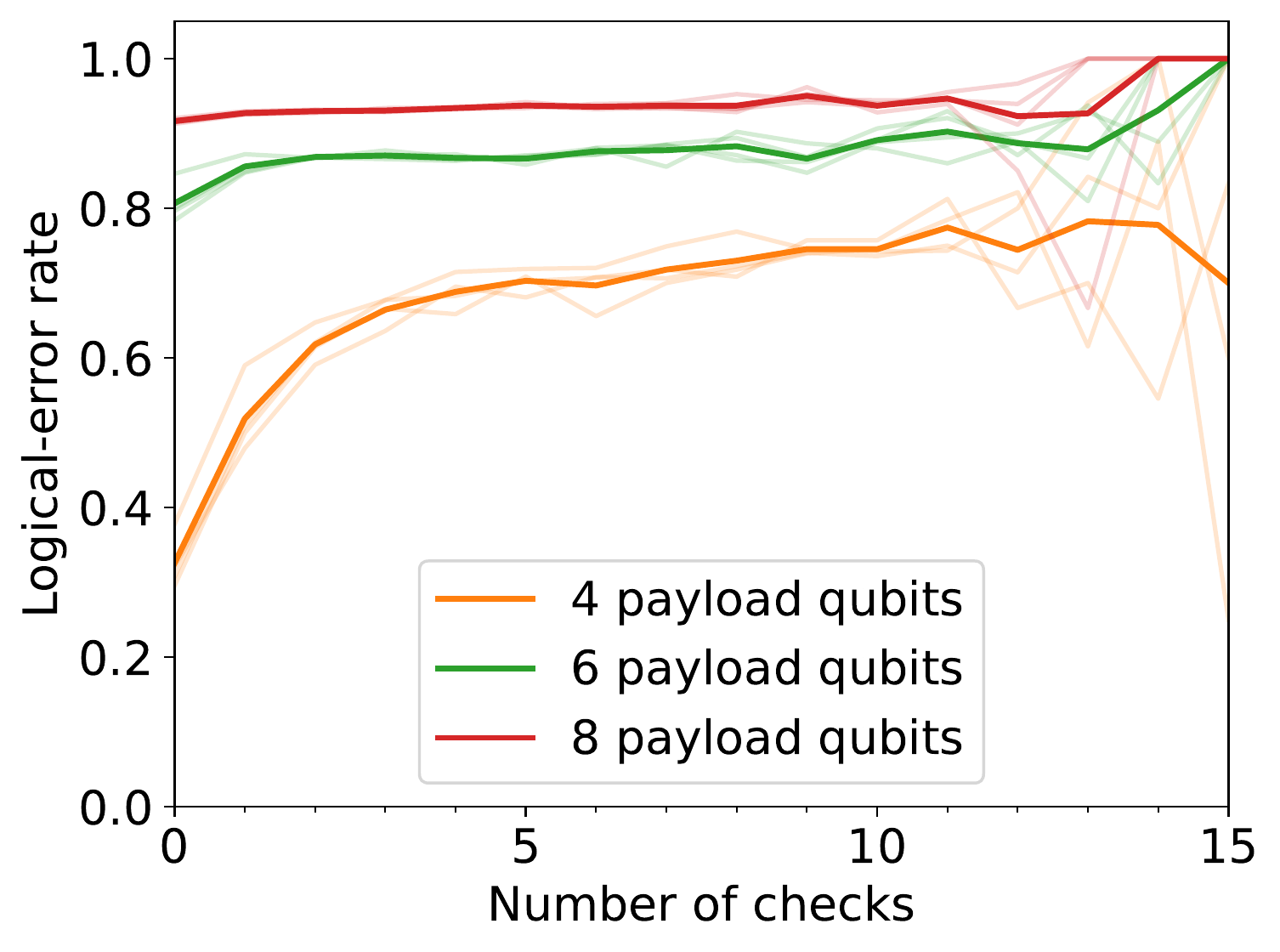} & \includegraphics[width=0.30\textwidth]{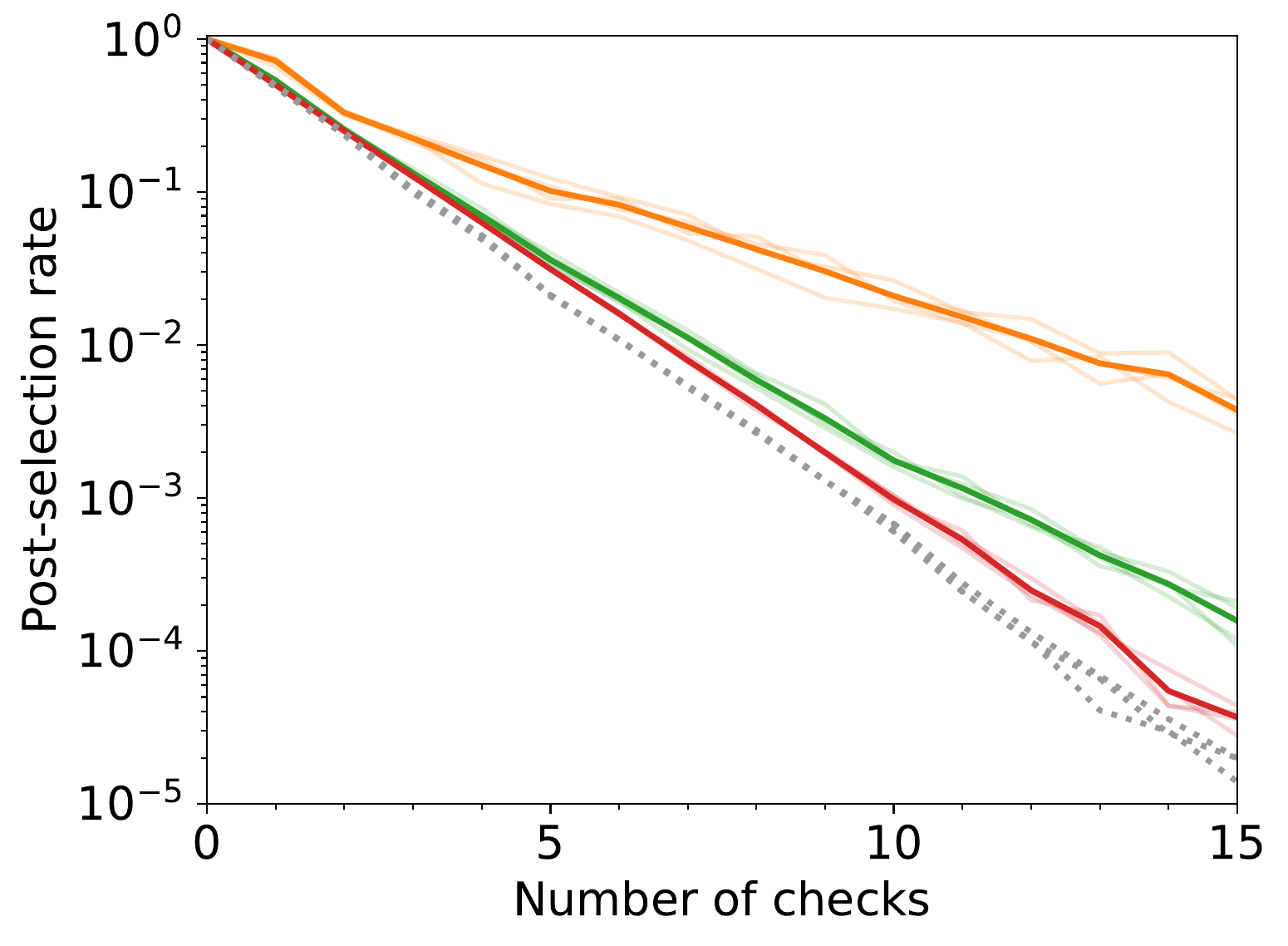}\\ 
({\bf{a}}) & ({\bf{b}}) & ({\bf{c}})\\[3pt]
\end{tabular}
\caption{Logical error rates using (a) left-only and (b) two-sided Pauli-Z checks on a randomly sampled permutation operation on 4, 6, and 8 qubits. The post-selection rates for the left-only check are shown in (c) along with those for the two-sided check (dotted gray lines).}\label{Fig:Exp48B_50B}
\end{figure}

\subsection{Repeated readout}

For our readout mitigation experiments, we use the chain of 11 qubits on IBM's 65-qubit superconducting quantum processor \ibmithaca, illustrated in Figure~\ref{Fig:Ithaca}(a). For our purposes, the chain facilitates the readout of a single qubit with up to 10 checks. As a start, we measure the readout-error rates of the qubits in the chain. Throughout this section we apply readout twirling by randomly applying an {\sc{i}} or an {\sc{x}} gate just prior to measurement with equal probability and classically undoing it. Doing so symmetrizes the readout channel of a single qubit and removes bias in the transitions. For basic benchmarking of the individual qubit readout rates we prepare four circuits: two twirl instances for each of the \ket{0} and \ket{1} initial states. We acquire 8192 measurements for each circuit instance to estimate the readout errors, which are tabulated in Figure~\ref{Fig:Ithaca}(b). Since simultaneous readout of several qubits may change the readout error on the individual qubits we also measure all qubits simultaneously and derive their error rates. For this we prepare 256 circuit instances where, randomly choose the initial state and twirl gates, such that each of the four combinations on each qubit appears in exactly 64 circuits. For each circuit instance, we gather 1024 shots, and list the results in Figure~\ref{Fig:Ithaca}(b).
Throughout the remainder of this section, we always measure all qubits, to avoid large changes in the readout noise and consider only the values on the qubits of interest.

We estimate the {\sc{cnot}} error rates by preparing each of the neighboring pairs of qubits in the chain, indexed by $(i,i+1)$, to the four computational states and applying a {\sc{cnot}} gate followed by simultaneous measurement of the qubits with a readout twirl. For each initial state, this gives an empirical probability distribution $\hat{p}$ of the four possible outcomes. These distributions are affected by the readout noise and we correct them using the estimated $2\,{\times}\, 2$ stochastic readout transition matrices $\{A_i\}$ for each qubit index $i$ of the chain. Following~\cite{MAC2020ZOa}, we then find the distribution $p$ whose mapping under the readout transition matrices most closely matches $\hat{p}$ in Euclidean norm:
\[
\mathop{\mbox{minimize}}_{p\in\mathbb{R}^4}\quad\half\Vert (A_i\otimes A_{i+1})p - \hat{p}\Vert_2^2\quad\mbox{subject to}\quad p\geq 0,\ \Vert p\Vert_1 = 1.
\]
Combined, this gives a $4\,{\times}\, 4$ stochastic matrix for each of the {\sc{cnot}} gates. The total error probability per column of these matrices is listed in Figure~\ref{Fig:Ithaca}(b).

\begin{figure}
\centering
\begin{tabular}{c|c}
\raisebox{-0.5\height}{\includegraphics[width=130pt]{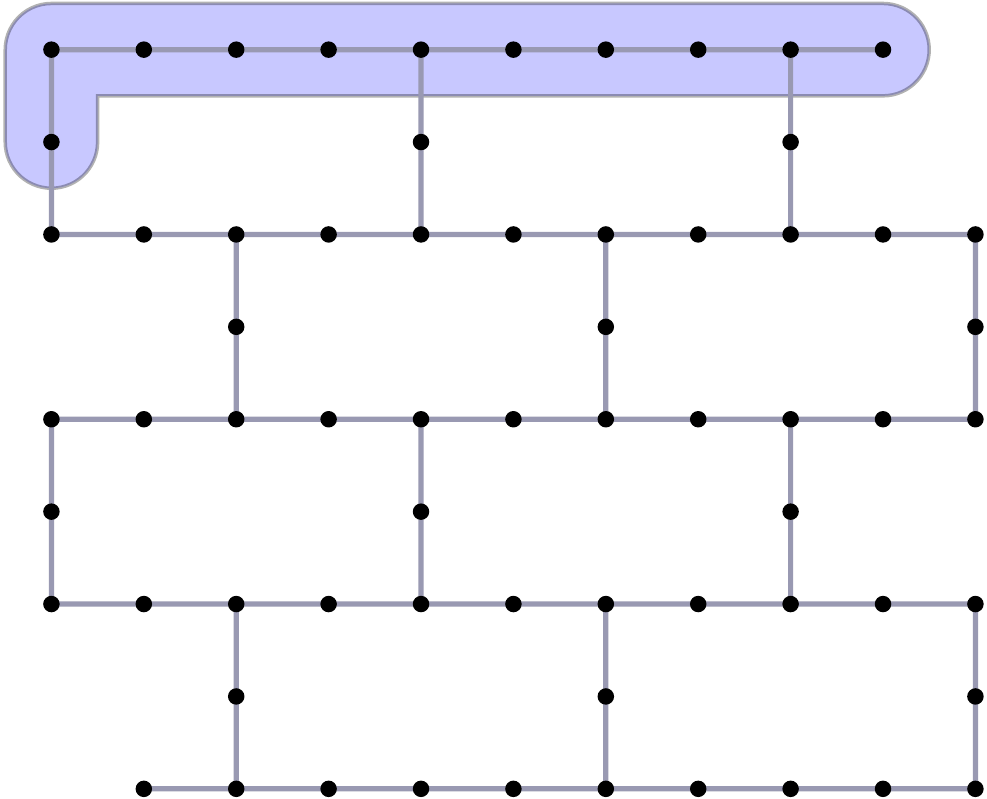}}\ \ \ &
\ \ \ {\small\setlength{\tabcolsep}{3.5pt}\begin{tabular}{llcrrrrrrrrrrr}
\hline
\multicolumn{2}{l}{Qubit} &&  10&0&1&2&3&4&5&6&7&8&9\\
\hline
Measure & \ket{0} &&
3.08&1.17&1.73&1.19&0.74&2.50&2.38&1.22&0.98&2.08&3.13\\
individual &\ket{1} &&
2.30&1.38&1.70&1.09&0.62&2.43&2.32&1.32&1.11&2.34&2.13\\
\hline
Measure & \ket{0} && 2.57&1.16&1.73&1.11&0.72&2.46&2.49&1.41&1.11&2.76&2.91\\
simultaneous & \ket{1} && 2.52&1.17&1.79&1.07&0.69&2.47&2.49&1.38&1.15&2.25&3.04\\
\hline
{\sc{cnot}} & \ket{00} &&
1.22&1.30&0.44&1.00&0.45&2.36&0.87&0.96&0.56&1.04\\
& \ket{01} &&
1.23&1.24&0.73&0.88&0.78&2.61&1.28&1.05&2.52&1.06\\
& \ket{10} &&
1.29&1.19&0.47&0.96&1.12&2.32&1.60&0.97&1.09&1.53\\
& \ket{11} &&
1.16&1.30&0.38&1.38&0.95&2.45&1.43&0.97&1.72&1.87\\
\hline
\end{tabular}}\\
\multicolumn{2}{c}{}\\[-5pt]
\multicolumn{1}{c}{({\bf{a}})} &
\multicolumn{1}{c}{({\bf{b}})}
\end{tabular}
\caption{The (a) processor topology of \ibmithaca\ along with the selected chain of 11 qubits and (b) estimated error rates of the different operations on the given qubits (in percent). Measurement errors are determined either by measuring individual qubits, or by measuring all qubits simultaneously. The {\sc{cnot}} operations use the given qubit as control and the qubit in the column that follows as target.}\label{Fig:Ithaca}
\end{figure}

\begin{figure}
\centering
\begin{tabular}{ccc}
\includegraphics[width=0.3\textwidth]{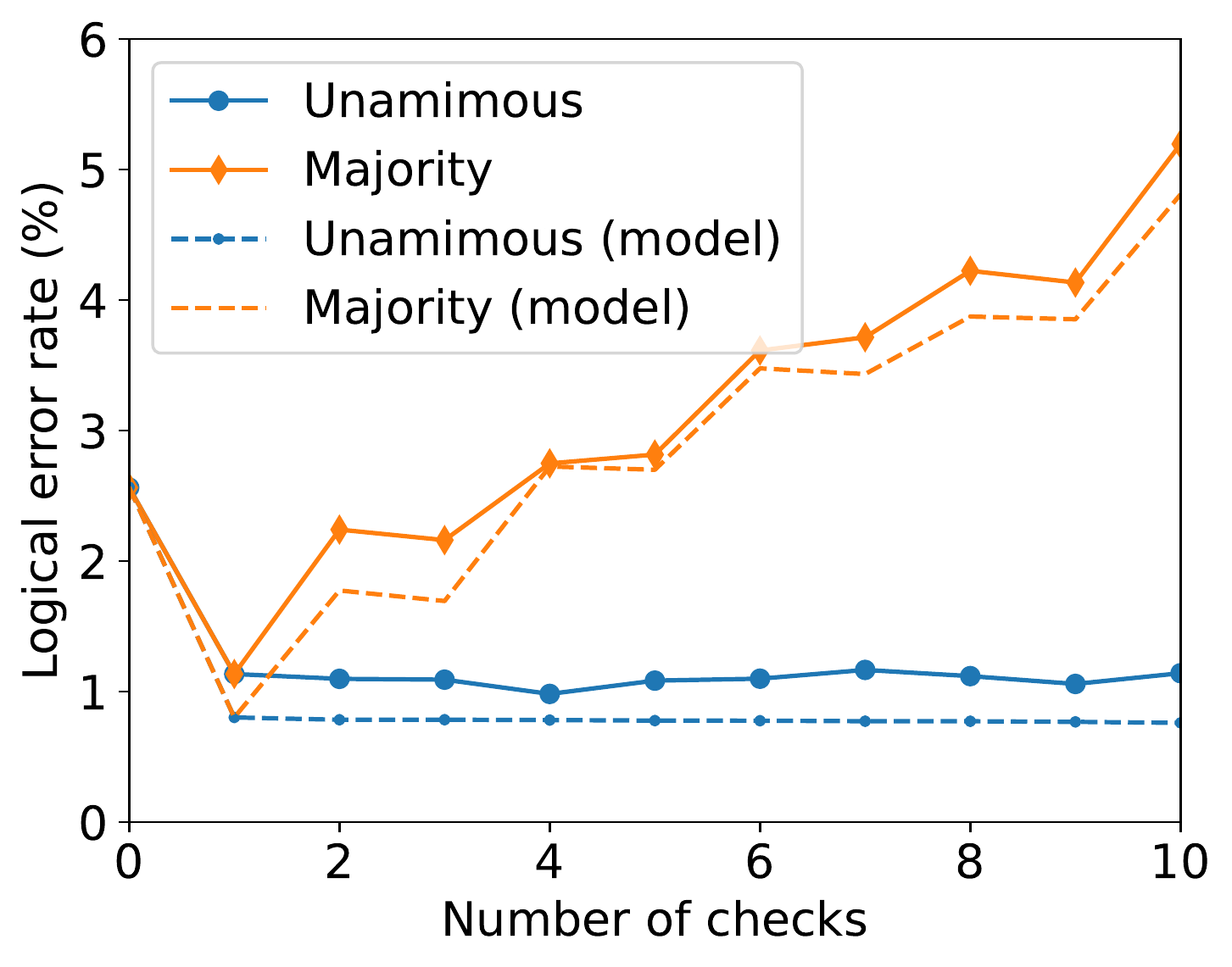}&
\includegraphics[width=0.3\textwidth]{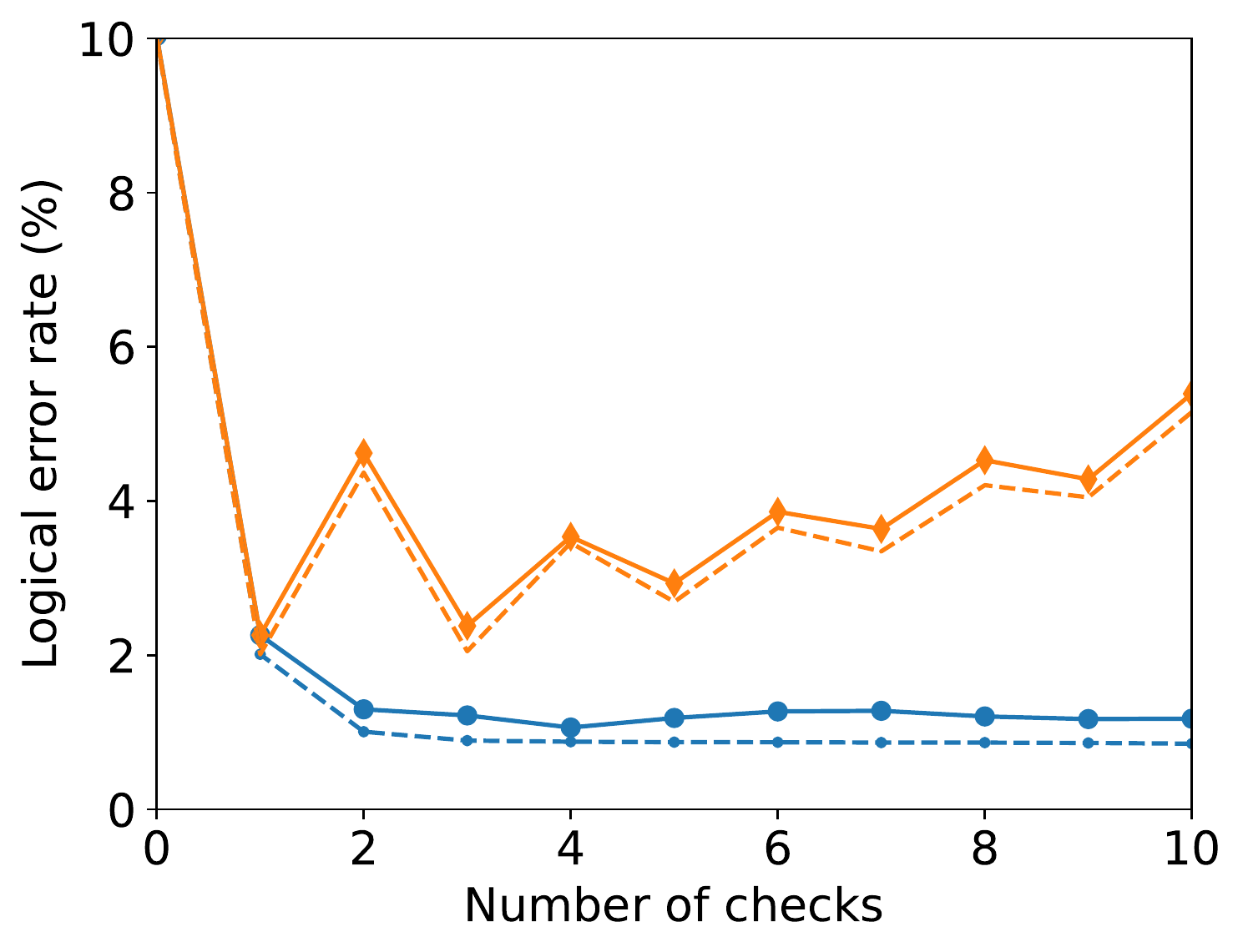}&
\includegraphics[width=0.3\textwidth]{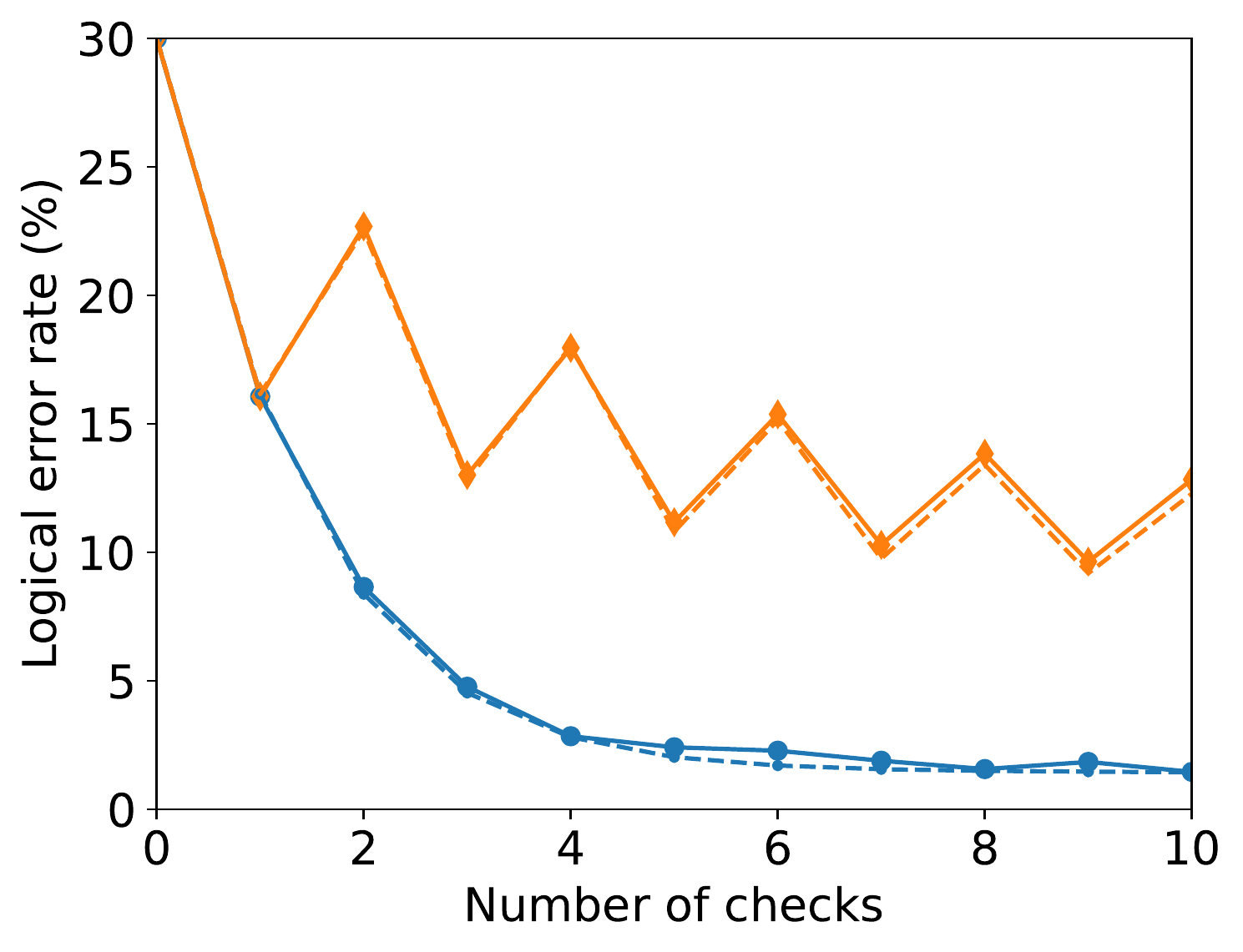}\\[3pt]
\includegraphics[width=0.3\textwidth]{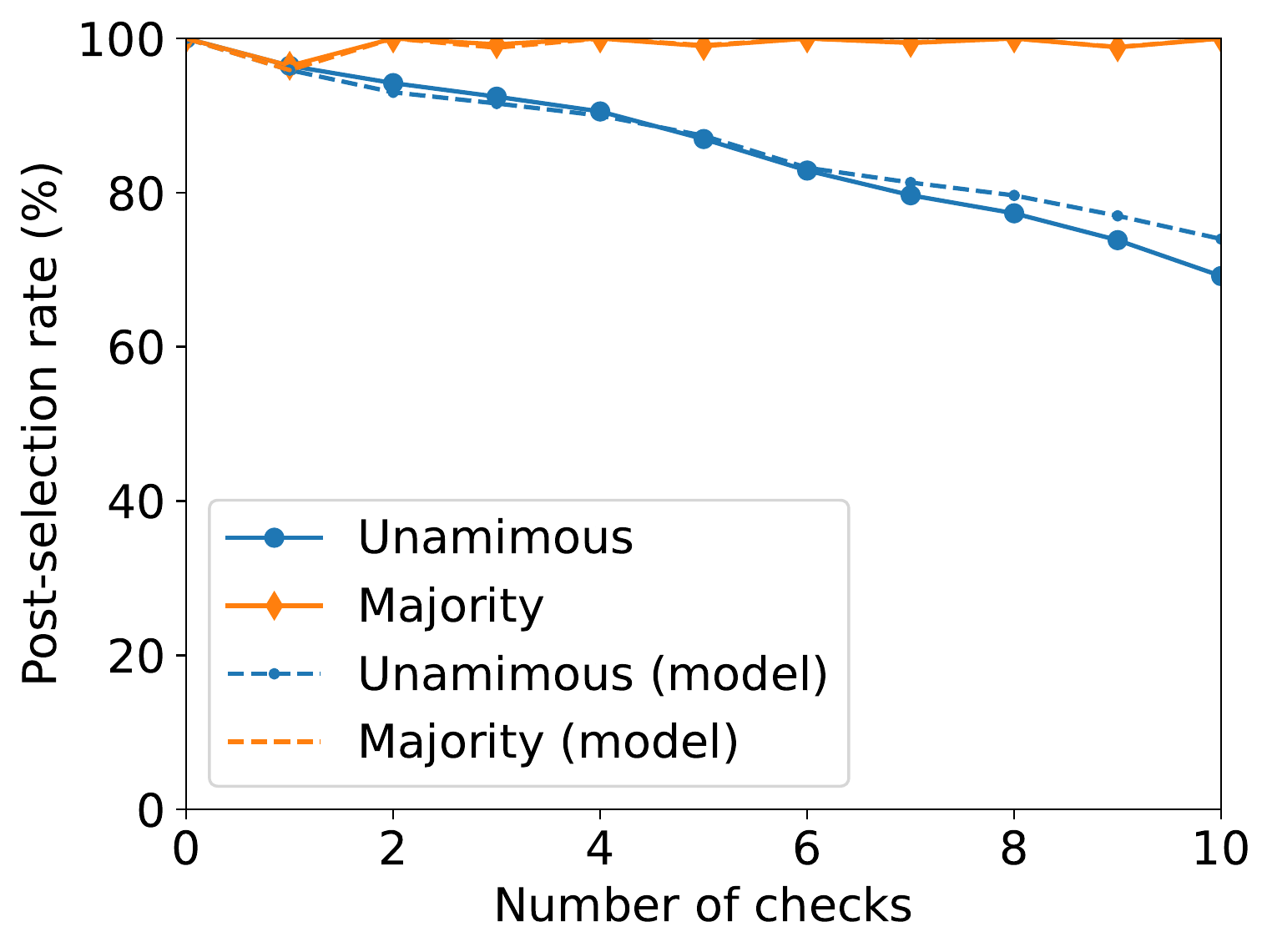}&
\includegraphics[width=0.3\textwidth]{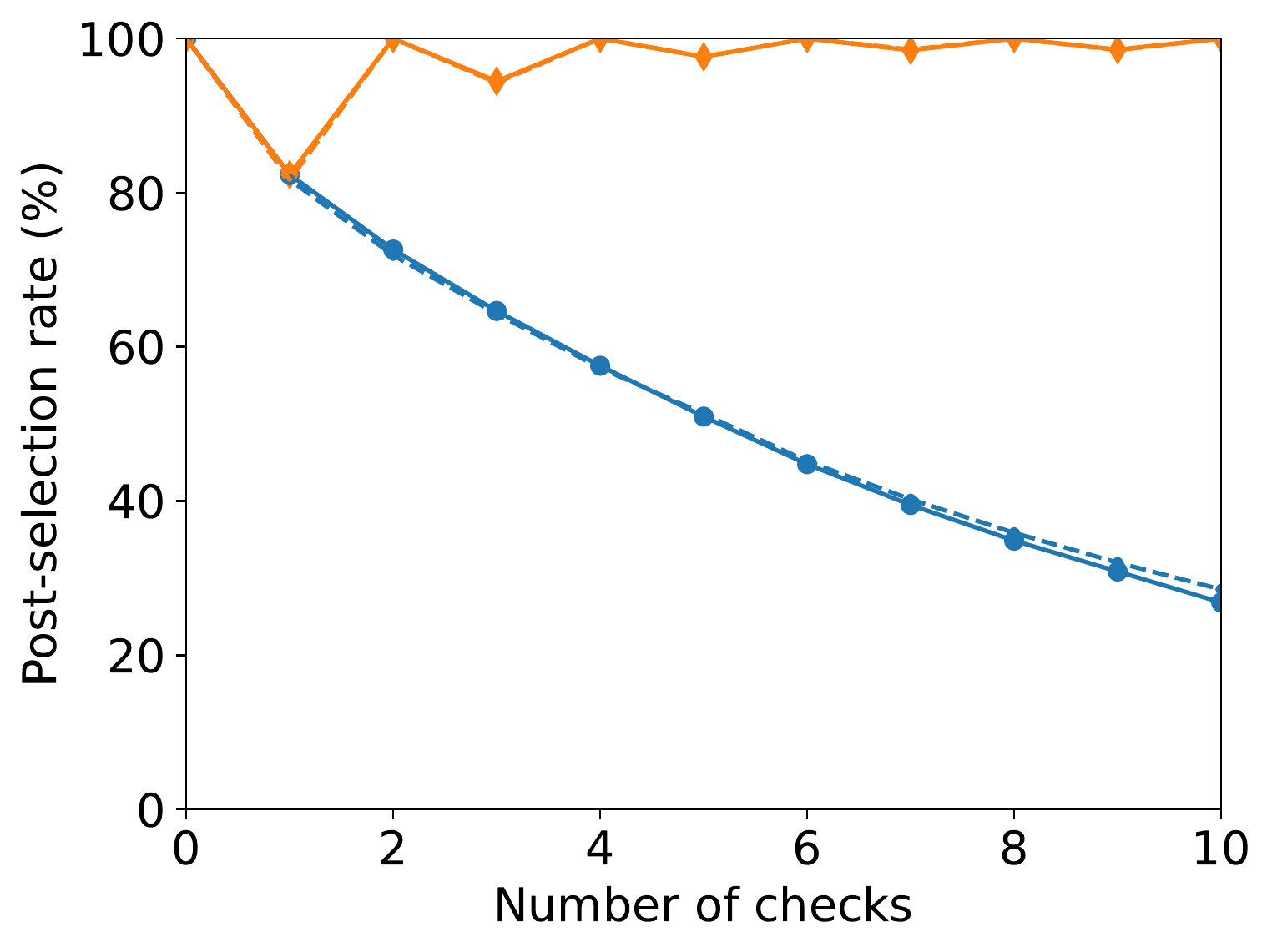}&
\includegraphics[width=0.3\textwidth]{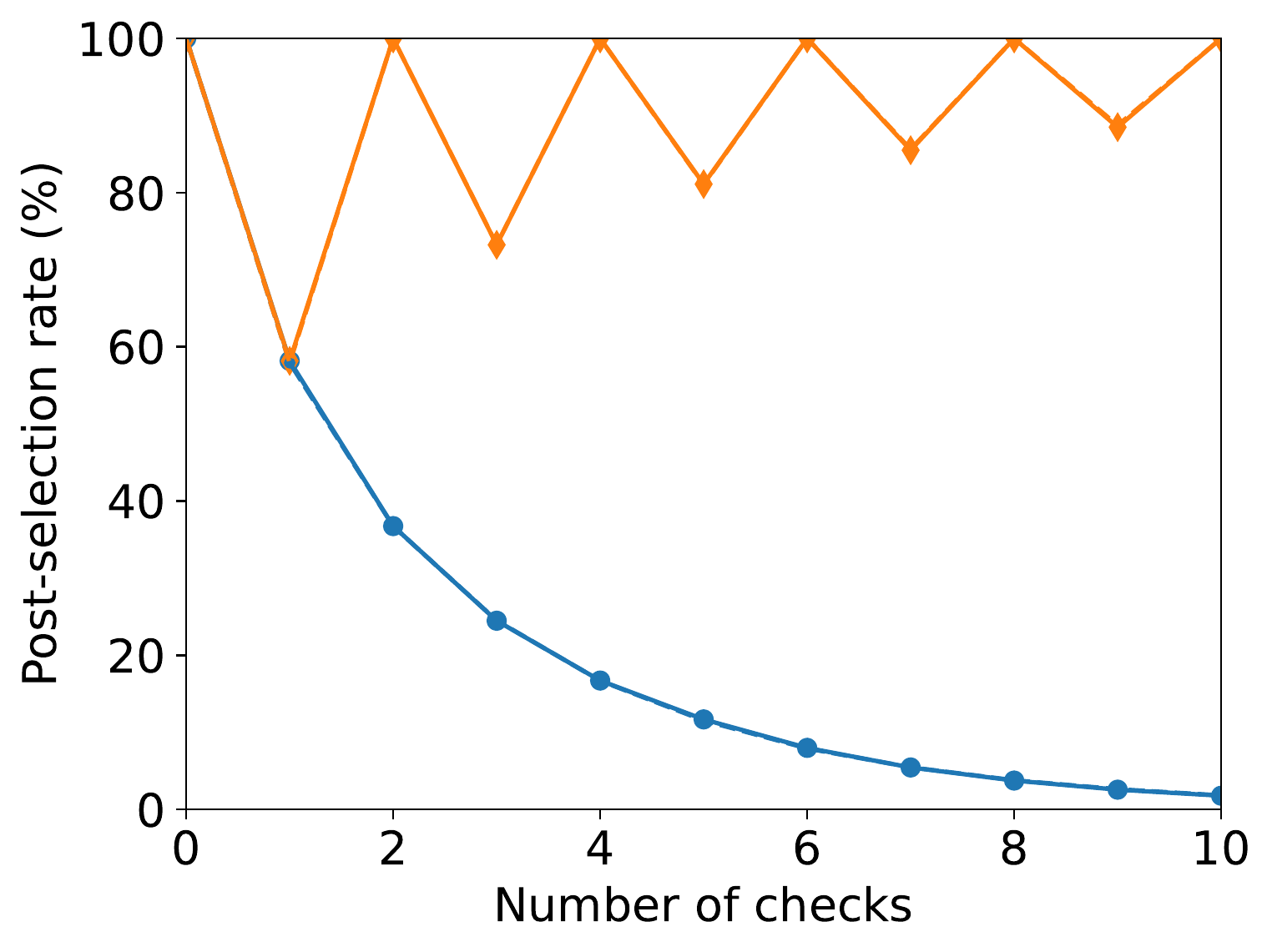}\\
({\bf{a}}) & ({\bf{b}}) & ({\bf{c}})
\end{tabular}
\caption{The (top) logical readout-error rates and (bottom) post-selection rates as a function of number of checks using all-equal and majority decoding. The results in (a) are based directly on experimental data, and modeled results using the estimated {\sc{cnot}} and measurement error rates. Plots (b) and (c) show the results obtained when classically boosting the experimental readout noise per qubit to approximately 10\% and 30\%, respectively.}\label{Fig:ReadoutMitigationExp}
\end{figure}

With all ingredients ready we can now measure the target qubit (the first qubit in the chain) with up to 10 checks. We consider three measurement settings. In the first, we measure a qubit directly following the last gate that applies to the qubit. In the second setting we apply a barrier following the gates on the qubits and measure all qubits simultaneously. The third setting matches the second, but applies dynamical decoupling on the idle time between the gates and the measurements that arises as a result of the {\sc{cnot}} ladder. A preliminary inspection of the results shows that there is little difference between the various settings and we therefore focus on the second setting.
For decoding of the measured bit strings, we use the unanimous approach in which all bits must match, as well as majority decoding in which the overall value is set to the largest number of 0 or 1 bits. In case of a tie, which possible only with an even number of qubits and therefore an odd number of checks, we reject the measurement. Figure~\ref{Fig:ReadoutMitigationExp}(a) plots the results obtained by 128 randomly twirled circuit instances with the target qubit initialized in the \ket{0} state, each sampled 1024 times. (Results obtained with the target state initialized to \ket{1} are similar and therefore omitted.)
For 0 or 1 checks the results of the two decoding schemes are identical. Beyond that, we see that majority decoding not only has a larger logical error rate, but also one that can exceed the rate without any checks. For comparison, we also plot the results obtained by simulating the respective probability distributions using the stochastic transition matrices for the {\sc{cnot}} and measurement operations.

We expect that in the future the fidelity of {\sc{cnot}} gates may be much higher than that of measurements. Therefore we like to study this disparity in accuracy. Although we cannot directly lower the {\sc{cnot}} error rate, what we can do is amplify the readout noise. Given a single qubit with readout-error rate $p$ and a target rate $r \geq p$, we can apply a classical bit-flip channel with probability $q$ such that the overall bit-flip probability $p(1-q) + q(1-p) = r$, which gives $q = (r-p) / (1-2p)$. For each qubit, we can compute the corresponding $q_i$ and apply the noise channel to each measured shot. The results obtained this way for readout errors amplified to 10\% and 30\% are shown in Figure~\ref{Fig:ReadoutMitigationExp}(b,c). As expected, with  unanimous decoding, the logical error rate decreases to levels far below the readout-error rate. Although majority decoding manages to lower the logical error rate, it does not quite attain the error rates obtained using unanimous decoding. Finally, in Figure~\ref{Fig:ScaledCX}, we plot the results obtained from modeled probability distributions when scaling the error terms in the {\sc{cnot}} transition matrices by a factor $\alpha$ and renormalizing the columns, while leaving the readout errors at their original levels.
At $\alpha {=} 1$ we obtain the original experimental results, while for $\alpha {=} 0$ we obtain the projected results when all {\sc{cnot}} gates are noiseless. As $\alpha$ decreases, so do the logical error rates; the post-selection rates, on the other hand, increase.

\begin{figure}[!h]
\centering
\begin{tabular}{cc}
\includegraphics[width=0.4\textwidth]{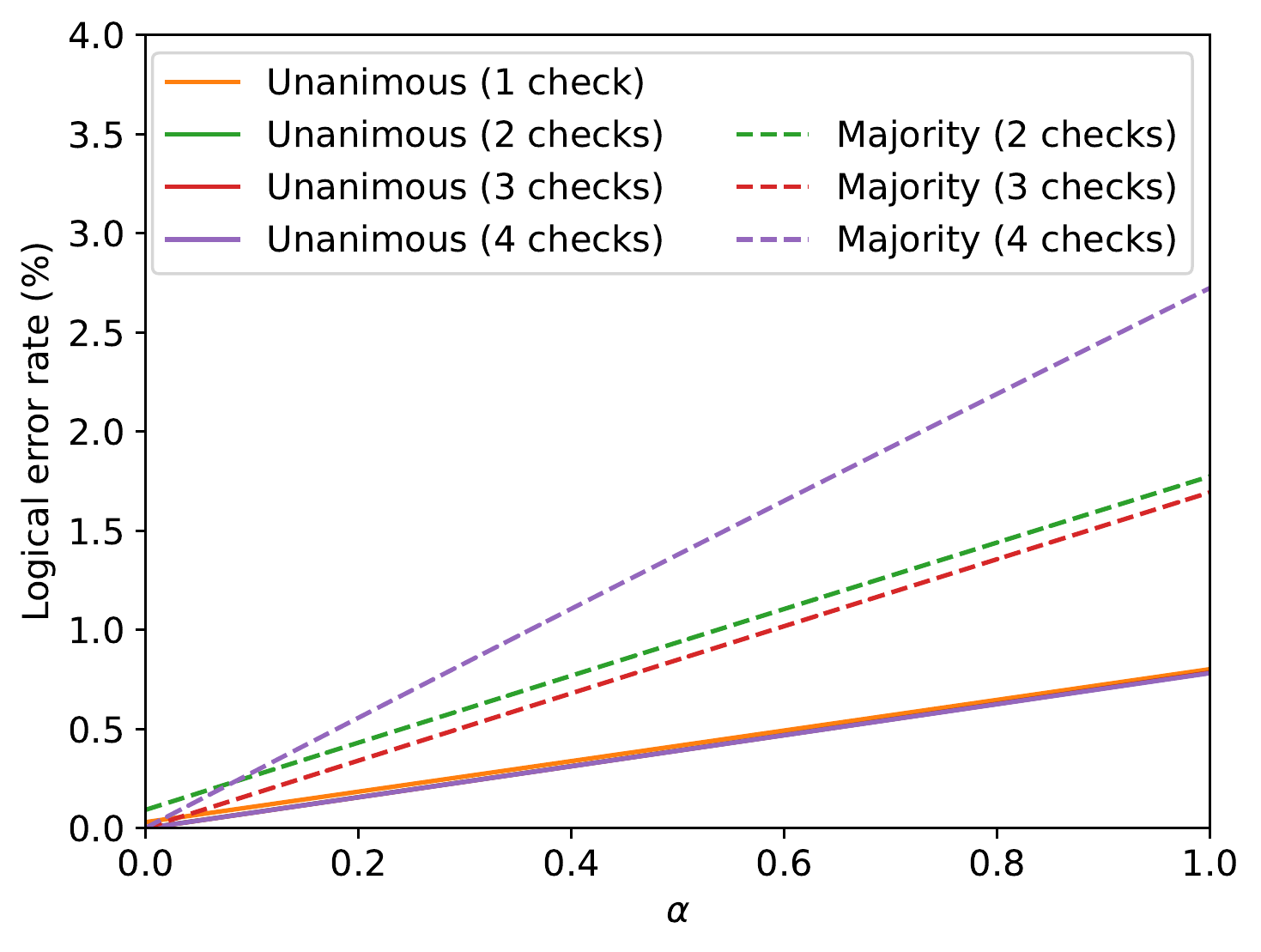}&
\includegraphics[width=0.4\textwidth]{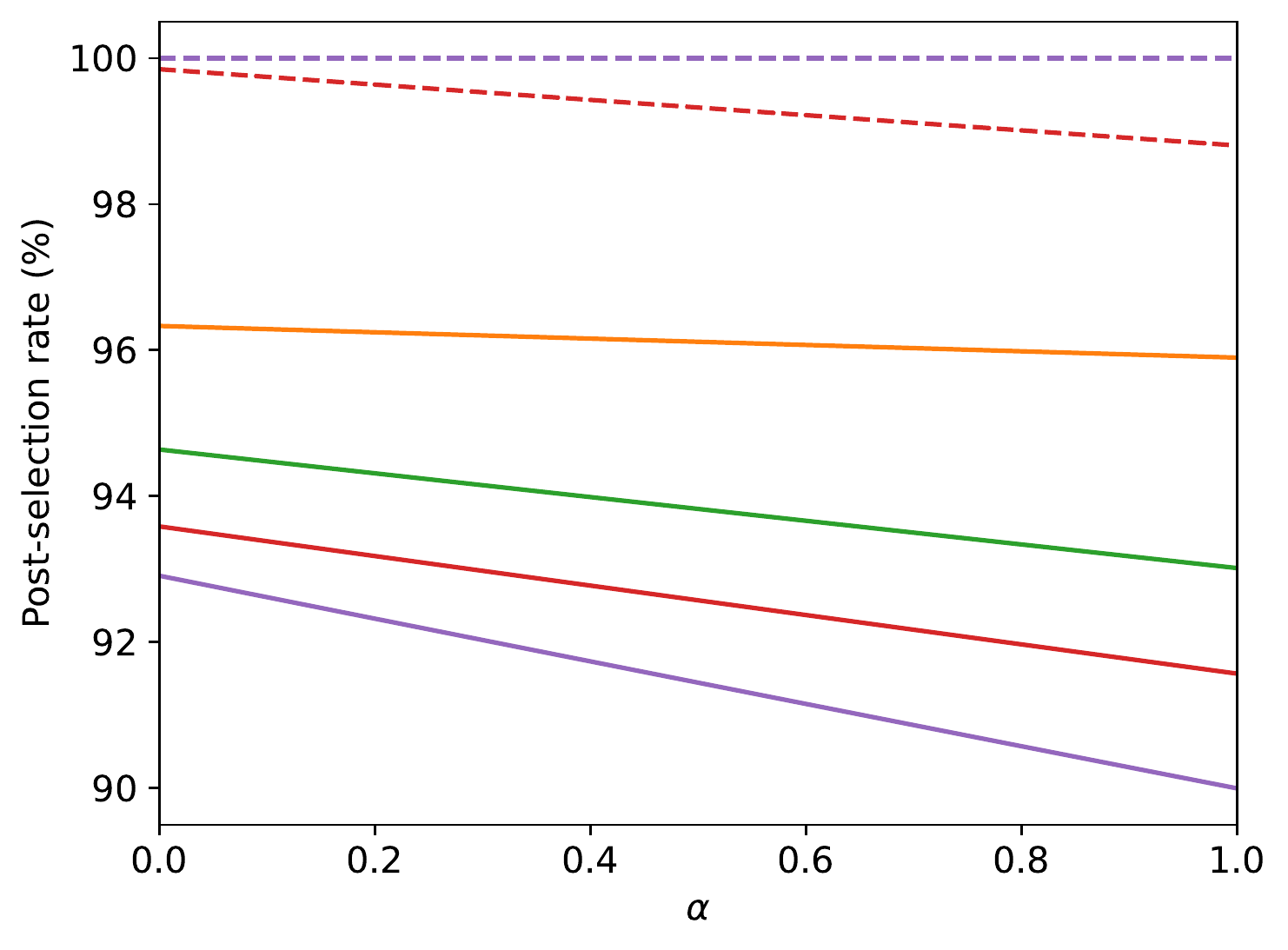}\\
({\bf{a}}) & ({\bf{b}})
\end{tabular}
\caption{Modeled (a) logical error rates and (b) post-selection rates for different scaling parameters $\alpha$ for the {\sc{cnot}} noise.}\label{Fig:ScaledCX}
\end{figure}

\section{Conclusion}
Our work has demonstrated, both analytically and experimentally,
that error mitigation methods can be successfully applied to quantum circuits with a  single-shot readout. It was shown that 
coherent Pauli checks provide a partially fault-tolerant implementation of Clifford circuits
with a small qubit and gate overhead. A large sampling overhead associated with
the postselection can be potentially avoided  if the measured syndromes
are used to correct errors in the payload circuit, rather than only detect them.
We leave the development of such error correction protocols for future work.  
While our main goal was error mitigation for Clifford payload circuits, we note that the same methods
can be applied to any layer of Clifford gates embedded into a larger, possibly non-Clifford, circuit.
As a concrete example, our methods can be straightforwardly applied to conjugated
Clifford circuits proposed by Bouland et al.~\cite{bouland2017complexity}.
Such circuits have a form $LUL^{-1}$, where $U$ is a random $n$-qubit Clifford operator
and $L$ is a layer of single-qubit SU$(2)$ gates. As shown in~\cite{bouland2017complexity},
sampling the output distribution of conjugated Clifford circuits is classically hard, under
plausible complexity-theoretic assumptions. Since single-qubit gates can be implemented
with high fidelity on almost all available quantum processors, the dominant source of errors
is likely to be the Clifford layer $U$ and our methods can be applied to mitigate these errors
in the single-shot setting. An interesting direction for future work is generalizing coherent Pauli
checks to  measurement-based models of quantum computation~\cite{raussendorf2001one}.
For example, the Pauli Based Computation introduced in~\cite{bravyi2016trading}
can efficiently simulate a universal quantum computer by initializing a register
of $n$ qubits in the tensor product of single-qubit magic states and performing
a suitable sequence of Clifford gates and measurements. In this example the dominant
source of errors is likely to be the Clifford part of the computation since the initial magic
states can be prepared with high fidelity by the single-qubit gates. 
Finally, our theoretical Markov chain model and the numerical simulation results indicate that near-term quantum processors with error rates in the range 
$0.1\%$ will be capable of sampling the output distribution of 
near-Clifford circuits  on approximately $50$ qubits.

\section*{Acknowledgements}

We would like to thank Shelly Garion and Ben Zindorf for help with compilation of Clifford circuits on linear qubit connectivity. This material is based upon work supported by the U.S. Department of Energy, Office of Science, National Quantum Information Science Research Centers, Quantum Science Center.

\end{document}